\documentclass{siamart171218}
\usepackage{ifpdf}
\usepackage{amsmath}
\usepackage{mathrsfs}
\usepackage{lipsum}
\usepackage{amsfonts}
\usepackage{graphicx}
\usepackage{epstopdf}
\usepackage{algorithmic}
\usepackage{enumerate}
\usepackage{subfigure}
\usepackage{amsopn}
\DeclareMathOperator{\diag}{diag}
\DeclareMathOperator{\sgn}{sgn}
\ifpdf
  \DeclareGraphicsExtensions{.eps,.pdf,.png,.jpg}
\else
  \DeclareGraphicsExtensions{.eps}
\fi

\newsiamremark{remark}{Remark}
\newsiamremark{hypothesis}{Hypothesis}
\crefname{hypothesis}{Hypothesis}{Hypotheses}
\newsiamthm{claim}{Claim}
\newsiamremark{assumption}{Assumption}
\newsiamremark{example}{Example}

\long\def\comment#1{}

\graphicspath{{figures/}}

\headers{Sub/super-stochastic matrix and bipartite tracking}{Lei Shi, Wei Xing Zheng, Jinliang Shao, and Yuhua Cheng}

\title{Sub/super-stochastic matrix with applications to bipartite tracking control over signed networks
 \thanks{Submitted to SIAM Journal on Control and Optimization
}}

\author{Lei Shi\thanks{School of Automation Engineering, University of Electronic Science and Technology of China, Chengdu, Sichuan, 611731, China
  (\email{shilei@std.uestc.edu.cn}, \email{jinliangshao@uestc.edu.cn}, \email{yhcheng@uestc.edu.cn}).}
\and Wei Xing Zheng\thanks{School of Computing, Engineering and Mathematics, Western Sydney University, Sydney, NSW 2751, Australia
  (\email{w.zheng@westernsydney.edu.au}).}
\and Jinliang Shao\footnotemark[2]
\and Yuhua Cheng\footnotemark[2]}

\begin{document}

\maketitle
\begin{abstract}
In this contribution, the properties of sub-stochastic matrix and super-stochastic matrix are applied to analyze the bipartite tracking issues of multi-agent systems (MASs) over signed networks, in which the edges with positive weight and negative weight are used to describe the cooperation and competition among the agents, respectively. For the sake of integrity of the study, the overall content is divided into two parts. In the first part, we examine the dynamics of bipartite tracking for first-order MASs, second-order MASs and general linear MASs in the presence of asynchronous interactions, respectively. Asynchronous interactions mean that each agent only interacts with its neighbors at the instants when it wants to update the state rather than keeping compulsory consistent with other agents. In the second part, we investigate the problems of bipartite tracing in different practical scenarios, such as time delays, switching topologies, random networks, lossy links, matrix disturbance, external noise disturbance, and a leader of unmeasurable velocity and acceleration. The bipartite tracking problems of MASs under these different scenario settings can be equivalently converted into the product convergence problems of infinite sub-stochastic matrices (ISubSM) or infinite super-stochastic matrices (ISupSM). With the help of nonnegative matrix theory together with some key results related to the compositions of directed edge sets, we establish systematic algebraic-graphical methods of dealing with the product convergence of ISubSM and ISupSM. Finally, the efficiency of the proposed methods is verified by computer simulations.
\end{abstract}



\section{Introduction}\label{section:1}

Consensus of multi-agent systems (MASs) has caused intense interest among researchers for quite a long period of time, thanks to its widespread applications in areas such as physics and computer science. Consensus refers to that a set of agents can implement a completely consistent state, such as position or opinion, by developing a distributed control protocol. In the existing literature focusing on consensus of MASs \cite{Cao2008Reaching,Li2010Consensus,Huang2009Coordination,Gao2011Sampled,Yu2011Second-order,Chen2013Consensus,Qin2013On,
Wen2014Consensus,Alvergue2016const,Liu2018net,Yan2018new,Yu2018dist}, it is usually assumed that the agents cooperate with each other, that is, the weight of each edge in the communication topology is positive. Nevertheless, in many real-world network systems including social networks \cite{Tao2006Interest,Green2006Children}, bus transport networks \cite{Hu2016Bus} and neural networks \cite{Mao2007Dynamics,Chandra2015Competition}, the coexistence of cooperation and competition among agents is an indisputable fact. Tao \emph{et al.}~\cite{Tao2006Interest} believed that agents can only achieve mutual benefit and win-win in the business environment by working together, but on the other hand, some selfish agents may compete with other agents for their own benefit. Green \emph{et al.}~\cite{Green2006Children} analyzed that the ability to balance cooperative and competitive behavior is of great importance to the overall development of kids. Hu \emph{et al.}~\cite{Hu2016Bus} found that cooperation and competition are common in bus transport network systems due to the existence of shared sites between different bus lines. It was demonstrated in \cite{Mao2007Dynamics} that competition and collaboration are critical to the survival of different species in neural networks with limited resources. Chandra \cite{Chandra2015Competition} explained how the lateral inhibition leads to competition between neurons in neural networks.

In recent years, consensus behavior has been expanded to the general MASs running on the structurally balanced network with both cooperation and competition, in which the agents are usually divided into two completely competitive subgroups and the agents in the same group are cooperative. As time goes by, the two agent subgroups will gradually reach two different states with the same modulus but opposite signs. Such a consensus phenomenon over the structurally balanced network is known as bipartite consensus of MASs, which is also called the opinion polarization in social networks. This phenomenon of bipartite consensus is common in many systems that describe the confrontation of bimodal alliances, such as bipartisan political systems, teams opposed in sports competitions, competitive commercial cartels, duopoly markets, competitive international alliances. For example, Cartwright \emph{et al.} \cite{Cartwright1975Gang} found that in actual community conflicts, people with similar opinions will gradually unite, so that their common tendencies will be strengthened, and people with different opinions will have competitive relationships and suspicions. Finally, people with similar views will gradually form a group with the same opinion, and the final opinions of the agents with different opinions will be completely opposite. Recently, bipartite consensus of MASs has attracted wide attention. Meng et al.~\cite{Meng2016Interval} studied first-order bipartite consensus based on the rooted cycles of digraphs. A bipartite consensus issue in the second-order dynamic model with finite-time setting was analyzed in \cite{Zhao2017Adaptive}. The bipartite consensus with arbitrary finite communication delays was investigated in \cite{Guo2018Bipartite}. In \cite{Zhang2016Bipartite} a bipartite consensus issue for generic linear agents was addressed by designing a control protocol with state feedback, while in \cite{Qin2016On} the bipartite consensus issue was solved for a set of agents with general linear dynamics under input saturation.

In fact, the leader plays a vital role in many engineering applications such as unmanned aerial vehicle formation flight, multiple mobile robots, and more. In addition, a stubborn agent (leader) who insists on its own opinion has a significant impact on the opinions formation of the remaining agents in most of social networks. So far, a lot of valuable results focusing on the cooperative consensus with a leader, which is often known as tracking consensus, have been announced \cite{Ding2013Networked,Cheng2016On,Shao2016A,He2017Leader,Shao2018On}, in which there is only mutual cooperation between agents. As we introduced above, it is more practical to study the tracking issue under signed networks, in which cooperation and competition are coexistent. Recently, the bipartite tracking phenomenon under signed networks has generated a strong interest among researchers. In such an dynamics phenomenon, the followers in the same subgroup with the leader will be eventually in exactly the same state as the leader, and the followers in different subgroups will reach the opposite state of the leader. Up to now, some achievements have been made in the bipartite tracking issue of MASs. In \cite{Ma2018Necessary}, Ma \emph{et al.} discussed the bipartite tracking issue of first-order MASs with measurement noise. Wen \emph{et al.}~\cite{Wen2018Bipartite} explored the bipartite tracking issue for generic linear agents by designing non-smooth protocols. The high-order bipartite tracking issue with nonidentical matching uncertainties was handled in \cite{Liu2018Robust}, and the nonlinear bipartite tracking under the setting of hysteresis input uncertainties was analyzed in \cite{Yu2018Prescribed}. One study by Wu \emph{et al.}  \cite{Wu2019Bipartite} considered bipartite tracking of MASs with Lipschitz-Type nonlinear dynamics. The fixed-time bipartite tracking for fractional-order MASs was investigated in \cite{Gong2019Fixed}.

Remarkably, the above-mentioned studies on the bipartite tracking issue were carried out under the synchronous setting, in which all agents' clocks are consistent in the sense that each agent interacts with its neighbors at any time. In actual engineering applications, the clocks on different processors are generally different and independent of each other. Further, if each agent is equipped with a time-driven sensor, then different clocks can have a serious impact on system performance and may even lead to system instability. Accordingly, it is of great interest and practical importance to study distributed control issues of MASs under the asynchronous setting. In recent years, some consensus issues of MASs with asynchronous settings have gained considerable attention in the literatures \cite{Xiao2008asy,Qin2012Stationary,Zhan2015Asynchronous,Shi2019Containment}, where there is only mutual cooperation between the agents. However, there have been almost no reports on the asynchronous bipartite tracking issue of MASs over signed networks so far.
In addition, the existing literature of bipartite tracking has rarely considered other unavoidable scenarios in the actual system, such as time delays, switching topologies, random interaction networks, lossy links, matrix disturbance, external noise disturbance, and a leader of unmeasurable velocity and acceleration. These actual scenario settings are usually caused by congestion and random fading of communication channels, physical characteristics of information transmitted by media, communication range limitations, channel noise and measurement errors, etc. Undoubtedly, it is very interesting and practical to analyze the impact of these actual scenarios on the bipartite tracking dynamics of MASs on signed networks.

Motivated by the above discussions, the main purpose of this paper is to apply the properties of sub-stochastic matrix and super-stochastic matrix to the analysis of bipartite tracking dynamics of MASs on signed networks under different practical scenario settings, such as asynchronous interactions, time delays, switching topologies, random networks, lossy links, matrix disturbance, external noise disturbance, and a leader of unmeasurable velocity and acceleration. The main contributions of this paper are shown below.
\begin{enumerate}
\item [I)] To date, the product convergence of infinite row-stochastic matrices in which all row sums are equal to 1 is a widely used method of solving the stability of discrete-time linear systems, such as the consensus issues of discrete-time MASs \cite{Xiao2008asy,Qin2012Stationary,Zhan2015Asynchronous}. However, it is worth noting that this existing method can only be used to analyze certain discrete-time linear systems with the coefficient matrices being row-stochastic, and it is invalid for most discrete-time linear systems in which the coefficient matrices may contain some row sums being less than or greater than 1. In view of this, this paper proposes the concept of super-stochastic matrix for the first time, in which the row sums are allowed to be less than 1, equal to 1 or greater than 1. In addition, some conclusions about super-stochastic matrix and sub-stochastic matrix (in which all row sums are less than or equal to 1) are established by using algebraic-graphical methods, including the upper bound of spectral radius and product properties. These conclusions together with some key results related to the compositions of directed edge sets can be used to solve the product convergence issues of infinite sub-stochastic matrices (ISubSM) or infinite super-stochastic matrices (ISupSM).

\item [II)] In Section~3-Section~6, we discuss, respectively, the issues of bipartite tracking for first-order MASs, second-order MASs and general linear MASs under the asynchronous setting, where each agent only interacts with its neighbors at the instants when it wants to update the state rather than keeping compulsory consistent with other agents. In the asynchronous setting, it is uncertain which agents will communicate with their neighbors at each discrete-time instant. And thus, it is required to analyze the stability of time-varying discrete-time systems. Our main analysis strategy is to design suitable augmented vectors to transform the stability of time-varying discrete-time systems to the product convergence issues of ISubSM or ISupSM.

\item [III)] In Section~7-Section~13, we analyze the effects of different actual scenario settings on the dynamics of second-order MASs on signed networks. By using the product properties of ISubSM and ISupSM, we establish algebraic conditions for realizing bipartite tracking under some scenario settings, including time delays, switching topologies, random interaction networks, lossy links, and matrix disturbance. In addition, in the case with external noise and the case with a leader of unmeasurable velocity and acceleration, although the system may not appear bipartite tracking phenomenon, the bipartite bounded tracking phenomenon, first proposed in this paper, can be realized. In the results of bipartite bounded tracking, we also give the upper bound of the errors between the followers' final states and the leader's state or sign opposite state.
\end{enumerate}

Except for this section, the rest of this article is organized as follows. Some basic notations, preliminaries, and conclusions on the sub-stochastic (super-stochastic) matrix, the signed digraph, the composition of directed edge sets and the asynchronous setting are introduced in Section~\ref{section:2}. Section~\ref{section:3}--\ref{section:6} undertake separately the stability analysis for the asynchronous bipartite tracking of first-order MASs, second-order MASs with a static leader, second-order MASs with an active leader, and general linear MASs. In Section~\ref{section:7}--\ref{section:11}, the issues of bipartite tracking of second-order MASs under some actual scenario settings, such as time delays, switching topologies, random interaction networks, lossy links, and matrix disturbance. Section~\ref{section:12} and Section \ref{section:13} analyze the phenomenon of bipartite bounded tracking for the case with external noise and the case with a leader of unmeasurable velocity and acceleration. Finally, some conclusions of this article are given in Section~\ref{section:14}.

\section{Preliminaries}\label{section:2}

\subsection{Matrix notations}

For a real matrix $F=[f_{ij}]\in\mathbb{R}^{n\times n}$, let $|F|=[|f_{ij}|]$ be an $n\times n$ matrix in which $|f_{ij}|$ is the absolute value of $f_{ij}$; $\rho(F)$ denotes the spectral radius of $F$; $\Lambda_i[F]=\sum_{j=1}^{n}f_{ij}$ stands for the $i$th row's sum of matrix $F$; and $\|F\|_{\infty}=\max\big\{\sum_{j=1}^{n}|f_{ij}| \mid i=1,2,\ldots,n\big\}$ is the infinite norm of matrix $F$. Moreover, $[F]_{ij}$ can also be used to denote the element $f_{ij}$ of matrix $F$ if there is no ambiguity. A real matrix $F$ is nonnegative if $f_{ij}\geq0$, $i,j=1,2,\ldots,n$, and $F$ is non-positive if $f_{ij}\leq0$, $i,j=1,2,\ldots,n$. $I_{n}$ is an $n$-order identity matrix and $\mathbf{0}\in\mathbb{R}^{n\times m}$ is a matrix whose all elements are equal to 0. The sets of positive integers and natural numbers are denoted by $\mathbb{Z}_{+}$ and $\mathbb{N}$, respectively. $\otimes$ stands for the Kronecker product. It is said that the real matrices $F=[f_{ij}]$ and $H=[h_{ij}]$ have the same type, denoted by $F\sim H$, if $f_{ij}>0\Leftrightarrow h_{ij}>0$, $f_{ij}<0\Leftrightarrow h_{ij}<0$ and $f_{ij}=0\Leftrightarrow h_{ij}=0$. The signum function $\sgn(x)$ is given by
\begin{equation*}
\sgn(x)=\left\{
\begin{array}{ll}
1, & x>0,\\
0, & x=0,\\
-1, & x<0.
\end{array}
\right.
\end{equation*}

\subsection{Some results about sub/super-stochastic matrix}\label{section:2.2}

\begin{definition}[see \cite{Pullman1966Infinite}]\label{definition:2.1}
A real matrix $F\in\mathbb{R}^{n\times n}$ is sub-stochastic if it is nonnegative and $\Lambda_i[F]\leq1$ for any $i=1,2,\ldots,n$.
\end{definition}

Some new conclusions about the upper bound of spectral radius and product for sub-stochastic matrices are given below.

\begin{theorem}\label{theorem:2.2}
For an $n\times n$ sub-stochastic matrix $F$, let $\mathcal{S}_1=\{s \mid \Lambda_s[F]<1\}$ and $\mathcal{S}_2=\{s \mid \Lambda_s[F]=1\}$ be non-empty. If there exist non-zero element chains $[F]_{ki_r},\ldots,[F]_{i_2i_1},[F]_{i_1i}$ (denoted by $\mathcal{C}_{i\rightarrow k}$) for each $k\in\mathcal{S}_2$, where $i\in\mathcal{S}_1$, $i_r\neq k,\ldots,i_1\neq i_2, i\neq i_1$, then one has
\begin{equation}\label{sys:2.1}
\begin{aligned}
\rho(F)\leq\sqrt[|\mathcal{C}^*|+1]{1-(1-\alpha_1)\alpha^{|\mathcal{C}^*|}_2}<1,
\end{aligned}
\end{equation}
where
\begin{equation*}
\begin{aligned}
&\alpha_1=\max\{\Lambda_s[F] \mid \Lambda_s[F]<1\},\\
&\alpha_2=\min\{[F]_{ij}>0 \mid i,j=1,2,\ldots,n\},\\
&|\mathcal{C}^*|=\max\{|\mathcal{C}_{i\rightarrow k}| \mid i\in\mathcal{S}_1, k\in\mathcal{S}_2\},
\end{aligned}
\end{equation*}
in which $|\mathcal{C}_{i\rightarrow k}|$ is the number of elements of $\mathcal{C}_{i\rightarrow k}$.
\end{theorem}

\begin{proof}
Since $i\in\mathcal{S}_1$, we have $\Lambda_i[F]\leq\alpha_1<1$, which together with the facts $[F]_{ki_r}\geq\alpha_2, \ldots, [F]_{i_2i_1}\geq\alpha_2, [F]_{i_1i}\geq\alpha_2$, guarantees that
\begin{equation*}
\begin{aligned}
&\Lambda_{i_1}[F^2]=\sum_{j=1, j\neq i}^n[F]_{i_1j}\Lambda_{j}[F]+[F]_{i_1i}\Lambda_{i}[F]\leq1-(1-\alpha_1)\alpha_2<1,\\
&\Lambda_{i_2}[F^3]=\sum_{j=1, j\neq i_1}^n[F]_{i_2j}\Lambda_{j}[F^2]+[F]_{i_2i_1}\Lambda_{i_1}[F^2]\leq1-(1-\alpha_1)\alpha^2_2<1,\\
& \ \ \vdots \\
&\Lambda_{k}[F^{r+2}]=\sum_{j=1, j\neq i_r}^n[F]_{kj}\Lambda_{j}[F^{r+1}]+[F]_{ki_r}\Lambda_{i_r}[F^{r+1}]\leq1-(1-\alpha_1)\alpha^{r+1}_2<1.
\end{aligned}
\end{equation*}
According to the definition of $|\mathcal{C}^*|$, we have $|\mathcal{C}^*|\geq r+1$. And thus,
\begin{equation*}
\begin{aligned}
\Lambda_{k}[F^{|\mathcal{C}^*|+1}]&=\sum_{j=1}^n[F^{r+2}]_{kj}\Lambda_{j}[F^{r+1}]+[F^{|\mathcal{C}^*|-r-1}]\\
&\leq\sum_{j=1}^n[F^{r+2}]_{kj}\leq1-(1-\alpha_1)\alpha^{|\mathcal{C}^*|}_2<1.
\end{aligned}
\end{equation*}
In addition, for any $i\in\mathcal{S}_1$, we have
\begin{equation*}
\begin{aligned}
\Lambda_{i}[F^{|\mathcal{C}^*|+1}]=\sum_{j=1}^n[F]_{ij}\Lambda_{j}[F^{|\mathcal{C}^*|}]\leq\sum_{j=1}^n[F]_{ij}
\leq1-(1-\alpha_1)\alpha^{|\mathcal{C}^*|}_2<1.
\end{aligned}
\end{equation*}
This implies that $\Lambda_{s}[F^{|\mathcal{C}^*|+1}]\leq1-(1-\alpha_1)\alpha^{|\mathcal{C}^*|}_2$ for any $s=1,2,\ldots,n$. As a result,
\begin{equation*}
\begin{aligned}
\big\|F^{|\mathcal{C}^*|+1}\big\|_{\infty}\leq1-(1-\alpha_1)\alpha^{|\mathcal{C}^*|}_2<1.
\end{aligned}
\end{equation*}
According to the known facts $\rho(F^{|\mathcal{C}^*|+1})=\rho^{|\mathcal{C}^*|+1}(F)$ and $\rho(F^{|\mathcal{C}^*|+1})\leq\|F^{|\mathcal{C}^*|+1}\|_{\infty}$, we know that the result (\ref{sys:2.1}) holds.
\end{proof}

\begin{theorem}\label{theorem:2.3}
Let $F$ be an $n\times n$ sub-stochastic matrix in which $[F]_{ii}=0$, $i=1,2,\ldots,n$. If the number of elements in the set $\mathcal{S}_2=\{s \mid \Lambda_s[F]=1\}$ is equal to 1, there holds that
\begin{equation}\label{sys:2.2}
\begin{aligned}
\rho(F)\leq\sqrt{\alpha_1}<1,
\end{aligned}
\end{equation}
where $\alpha_1$ is defined in Theorem \ref{theorem:2.2}.
\end{theorem}

\begin{proof}
Assume that $\lambda$ is an eigenvalue of $F$, and $x=[x_1,x_2,\ldots,x_n]^T$ is the corresponding eigenvector. Thus, we have $Fx=\lambda x$. Let $|x_p|=\max\{|x_i| \mid i=1,2,\ldots,n\}$ and $|x_q|=\max\{|x_i| \mid i=1,2,\ldots,n, \ i\neq p\}$. Clearly, $|x_p|\geq|x_q|\geq|x_i|$, $i\neq p,q$. According to $Fx=\lambda x$, we can derive $x_p(\lambda-[F]_{pp})=\sum_{j=1, j\neq p}^n[F]_{pj}x_j$. Since $[F]_{pp}=0$, one has $\lambda x_p=\sum_{j=1, j\neq p}^n[F]_{pj}x_j$. It follows that
\begin{equation*}
\begin{aligned}
|\lambda||x_p|=\Big|\sum_{j=1, j\neq p}^n[F]_{pj}x_j\Big|\leq\sum_{j=1, j\neq p}^n|[F]_{pj}||x_q|=\Lambda_{p}[F]|x_q|.
\end{aligned}
\end{equation*}
This means that
\begin{equation}\label{sys:2.3}
\begin{aligned}
|\lambda|\leq\Lambda_{p}[F]\frac{|x_q|}{|x_p|}.
\end{aligned}
\end{equation}
In addtion, we can also derive that $\lambda x_q=\sum_{j=1, j\neq p}^n[F]_{qj}x_j$ by $Fx=\lambda x$. This means
\begin{equation*}
\begin{aligned}
|\lambda||x_q|=\Big|\sum_{j=1, j\neq q}^n[F]_{qj}x_j\Big|\leq\sum_{j=1, j\neq q}^n|[F]_{qj}||x_p|=\Lambda_{q}[F]|x_p|,
\end{aligned}
\end{equation*}
and further
\begin{equation}\label{sys:2.4}
\begin{aligned}
|\lambda|\leq\Lambda_{q}[F]\frac{|x_p|}{|x_q|}.
\end{aligned}
\end{equation}
Multiply (\ref{sys:2.3}) and (\ref{sys:2.4}) to get
\begin{equation*}
\begin{aligned}
|\lambda|^2\leq\Lambda_{p}[F]\frac{|x_q|}{|x_p|}\Lambda_{q}[F]\frac{|x_p|}{|x_q|}=\Lambda_{p}[F]\Lambda_{q}[F]\leq\alpha_1.
\end{aligned}
\end{equation*}
Equivalently, $\rho^2(A)\leq\alpha_1$. That is, the result (\ref{sys:2.2}) holds.
\end{proof}

In Theorem \ref{theorem:2.2} and Theorem \ref{theorem:2.3}, we give sufficient conditions for judging whether a sub--stochastic matrix $F$ is strictly stable ($\rho(F)<1$), and establish mathematical expressions, which are closely related to $\alpha_1=\max\{\Lambda_s[F] \mid \Lambda_s[F]<1\}$ and $\alpha_2=\min\{[F]_{ij}>0 \mid i,j=1,2,\ldots,n\}$, for the upper bound of the spectral radius of the convergent sub-stochastic matrix. For example, in the following sub-stochastic matrix
\begin{equation*}
F'=\left(
    \begin{array}{ccc}
      0.5 & 0.5 & 0 \\
      0.3 & 0.2 & 0.4 \\
      0.4 & 0 & 0.6 \\
    \end{array}
  \right),
\end{equation*}
we can see that $\Lambda_2[F']<1$ and $\Lambda_1[F']=\Lambda_3[F']=1$, that is $\mathcal{S}_1=\{2\}, \mathcal{S}_2=\{1,3\}$. Since there exist two non-zero element chains $[F']_{12}>0$ and $[F']_{31}>0, [F']_{12}>0$ that satisfy the conditions in Theorem \ref{theorem:2.2}, we can directly determine that $\rho(F')<1$. Another example, when consider a sub-stochastic matrix
\begin{equation*}
F^*=\left(
    \begin{array}{ccc}
      0 & 0.2 & 0.8 \\
      0.5 & 0 & 0.4 \\
      0.3 & 0.6 & 0 \\
    \end{array}
  \right),
\end{equation*}
in which the diagonal elements are equal to zeros and $\mathcal{S}_2=\{1\}$ contains only a element, we can directly determine $\rho(F^*)\leq\sqrt{0.9}<1$ according to the result in Theorem \ref{theorem:2.3}.

\begin{theorem}\label{theorem:2.4}
For a series of sub-stochastic matrices $F_1,F_2,\ldots,F_q\in\mathbb{R}^{n\times n}$, the following results hold
\begin{enumerate}
\item [1)] the matrix product $\prod_{s=1}^q F_s$ is a sub-stochastic matrix;
\item [2)] if all matrices $F_1,F_2,\ldots,F_q$ contain positive diagonal elements, and there exist $1\leq s_1<s_2\leq q$ such that $\Lambda_{i_1}[F_{s_1}]<1$ and $[F_{s_2}]_{i_2i_1}>0$ for any $i_1,i_2\in\{1,2,\ldots,n\}$, then one has $\big\|\prod_{s=1}^qF_s\big\|_{\infty}<1$ and $\rho(\prod_{s=1}^q F_s)<1$.
\end{enumerate}
\end{theorem}

\begin{proof}
\emph{\textbf{The proof of 1).}} Since $F_1,F_2,\ldots,F_q$ are sub-stochastic matrices, we have $\Lambda_{i}[F_s]\leq 1$ for any $s=1,2,\ldots,q$ and $i=1,2,\ldots,n$. It follows that
\begin{equation*}
\begin{aligned}
\Lambda_{i}\left[\prod_{s=1}^q F_s\right]&=\sum_{j_1=1}^{n}[F_q]_{ij_1}\Lambda_{j_1}\left[\prod_{s=1}^{q-1} F_s\right]\\
&=\sum_{j_1=1}^{n}[F_q]_{ij_1}\sum_{j_2=1}^{n}[F_{q-1}]_{j_1j_2}\Lambda_{j_2}\left[\prod_{s=1}^{q-2} F_s\right]\\
&=\sum_{j_1=1}^{n}[F_q]_{ij_1}\sum_{j_2=1}^{n}[F_{q-1}]_{j_1j_2}\ldots\sum_{j_q=1}^{n}[F_1]_{j_{q-1}j_q}\leq1.
\end{aligned}
\end{equation*}
Consequently, $F_qF_{q-1}\cdots F_1$ is a sub-stochastic matrix.

\emph{\textbf{The proof of 2).}} By 1), one can derive that
\begin{equation*}
\begin{aligned}
\Lambda_{i_1}\left[\prod_{s=1}^{s_1} F_s\right]=\sum_{j=1}^{n}[F_{s_1}]_{i_1j}\Lambda_{j}\left[\prod_{s=1}^{s_1-1} F_s\right]\leq\Lambda_{i_1}[F_{s_1}]<1,
\end{aligned}
\end{equation*}
which together with the following result
\begin{equation*}
\begin{aligned}
\left[\prod_{s=s_1+1}^{s_2} F_s\right]_{i_2i_1}&\geq\sum_{j=1}^{n}[F_{s_2}]_{i_2j}\left[\prod_{s=s_1+1}^{s_2-1} F_s\right]_{ji_1}\\
&\geq[F_{s_2}]_{i_2i_1}\left[\prod_{s=s_1+1}^{s_2-1} F_s\right]_{i_1i_1}>0,
\end{aligned}
\end{equation*}
guarantees that
\begin{equation*}
\begin{aligned}
\Lambda_{i_2}\left[\prod_{s=1}^{s_2} F_s\right]&=\sum_{j=1,j\neq i_1}^{n}\left[\prod_{s=s_1+1}^{s_2} F_s\right]_{i_2j}\Lambda_{j}\left[\prod_{s=1}^{s_1} F_s\right]\\
& \ \ \ \ +\left[\prod_{s=s_1+1}^{s_2} F_s\right]_{i_2i_1}\Lambda_{i_1}\left[\prod_{s=1}^{s_1} F_s\right]\\
&<\Lambda_{i_2}\left[\prod_{s=s_1+1}^{s_2} F_s\right]\leq1.
\end{aligned}
\end{equation*}
It follows that
\begin{equation*}
\begin{aligned}
\Lambda_{i_2}\left[\prod_{s=1}^{q} F_s\right]&=\sum_{j=1,j\neq i_2}^{n}\left[\prod_{s=s_2+1}^{q} F_s\right]_{i_2j}\Lambda_{j}\left[\prod_{s=1}^{s_2} F_s\right]\\
& \ \ \ \ +\left[\prod_{s=s_2+1}^{q} F_s\right]_{i_2i_2}\Lambda_{i_2}\left[\prod_{s=1}^{s_2} F_s\right]\\
&<\Lambda_{i_2}\left[\prod_{s=s_2+1}^{q} F_s\right]\leq1.
\end{aligned}
\end{equation*}
This in turn means that $\|\prod_{s=1}^{q} F_s\|_{\infty}<1$. Consequently, we have $\rho(\prod_{s=1}^{q} F_s)\leq\|\prod_{s=1}^{q} F_s\|_{\infty}<1$. This completes the proof.
\end{proof}

In Theorem \ref{theorem:2.4}, we give some properties of the product of different sub-stochastic matrices, which can be used to analyze the stability of positive switched systems. Consider discrete-time positive switched linear systems
\begin{equation*}
\begin{aligned}
x(k+1)=F_{\sigma(k)}x(k),
\end{aligned}
\end{equation*}
where $x(k)$ is the system state at time instant $k\tau$, $\tau$ is the time step-size, $\sigma(k)$ is a switching signal which determines the active subsystem at each time instant, and subsystem $F_{\sigma(k)}$ is marginally stable, namely, $\rho(F_{\sigma(k)})=1$. We divide the time axis into a series of time interval $[0,k_1),[k_1,k_2),\ldots,[k_s,k_{s+1}),\ldots,$ $s\in\mathbb{N}$. When all subsystems in each interval $[k_s,k_{s+1})$ satisfy the conditions in 2) of Theorem \ref{theorem:2.4}, the stability of switched systems can be guaranteed. For example, consider the following three sub-stochastic matrices
\begin{equation*}
\begin{aligned}
F_1=\left(
      \begin{array}{ccc}
        0.3 & 0.6 & 0 \\
        0 & 1 & 0 \\
        0 & 0 & 1 \\
      \end{array}
    \right),~
F_2=\left(
      \begin{array}{ccc}
        1 & 0 & 0 \\
        0 & 1 & 0 \\
        0.4 & 0 & 0.5 \\
      \end{array}
    \right),~
F_3=\left(
      \begin{array}{ccc}
        1 & 0 & 0 \\
        0.5 & 0.5 & 0 \\
        0 & 0 & 1 \\
      \end{array}
    \right).
\end{aligned}
\end{equation*}
Obviously, all three matrices are marginally stable, that is, $\rho(F_s)=1$, $s=1,2,3$. In addition, $\Lambda_1[F_1]=0.9<1$, $[F_3]_{21}=0.5>0$ and $[F_2]_{31}>0$, which satisfy the conditions in Theorem \ref{theorem:2.4}. Thus, we have $\|F_3F_2F_1\|_{\infty}<1$. Assume that the switching signal satisfies $\sigma(3k)=1, \sigma(3k+1)=2, \sigma(3k+2)=3$. Then, one has
\begin{equation*}
\begin{aligned}
\lim_{k\rightarrow\infty}\|x(k)\|&=\Big\|\prod_{k=0}^{\infty}F_{\sigma(k)}x(0)\Big\|_{\infty}\\
&\leq\lim_{s\rightarrow\infty}\|F_3F_2F_1\|^s_{\infty}\|x(0)\|_{\infty}=0.
\end{aligned}
\end{equation*}
That is, the stability of switched systems is achieved.

We next give the definition of super-stochastic matrix.

\begin{definition}\label{definition:2.5}
A real matrix $F=[f_{ij}]\in\mathbb{R}^{n\times n}$ is super-stochastic if the following conditions hold:
\begin{enumerate}[1)]
\item $f_{ij}\geq0$ for all $i,j=1,2,\ldots,n$;
\item there exists a set $\mathcal{W}_1\subset\{1,2,\ldots,n\}$ such that $\Lambda_{i}[F]\geq1$, $i\in\mathcal{W}_1$ and $\Lambda_{i}[F]<1$, $i\in\{1,2,\ldots,n\}/\mathcal{W}_1$.
\end{enumerate}
\end{definition}

Next, we present some valuable results about the upper bound of spectral radius and product for super-stochastic matrices.

\begin{theorem}\label{theorem:2.6}
For an $n\times n$ super-stochastic matrix $F$, let $\mathcal{R}_1=\{s \mid \Lambda_s[F]<1\}$ and $\mathcal{R}_2=\{s \mid \Lambda_s[F]\geq1\}$ be non-empty. Under the inequality
\begin{equation}\label{sys:2.5}
\begin{aligned}
\sqrt[|\mathcal{C}^*|]{\alpha_3^{|\mathcal{C}^*|+1}-(\alpha_3-\alpha_1)\alpha^{|\mathcal{C}^*|}_2}<1,
\end{aligned}
\end{equation}
where $\alpha_1, \alpha_2, |\mathcal{C}^*|$ are defined in Theorem \ref{theorem:2.2} and $\alpha_3=\max\{\Lambda_s[F] \mid \Lambda_s[F]\geq 1\}$, there holds $\rho(F)<1$ if there exist non-zero element chains $\mathcal{C}_{i\rightarrow k}=[F]_{ki_r},\ldots,[F]_{i_2i_1},[F]_{i_1i}$ for each $k\in\mathcal{R}_2$, where $i\in\mathcal{R}_1$, $i_r\neq k,\ldots,i_1\neq i_2, i\neq i_1$.
\end{theorem}

\begin{proof}
Similar to the analysis of Theorem \ref{theorem:2.2}, it can be derived that
\begin{equation*}
\begin{aligned}
&\Lambda_{k}[F^{r+2}]\leq\alpha_3^{r+2}-(\alpha_3-\alpha_1)\alpha^{r+1}_2.
\end{aligned}
\end{equation*}
Since $|\mathcal{C}^*|\geq r+1$, we have
\begin{equation*}
\begin{aligned}
\Lambda_{k}[F^{|\mathcal{C}^*|+1}]\leq\alpha^{|\mathcal{C}^*|+1}_3-(\alpha_3-\alpha_1)\alpha^{|\mathcal{C}^*|}_2<1.
\end{aligned}
\end{equation*}
In addition, for any $i\in\mathcal{R}_1$, we have
\begin{equation*}
\begin{aligned}
\Lambda_{i}[F^{|\mathcal{C}^*|+1}]=\sum_{j=1}^n[F]_{ij}\Lambda_{j}[F^{|\mathcal{C}^*|}]
\leq\alpha^{|\mathcal{C}^*|+1}_3-(\alpha_3-\alpha_1)\alpha^{|\mathcal{C}^*|}_2<1.
\end{aligned}
\end{equation*}
This implies that $\Lambda_{s}[F^{|\mathcal{C}^*|+1}]\leq\alpha^{|\mathcal{C}^*|+1}_3-(\alpha_3-\alpha_1)\alpha^{|\mathcal{C}^*|}_2$ for any $s=1,2,\ldots,n$. Therefore,
\begin{equation*}
\begin{aligned}
\big\|F^{|\mathcal{C}^*|+1}\big\|_{\infty}\leq\alpha^{|\mathcal{C}^*|+1}_3-(\alpha_3-\alpha_1)\alpha^{|\mathcal{C}^*|}_2<1.
\end{aligned}
\end{equation*}
Which in turn means that
\begin{equation*}
\begin{aligned}
\rho(F)\leq\sqrt[|\mathcal{C}^*|]{\alpha_3^{|\mathcal{C}^*|+1}-(\alpha_3-\alpha_1)\alpha^{|\mathcal{C}^*|}_2}<1.
\end{aligned}
\end{equation*}
This proof is completed.
\end{proof}

\begin{theorem}\label{theorem:2.7}
Let $F$ be an $n\times n$ super-stochastic matrix in which $[F]_{ii}=0$, $i=1,2,\ldots,n$. If the number of elements in the set $\mathcal{R}_2$ is equal to 1, there holds that
\begin{equation}
\begin{aligned}
\rho(F)\leq\sqrt{\alpha_1\alpha_3},
\end{aligned}
\end{equation}
where $\alpha_1$ is defined in Theorem \ref{theorem:2.2}.
\end{theorem}

\begin{proof}
This proof is similar to Theorem \ref{theorem:2.3}, so it is omitted here.
\end{proof}

In Theorem \ref{theorem:2.6} and Theorem \ref{theorem:2.7}, we also show the mathematical formulas for determining the strict stability of a super-stochastic matrix. For example, for a super-stochastic matrix
\begin{equation*}
\tilde{F}=\left(
    \begin{array}{ccc}
      0.6 & 0 & 0 \\
      0.55 & 0 & 0.5 \\
      0.5 & 0.55 & 0 \\
    \end{array}
  \right),
\end{equation*}
in which $\alpha_1=0.6$, $\alpha_2=0.5$, $\alpha_3=1.05$ and $|\mathcal{C}^*|=1$, then it can be derived according to the result in Theorem \ref{theorem:2.6} that $\rho(\tilde{F})=0.8275<1$. As another example, in the following super-stochastic matrix
\begin{equation*}
\bar{F}=\left(
    \begin{array}{ccc}
      0 & 0.5 & 0.7 \\
      0.4 & 0 & 0.4 \\
      0.4 & 0.4 & 0 \\
    \end{array}
  \right),
\end{equation*}
all the diagonal elements are equal to zeros and $\mathcal{R}_2=\{1\}$ contains only an element. Thus, by Theorem \ref{theorem:2.7}, we have  $\rho(\bar{F})\leq\sqrt{1.2\times 0.8}=\sqrt{0.96}<1$.

\begin{theorem}\label{theorem:2.8}
Consider a series of $n\times n$ super-stochastic matrices $F_1,F_2,\ldots,F_q$ with positive diagonal elements. Denote
\begin{equation*}
\begin{aligned}
&g=\max\{\Lambda_{i}[F_s]\mid \Lambda_{i}[F_s]>1, \ s=1,2,\ldots,q, \ i=1,2,\ldots,n\},\\
&c=\max\{\Lambda_{i}[F_s]\mid \Lambda_{i}[F_s]<1, \ s=1,2,\ldots,q, \ i=1,2,\ldots,n\},\\
&\varphi=\min\{[F_{s}]_{ij} \mid s=1,2,\ldots,q, i,j=1,2,\ldots,n\}.
\end{aligned}
\end{equation*}
For each $i_2\in\{1,2,\ldots,n\}$, if there exist $1\leq s_1<s_2\leq q$ such that $\Lambda_{i_1}[F_{s_1}]<1$, $[F_{s_2}]_{i_2i_1}>0$, and the following condition holds:
\begin{equation}\label{sys:2.7}
\begin{aligned}
&g^q-(g-c)\varphi^{q-1}<1,
\end{aligned}
\end{equation}
then one has $\Lambda_{i_2}\big[\prod_{s=1}^{q} F_s\big]<1$ and $\rho(\prod_{s=1}^{q} F_s)<1$.
\end{theorem}

\begin{proof}
According to the given conditions, we have $\Lambda_{i}\big[\prod_{s=1}^{s_1-1}F_{s}\big]\leq g^{s_1-1}$ for any $i=1,2,\ldots,n$. This in turn leads to the following fact
\begin{equation}\label{sys:2.8}
\begin{aligned}
\Lambda_{i_1}\left[\prod_{s=1}^{s_1} F_s\right]&=\sum_{j=1}^{n}[F_{s_1}]_{i_1j}\Lambda_{j}\left[\prod_{s=1}^{s_1-1} F_s\right]\\
&\leq\Lambda_{i_1}[F_{s_1}] g^{s_1-1}\leq c g^{s_1-1},\\
\Lambda_{i}\left[\prod_{s=1}^{s_1} F_s\right]&\leq g^{s_1}, \ i\neq i_1.
\end{aligned}
\end{equation}
In addition, we have for the matrix $\prod_{s=s_1+1}^{s_2-1} F_s$ that
\begin{align}\label{sys:2.9}
\left[\prod_{s=s_1+1}^{s_2-1} F_s\right]_{i_1i_1}\geq \varphi^{s_2-s_1-1}, \ \Lambda_{i_1}\left[\prod_{s=s_1+1}^{s_2-1} F_s\right]\leq g^{s_2-s_1-1}.
\end{align}
Combining (\ref{sys:2.8}) and (\ref{sys:2.9}), we can deduce
\begin{equation*}
\begin{aligned}
\Lambda_{i_1}\left[\prod_{s=1}^{s_2-1} F_s\right]&=\sum_{j=1,j\neq i_1}^{n}\left[\prod_{s=s_1+1}^{s_2-1} F_s\right]_{i_1j}\Lambda_{j}\left[\prod_{s=1}^{s_1} F_s\right]\\
& \ \ \ \ +\left[\prod_{s=s_1+1}^{s_2-1} F_s\right]_{i_1i_1}\Lambda_{i_1}\left[\prod_{s=1}^{s_1} F_s\right]\\
&\leq g^{s_2-1}-(g-c)\varphi^{s_2-s_1-1}g^{s_1-1},\\
\Lambda_{i}\left[\prod_{s=1}^{s_2-1} F_s\right]&\leq g^{s_2-1}, \ i\neq i_1.
\end{aligned}
\end{equation*}
Furthermore,
\begin{equation*}
\begin{aligned}
\Lambda_{i_2}\left[\prod_{s=1}^{s_2} F_s\right]&=\sum_{j=1,j\neq i_2}^{n}\big[F_{s_2}\big]_{i_2j}\Lambda_{j}\left[\prod_{s=1}^{s_2-1} F_s\right]\\
& \ \ \ \ +\big[F_{s_2}\big]_{i_2i_1}\Lambda_{i_1}\left[\prod_{s=1}^{s_2-1} F_s\right]\\
&\leq g^{s_2}-(g-c)\varphi^{s_2-s_1}g^{s_1-1},\\
\Lambda_{i}\left[\prod_{s=1}^{s_2} F_s\right]&\leq g^{s_2}, \ i\neq i_2,
\end{aligned}
\end{equation*}
which together with the facts
\begin{align*}
\left[\prod_{s=s_2+1}^{q} F_s\right]_{i_2i_2}&\geq \varphi^{q-s_2}, \ \Lambda_{i_2}\left[\prod_{s=s_2+1}^{q} F_s\right]\leq g^{q-s_2},
\end{align*}
guarantees that
\begin{equation*}
\begin{aligned}
\Lambda_{i_2}\left[\prod_{s=1}^{q} F_s\right]&=\sum_{j=1,j\neq i_2}^{n}\left[\prod_{s=s_2+1}^{q} F_s\right]_{i_2j}\Lambda_{j}\left[\prod_{s=1}^{s_2} F_s\right]\\
& \ \ \ \ +\left[\prod_{s=s_2+1}^{q} F_s\right]_{i_2i_2}\Lambda_{i_2}\left[\prod_{s=1}^{s_2} F_s\right]\\
&\leq g^{q}-(g-c)\varphi^{q-s_1}g^{s_1-1}\\
&\leq g^{q}-(g-c)\varphi^{q-1}<1.
\end{aligned}
\end{equation*}
This implies that $\big\|\prod_{s=1}^{q} F_s\big\|_{\infty}<1$, and further $\rho(\prod_{s=1}^{q} F_s)\leq\|\prod_{s=1}^{q} F_s\|_{\infty}<1$. The proof is completed.
\end{proof}

In Theorem \ref{theorem:2.8}, we show some product properties of super-stochastic matrices. It is noteworthy that even if the subsystem $F_{\sigma(k)}$ is strictly unstable, that is $\rho(F_{\sigma(k)})>1$, the stability of switched systems can be guaranteed through the result of Theorem \ref{theorem:2.8}. For example, given the following three super-stochastic matrices
\begin{equation*}
\begin{aligned}
F_1=\left(
      \begin{array}{ccc}
        0 & 0.5 & 0 \\
        0 & 1.02 & 0 \\
        0.5 & 0 & 0.5 \\
      \end{array}
    \right),
F_2=\left(
      \begin{array}{ccc}
        0.5 & 0.52 & 0 \\
        0.5 & 0 & 0.5 \\
        0 & 0 & 1.02 \\
      \end{array}
    \right),
F_3=\left(
      \begin{array}{ccc}
        1.02 & 0 & 0 \\
        0.5 & 0.5 & 0 \\
        0 & 0.5 & 0.5 \\
      \end{array}
    \right).
\end{aligned}
\end{equation*}
It can be seen that $\rho(F_s)=1.02>1$, $s=1,2,3$, namely, all matrices are strictly unstable. However, because these three super-stochastic  matrices satisfy the conditions of Theorem \ref{theorem:2.8}, we have $\rho(F_3F_2F_1)\leq\|F_3F_2F_1\|_{\infty}<1$. When the switching signal satisfies $\sigma(3k)=1, \sigma(3k+1)=2, \sigma(3k+2)=3$, the stability of switched systems is achieved.

\begin{remark}\label{remark:2.9}
In the existing literature \cite{Zheng2018Stability} on the stability of positive switched systems, subsystems are required to be stable
and marginally stable, that is, $\rho(F_{\sigma(k)})$ must be less than or equal to 1. Compared with the work \cite{Zheng2018Stability}, the results given in Theorems \ref{theorem:2.8} can still guarantee the stability of the system even if all subsystems satisfying the given conditions are strictly unstable, namely, $\rho(F_{\sigma(k)})>1$. In other words, the conclusions about sub-stochastic matrix and super-stochastic matrix provide a reference for extending the existing study of the stability of positive switched systems.
\end{remark}

\subsection{Signed digraph}

Consider a multi-agent network composed of a leader (denoted by $v_{0}$) and $n$ followers (denoted by $v_{1},v_{2},\ldots,v_{n}$). The leader is an agent that does not receive information from other agents, and the followers update the states by receiving information from the neighbor agents. The information interactions among the agents are described by a signed digraph $\mathscr{G}=(\mathscr{E},\mathscr{V},\mathscr{A})$. An edge that starts at $v_j$ and ends at $v_i$ in $\mathscr{E}$ is defined as $(v_{j},v_{i})$. Edge $( v_{j}, v_{i})\in \mathscr{E}$ if and only if the information of follower $ v_{j}$ can be detected by follower $ v_{i}$. The set of the neighbors of follower $v_i$ is denoted by $\mathscr{N}_i=\{v_j \mid (v_j,v_i)\in\mathscr{E}\}$. The elements in the adjacency matrix $\mathscr{A}=[a_{ij}]$ satisfy: $a_{ij}\neq0\Leftrightarrow(v_j,v_i)\in\mathscr{E}$; otherwise, $a_{ij}=0$. Furthermore, $a_{ij}>0$ ($a_{ij}<0$) means that follower $v_i$ can receive follower $v_j$'s cooperative (competitive) information. Let $a_{i0}$ describe the communication from leader $v_{0}$ to follower $v_i$, where $a_{i0}>0$ ($a_{i0}<0$) if follower $v_{i}$ can detect the leader's cooperative (competitive) information, and $a_{i0}=0$ means that follower $v_{i}$ cannot receive the information of the leader. Considering the leader, we construct a new digraph $\tilde{\mathscr{G}}=(\tilde{\mathscr{E}},\tilde{\mathscr{V}})$ which consists of digraph $\mathscr{G}$, vertex $ v_0$ and those edges that start at the leader and end at the followers. A directed path that starts at $ v_{i_0}$ and ends at $ v_{i_r}$ in digraph $\tilde{\mathscr {G}}$ is denoted by $v_{i_0}\rightarrow v_{i_1}\rightarrow\cdots\rightarrow v_{i_r}$, where $ v_{i_0}, v_{i_1}, \ldots , v_{i_r}\in \tilde{\mathscr{V}}$ are all distinct. Let $d(v_{i},v_{j})$ denote the directed distance from $v_{i}$ to $v_{j}$, which is the number of edges in the shortest path from $v_{i}$ to $v_{j}$. A digraph $\bar{\mathscr{G}}=\{\bar{\mathscr{E}},\bar{\mathscr{V}}\}$ is called a spanning subgraph of digraph $\tilde{\mathscr{G}}$ if $\bar{\mathscr{V}}=\tilde{\mathscr{V}}$ and $\bar{\mathscr{E}}\subseteq\tilde{\mathscr{E}}$. For subsequent use, we denote
\begin{equation*}
\begin{aligned}
&\mathscr{D}=\diag\big\{\sum_{j=1}^{n}|a_{1j}|,\sum_{j=1}^{n}|a_{2j}|,\ldots,\sum_{j=1}^{n}|a_{nj}|\big\},\\
&\mathscr{B}=\diag\big\{|a_{10}|,|a_{20}|,\ldots,|a_{n0}|\big\}.
\end{aligned}
\end{equation*}

For realizing asynchronous bipartite tracking, the following conditions based on the communication topology are required:
\begin{enumerate}
\item [\textbf{C1}.] The signed digraph $\tilde{\mathscr{G}}$ is structurally balanced, that is, all vertices are divided into two subsets $\mathscr{V}_1$ and $\mathscr{V}_2$ such that $\mathscr{V}_1\cup\mathscr{V}_2=\mathscr{\tilde{V}}$, $\mathscr{V}_1\cap\mathscr{V}_2=\emptyset$. Besides, $a_{ij}\geq0$ if vertices $v_i$ and $v_j$ belong to the same subset, and $a_{ij}\leq0$ if vertices $v_i$ and $v_j$ belong to different subsets.

\item [\textbf{C2}.] For each follower, digraph $\tilde{\mathscr{G}}$ contains a directed path that starts at the leader and ends at that follower.
\end{enumerate}

\begin{remark}\label{remark:2.10}
Condition \textbf{C1} is to ensure that the agents belonging to the same subset eventually reach the exact same state, while the agents belonging to different subsets eventually reach the states with same size and opposite signs. Condition \textbf{C2} is to guarantee that each follower can directly or indirectly detect the information from the leader. It is noteworthy that \textbf{C1} and \textbf{C2} are the basic conditions that the communication topology must satisfy in order to achieve bipartite tracking so far (e.g., see \cite{Ma2018Necessary,Wen2018Bipartite,Liu2018Robust,Yu2018Prescribed,Wu2019Bipartite,Gong2019Fixed}).
\end{remark}

With no loss of generality, it is assumed that $\mathscr{V}_1=\{v_0,v_1,\ldots,v_m\}$ and $\mathscr{V}_2=\{v_{m+1},\ldots,v_n\}$. According to the division of subsets $\mathscr{V}_1$ and $\mathscr{V}_2$, the adjacency matrix $\mathscr{A}$ of the digraph $\mathscr{G}$ is partitioned as
\begin{equation*}
\mathscr{A}=\left(
    \begin{array}{cc}
      \mathscr{A}_{11} & \mathscr{A}_{12} \\
      \mathscr{A}_{21} & \mathscr{A}_{22} \\
    \end{array}
  \right),
\end{equation*}
where $\mathscr{A}_{11}\geq0$, $\mathscr{A}_{22}\geq0$, $\mathscr{A}_{12}\leq0$ and $\mathscr{A}_{21}\leq0$.

\subsection{Composition of edge sets}

Let $\mathscr{E}_{1}$ and $\mathscr{E}_{2}$ be two directed edge sets that depend on the same vertex set $\tilde{\mathscr{V}}$. The composition of $\mathscr{E}_{1}$ and $\mathscr{E}_{2}$ is represented by $\mathscr{E}_{1}\circ\mathscr{E}_{2}$, which is a directed edge set and satisfies: $(i,j)\in\mathscr{E}_{1}\circ\mathscr{E}_{2}$ if for some $k$, $(i,k)\in\mathscr{E}_{1}$ and $(k,j)\in\mathscr{E}_{2}$. For a finite sequence of directed edge sets $\mathscr{E}_{1},\mathscr{E}_{2},\ldots,\mathscr{E}_{s}$ that depend on the same vertex set $\tilde{\mathscr{V}}$, let us agree to say that the composition $\mathscr{E}_{1}\circ\mathscr{E}_{2}\circ\cdots\circ\mathscr{E}_{s}$ associated with the set $\mathscr{V}^*\subseteq\tilde{\mathscr{V}}$ is rooted at the vertex $i_0\in\tilde{\mathscr{V}}$ if for any $i_s\in\mathscr{V}^*$, there exist some vertices $i_1,i_2,\ldots,i_{s-1}\in\tilde{\mathscr{V}}$ such that $(i_0,i_1)\in\mathscr{E}_{1},(i_1,i_2)\in\mathscr{E}_{2},\ldots, (i_{s-1},i_s)\in\mathscr{E}_{s}$.

\begin{theorem}\label{theorem:2.11}
Consider a group of digraphs $\mathscr{G}_1=\{\tilde{\mathscr{V}},\mathscr{E}_1\},\mathscr{G}_2=\{\tilde{\mathscr{V}},\mathscr{E}_2\},\ldots$, $\mathscr{G}_q=\{\tilde{\mathscr{V}},\mathscr{E}_q\}$ with the same vertex set $\tilde{\mathscr{V}}$. If there exists a set $\mathscr{V}^*\subseteq\tilde{\mathscr{V}}$ and a vertex $v_{0}\in\tilde{\mathscr{V}}$ that satisfy:
\begin{enumerate}
\item [1).] $(v_{0},v_{i_q})\in\mathscr{E}_1\cup\mathscr{E}_2\cup\ldots\cup\mathscr{E}_q$ for any $v_{i_q}\in\mathscr{V}^*$;

\item [2).] the vertices $v_{0},v_{i_q}$ have self-loops in all digraphs $\mathscr{G}_p$, $p=1,2,\ldots,q$;
\end{enumerate}
then the composition $\mathscr{E}_1\circ\mathscr{E}_2\circ\cdots\circ\mathscr{E}_q$ associated with the set $\mathscr{V}^*$ is rooted at the vertex $v_{0}$.
\end{theorem}

\begin{proof}\
The following two different cases are analyzed to obtain this result.

Case I: There exist vertices $v_{i_1},v_{i_2},\ldots,v_{i_{q-1}}\notin\{v_{0},v_{i_q}\}$ such that $(v_{0},v_{i_1})\in\mathscr{E}_1$, $(v_{i_1},v_{i_2})\in\mathscr{E}_2,\ldots,(v_{i_{q-1}},v_{i_q})\in\mathscr{E}_q$, where $\{v_{i_1},v_{i_2},\ldots,v_{i_{q-1}}\}$ may contain the same elements. Then it is easy to obtain that
\begin{equation*}
(v_{0},v_{i_q})\in\mathscr{E}_1\circ\mathscr{E}_2\circ\cdots\circ\mathscr{E}_q.
\end{equation*}
Obviously, this result still holds if all digraphs $\mathscr{G}_p$, $p=1,2,\ldots,q$ have self-loops on the vertices $v_{0}, v_{iq}$.

Case II: There are no vertices $v_{i_1},v_{i_2},\ldots,v_{i_{q-1}}\notin\{v_{0},v_{i_q}\}$ such that $(v_{0},v_{i_1})\in\mathscr{E}_1$, $(v_{i_1},v_{i_2})\in\mathscr{E}_2,\ldots,(v_{i_{q-1}},v_{i_q})\in\mathscr{E}_q$. Since $(v_{0},v_{i_q})\in\mathscr{E}_1\cup\mathscr{E}_2\cup\ldots\cup\mathscr{E}_q$, there exists a digraph $\mathscr{G}_{s}\in\{\mathscr{G}_1,\mathscr{G}_2,\ldots,\mathscr{G}_q\}$ such that $(v_{0},v_{i_q})\in\mathscr{E}_s$, which, together with the condition that vertices $v_{0},v_{i_q}$ have self-loops in digraphs $\mathscr{G}_p$, $p=1,2,\ldots,q$, guarantees
\begin{equation*}
(v_{0},v_{i_q})\in\mathscr{E}_1\circ\mathscr{E}_2\circ\cdots\circ\mathscr{E}_q,
\end{equation*}
where $(v_{0},v_{0})\in\mathscr{E}_1$, $(v_{0},v_{0})\in\mathscr{E}_2,\ldots,(v_{0},v_{0})\in\mathscr{E}_{s-1}$, $(v_{0},v_{i_q})\in\mathscr{E}_s$, $(v_{i_q},v_{i_q})\in\mathscr{E}_{s+1}$, $(v_{i_q},v_{i_q})\in\mathscr{E}_{s+2},\ldots,(v_{i_q},v_{i_q})\in\mathscr{E}_{q}$.

This completes the proof.
\end{proof}

\subsection{Asynchronous setting}\label{section:2.4}

Let $\mathscr{T}=\{0,\tau,\ldots,k\tau,\ldots\}$ be the set including all discrete-time instants, where $\tau>0$ is the fixed update step-size. The asynchronous setting considered in this paper means that the communication time instants of each follower, at which the follower communicates with its neighbors, are independent of the other followers' and can be unevenly distributed. Assume that the set of follower $v_i$'s communication time instants is $\{s^{i}_{k}\tau\}=\{s^{i}_{0}\tau,s^{i}_{1}\tau,\ldots,s^{i}_{k}\tau,\ldots\}$, which satisfies  $s^{i}_{0}\tau,s^{i}_{1}\tau,\ldots,s^{i}_{k}\tau\in\mathscr{T}$ and $0=s^{i}_{0}\tau<s^{i}_{1}\tau<\cdots<s^{i}_{k}\tau<\cdots$. It is further assumed for any $k\in\mathbb{N}$ and $i\in\{1,2,\ldots,n\}$, $\{s^{i}_{k}\tau\}$ satisfies the following condition:
\begin{eqnarray}\label{sys:2.10}
s^{i}_{k+1}-s^{i}_{k}\leq h,
\end{eqnarray}
where $h\in\mathbb{Z}_{+}$ is a constant.

\begin{remark}\label{remark:2.12}
Condition (\ref{sys:2.10}) is necessary for implementing the asynchronous bipartite tracking of MASs. Without condition (\ref{sys:2.10}), there may be an agent who cannot receive its neighbors' information all the time, which means that the asynchronous bipartite tracking may not be implemented.
\end{remark}

\section*{Part I}

Through the the introduction in Section \ref{section:2.2}, it can be seen that the results of sub-stochastic matrix and super-stochastic matrix can be applied to analyze the stability of positive switched systems. In order to show the wide application of these results, below we will use these results to study the dynamics phenomenon of bipartite tracking of MASs on signed networks with both cooperation and competition in detail. In Section \ref{section:3}-Section \ref{section:6}, we mainly analyze the bipartite tracking dynamics of first-order MASs, second-order MASs and general linear MASs under the asynchronous setting by using the properties of ISubSM and ISupSM.

\section{Bipartite tracking of asynchronous first-order MASs}\label{section:3}

We study the bipartite tracking issue of first-order MASs with the asynchronous setting based on ISubSM in this section. In the discrete-time setting, the leader is modeled by the following dynamics
\begin{eqnarray}\label{sys:3.1}
x_{0}(k+1)=x_{0}(k),
\end{eqnarray}
and the state update rule of each follower $v_i\in\mathscr{V}$ is given by
\begin{eqnarray}\label{sys:3.2}
x_{i}(k+1)=x_{i}(k)+\tau u_{i}(k),
\end{eqnarray}
where $x_{i}(k)\in \mathbb{R}^{p}$ represents agent $v_i$'s position at discrete-time instant $k\tau$. For each follower $v_i$, the following asynchronous distributed control input (control protocol) is designed:
\begin{eqnarray}\label{sys:3.3}
\left\{
\begin{aligned}
u_{i}(k)=&\psi\sum\limits_{v_j\in\mathscr{N}_i}|a_{ij}|\big[\sgn(a_{ij})x_{j}(k)-x_{i}(k)\big]\\
&+\psi|a_{i0}|\big[\sgn(a_{i0})x_{0}(k)-x_{i}(k)\big],\ {\rm if} \ k\tau\in\{s^{i}_{k}\tau\};\\
u_{i}(k)=&\, 0, \ {\rm if} \ k\tau\notin\{s^{i}_{k}\tau\},
\end{aligned}
\right.
\end{eqnarray}
where $\psi>0$ is a constant gain parameter. Observing from (\ref{sys:3.3}), follower $v_i$'s position is fixed during time intervals $[s^{i}_{k}\tau,s^{i}_{k+1}\tau)$, $k\in\mathbb{N}$, and its position changes only at time instants $s^{i}_{k}\tau$, $k\in\mathbb{N}$.

\begin{definition}\label{definition:3.1}
The bipartite tracking for systems (\ref{sys:3.1}) and (\ref{sys:3.2}) is said to be realized if the following conditions are satisfied:
\begin{eqnarray*}
\begin{split}
& \forall v_i\in\mathscr{V}_1,\ \lim_{k\rightarrow \infty}\left\|x_{i}(k)-x_{0}(k)\right\|=0,\\
&\forall v_i\in\mathscr{V}_{2},\ \lim_{k\rightarrow \infty}\left\|x_{i}(k)+x_{0}(k)\right\|=0.
\end{split}
\end{eqnarray*}
\end{definition}

As can be seen from protocol (\ref{sys:3.3}), the interactions among the agents are time-varying in the asynchronous setting and the variability of interactions is arbitrary and irregular. For describing the information interactions at different time instants, we construct some new signed digraphs below. Let a digraph $\mathscr{G}(k)=(\mathscr{V},\mathscr{E}(k),\mathscr{A}(k))\subseteq\mathscr{G}$ portray the information exchange among the followers at time instant $k\tau$. For any edge $(j,i)\in\mathscr{E}$, $(j,i)\in\mathscr{E}(k)$ if and only if $k\tau\in\{s^{i}_{k}\tau\}$, and otherwise $(j,i)\notin\mathscr{E}(k)$. The set of follower neighbors of agent $v_i$ at time instant $k\tau$ is denoted by $\mathscr{N}_i(k)=\{v_j \mid (v_j,v_i)\in\mathscr{E}(k)\}$. The weighted adjacency matrix $\mathscr{A}(k)=[a_{ij}(k)]$ satisfies: for any $i=1,2,\ldots,n$, if $k\tau\in\{s^{i}_{k}\tau\}$, then $a_{ij}(k)=a_{ij}$, $j=1,2,\ldots,n$; otherwise, $a_{ij}(k)=0$, $j=1,2,\ldots,n$. Based on the division of subsets $\mathscr{V}_1$ and $\mathscr{V}_2$, we partition the adjacency matrix $\mathscr{A}(k)$ as
\begin{equation*}
\mathscr{A}(k)=\left(
    \begin{array}{cc}
      \mathscr{A}_{11}(k) & \mathscr{A}_{12}(k) \\
      \mathscr{A}_{21}(k) & \mathscr{A}_{22}(k) \\
    \end{array}
  \right),
\end{equation*}
where $\mathscr{A}_{11}(k)\geq0$, $\mathscr{A}_{22}(k)\geq0$, $\mathscr{A}_{12}(k)\leq 0$ and $\mathscr{A}_{21}(k)\leq 0$. Let $a_{i0}(k)$ describe the communication from the leader to follower $v_i$ at time instant $k\tau$, where $a_{i0}(k)=a_{i0}$ if $k\tau\in\{s^{i}_{k}\tau\}$, and $a_{i0}(k)=0$ if $k\tau\notin\{s^{i}_{k}\tau\}$. Below we construct a new graph $\tilde{\mathscr{G}}(k)=(\tilde{\mathscr{E}}(k),\tilde{\mathscr{V}}(k))$ consisting of the graph $\mathscr{G}(k)$, vertex $ v_0$ and the directed edges that start at the leader and end at the followers at time instant $k\tau$. Denote
\begin{equation*}
\begin{aligned}
&\mathscr{D}(k)=\diag\Big\{\sum_{v_j\in\mathscr{N}_1(k)}|a_{1j}(k)|,\ldots,\sum_{v_j\in\mathscr{N}_n(k)}|a_{nj}(k)|\Big\},\\
&\mathscr{B}(k)=\diag\big\{|a_{10}(k)|,|a_{20}(k)|,\ldots,|a_{n0}(k)|\big\}.
\end{aligned}
\end{equation*}
Based on the structure of $\tilde{\mathscr{G}}(k)$, protocol (\ref{sys:3.3}) can be equivalently expressed as
\begin{equation}\label{sys:3.4}
\begin{aligned}
u_{i}(k)&=\psi\sum\limits_{v_j\in\mathscr{N}_i(k)}|a_{ij}(k)|\big[\sgn(a_{ij}(k))x_{j}(k)-x_{i}(k)\big]\\
& \ \ \ \ +\psi|a_{i0}(k)|\big[\sgn(a_{i0}(k))x_{0}(k)-x_{i}(k)\big].
\end{aligned}
\end{equation}

Before moving on, the following model transformations are introduced:
\begin{eqnarray}\label{sys:3.5}
\begin{aligned}
&e_{x1}(k)=\left[x^T_{1}(k)-x^T_{0}(k),\ldots,x^T_{m}(k)-x^T_{0}(k)\right]^T,\\
&e_{x2}(k)=\left[-x^T_{m+1}(k)\!-\!x^T_{0}(k),\ldots,\!-x^T_{n}(k)\!-\!x^T_{0}(k)\right]^T,\\
&e_{x}(k)=\left[e_{x1}^{T}(k),e_{x2}^{T}(k)\right]^{T}.
\end{aligned}
\end{eqnarray}
Using (\ref{sys:3.4}) and (\ref{sys:3.5}), systems (\ref{sys:3.1}) and (\ref{sys:3.2}) can be written as an error system with the following compact form:
\begin{eqnarray}\label{sys:3.6}
\begin{aligned}
e_{x}(k+1)=\big[M(k)\otimes I_{p}\big]e_{x}(k),
\end{aligned}
\end{eqnarray}
where $M(k)=I_n-\tau\psi\mathscr{D}(k)-\tau\psi\mathscr{B}(k)+\tau\psi|\mathscr{A}(k)|$.

Apparently, the implementation of asynchronous bipartite tracking for systems (\ref{sys:3.1}) and (\ref{sys:3.2}) is equivalent to the achievement of asymptotic stability of error system (\ref{sys:3.6}). That is, we need to prove that $\lim_{k\rightarrow\infty}\prod_{s=0}^kM(s)=\mathbf{0}$ in order to achieve the asynchronous bipartite tracking. If the parameter $\psi$ satisfies
\begin{eqnarray}\label{sys:3.7}
\psi<\frac{1}{\tau d_{M}},
\end{eqnarray}
where $d_M\!=\!\max\big\{\sum_{v_j\in\mathscr{N}_i}|a_{ij}|\!+\!|a_{i0}| \mid i\!=\!1,2,\ldots,n\big\}$, then $M(k)$, $k\in\mathbb{N}$ are sub-stochastic matrices in which the diagonal elements are positive and the row sums satisfy:
\begin{eqnarray*}
\left\{
\begin{split}
&\Lambda_{i}\big[M(k)\big]<1 \ {\rm if} \ (v_0,v_i)\in\tilde{\mathscr{E}}(k),\\
&\Lambda_{i}\big[M(k)\big]=1 \ {\rm if} \ (v_0,v_i)\notin\tilde{\mathscr{E}}(k),
\end{split}
\right.
\end{eqnarray*}
where $i\in\{1,2,\ldots,n\}$. Thus, the issue of asynchronous bipartite tracking is equivalently transformed into the product convergence issue of ISubSM.

\begin{lemma}\label{lemma:3.2}
Suppose that the inequality (\ref{sys:3.7}) holds. If the communication topology satisfies the conditions \textbf{C1} and \textbf{C2}, then for any $k\in\mathbb{N}$, there holds that
\begin{eqnarray}\label{sys:3.8}
\Big\|\prod_{s=k}^{k+Ph-1}M(s)\Big\|_{\infty}<1,
\end{eqnarray}
where $P=\max\{d( v_0, v_i)\mid i=1,2,\ldots,n\}$.
\end{lemma}

\begin{proof}\
Under the inequality (\ref{sys:3.7}), $M(k)$, $k\in\mathbb{N}$ are sub-stochastic matrices with positive diagonal elements. By the conditions \textbf{C1} and \textbf{C2}, there is a directed path $W_z=v_0\rightarrow v_{\delta_{1}}\rightarrow v_{\delta_{2}}\rightarrow\cdots\rightarrow v_{\delta_{z}}$ in $\tilde{\mathscr{G}}$ for each vertex $v_{\delta_{z}}\in\mathscr{V}$, where $\delta_{1},\delta_{2},\ldots,\delta_{z}\in\{1,2,\ldots,n\}$. Known from the condition in (\ref{sys:2.10}) that  $s^{i}_{k+1}-s^{i}_{k}\leq h$ for any $k\in\mathbb{N}$ and $i\in\{1,2,\ldots,n\}$, follower $v_{\delta_{1}}$ communicates with its neighbors at least once during the time interval $[k\tau,k\tau+h\tau)$. Assume that one of the $v_{\delta_{1}}$'s communication time instants during $[k\tau,k\tau+h\tau)$ is $k\tau+l_0\tau$, which satisfies $k\tau\leq k\tau+l_0\tau<k\tau+h\tau$. Then, we have $(v_0,v_{\delta_{1}})\in\tilde{\mathscr{E}}(k+l_0)$. It follows that $\Lambda_{\delta_{1}}[M(k+l_0)]<1$, and further
\begin{eqnarray}\label{sys:3.9}
\Lambda_{\delta_{1}}\left[\prod_{s=k}^{k+l_0}M(s)\right]
=\sum_{i=1}^n\big[M(k+l_0)\big]_{\delta_{1},i}\Lambda_i\left[\prod_{s=k}^{k+l_0-1}M(s)\right]<1.
\end{eqnarray}
Since $[M(k+h-1)]_{\delta_{1},\delta_{1}}>0$, we can obtain by using the result in Theorem \ref{theorem:2.4} that
\begin{eqnarray}\label{sys:3.10}
\Lambda_{\delta_{1}}\left[\prod_{s=k}^{k+h-1}M(s)\right]<1.
\end{eqnarray}

Below our purpose is to analyze $\Lambda_{\delta_{2}}\big[\prod_{s=k}^{k+2h-1}M(s)\big]$. It is known that follower $v_{\delta_{2}}$ communicates with its neighbors at least once during the time interval $[k\tau+h\tau,k\tau+2h\tau)$. Assume that one of the $v_{\delta_{2}}$'s communication instants during $[k\tau+h\tau,k\tau+2h\tau)$ is $k\tau+h\tau+l_1\tau$ that satisfies $k\tau+h\tau\leq k\tau+h\tau+l_1\tau<k\tau+2h\tau$. Then, we have $(v_{\delta_{1}},v_{\delta_{2}})\in\tilde{\mathscr{E}}(k+h+l_1)$. Equivalently, $[M(k+h+l_1)]_{\delta_{2}\delta_{1}}>0$. By combining this result with (\ref{sys:3.9}), it can be derived noting Theorem~\ref{theorem:2.4} that
\begin{eqnarray}\label{sys:3.11}
\Lambda_{\delta_{2}}\left[\prod_{s=k}^{k+2h-1}M(s)\right]<1.
\end{eqnarray}
Along this line of analysis, we can further get $\Lambda_{\delta_{z}}[\prod_{s=k}^{k+zh-1}M(s)]<1$. Taking note of this fact that $P\geq d(v_0, v_{u_z})=z$, where $z\in\mathbb{Z}_+$ is the length of $W_z$, we can also obtain according to Theorem~\ref{theorem:2.4} that
\begin{eqnarray}\label{sys:3.12}
\Lambda_{\delta_{z}}\left[\prod_{s=k}^{k+Ph-1}M(s)\right]<1.
\end{eqnarray}
This implies that
\begin{eqnarray}\label{sys:3.13}
\begin{aligned}
\Big\|\prod_{s=k}^{k+Ph-1}M(s)\Big\|_{\infty}=\max_{\delta_z=1,2,\ldots,n}\left\{\Lambda_{\delta_{z}}\left[\prod_{s=k}^{k+Ph-1}M(s)\right]\right\}<1.
\end{aligned}
\end{eqnarray}
The proof is completed.
\end{proof}

\begin{theorem}\label{theorem:3.3}
Suppose that the gain parameter $\psi$ satisfies the inequality (\ref{sys:3.7}). The asynchronous bipartite tracking for systems (\ref{sys:3.1}) and (\ref{sys:3.2}) with distributed protocol (\ref{sys:3.3}) can be achieved if and only if the communication topology meets the conditions \textbf{C1} and \textbf{C2}.
\end{theorem}

\begin{proof}
\emph{\underline{Sufficiency}}:\
First of all, we divide the time axis into a series of time intervals $[\theta Ph\tau,\theta Ph\tau\!+\!Ph\tau)$, $\theta\in\mathbb{N}$. Let $M^*(\theta)=\prod_{s=\theta Ph}^{\theta Ph+Ph-1}M(s)$. Then it is known from Lemma~\ref{lemma:3.2} that $\|\Xi(\theta)\|_{\infty}<1$, $\theta\in\mathbb{N}$. Thus, we can derive from error system (\ref{sys:3.6}) that
\begingroup
\allowdisplaybreaks
\begin{eqnarray}\label{sys:3.14}
\begin{aligned}
\lim_{k\rightarrow\infty}\|e_x(k+1)\|_{\infty}&=\lim_{k\rightarrow\infty}\left\|\left[\prod_{s=0}^kM(s)\otimes I_p\right]e_x(0)\right\|_{\infty}\\
&=\lim_{\theta\rightarrow\infty}\left\|\left[\prod_{s=0}^\theta M^*(s)\otimes I_p\right]e_x(0)\right\|_{\infty}\\
&\leq\lim_{\theta\rightarrow\infty}\prod_{s=1}^{\theta}\|M^*(s)\|_{\infty}\|e_x(0)\|_{\infty}=0.
\end{aligned}
\end{eqnarray}
\endgroup
This obviously means that $\lim_{k\rightarrow \infty}\|x_{i}(k)-x_{0}(k)\|_{\infty}=0$ for any $i=1,\ldots, m$, and $\lim_{k\rightarrow \infty}\|-x_{i}(k)-x_{0}(k)\|_{\infty}=0$ for any $i=m+1,\ldots, n$. By Definition~\ref{definition:3.1}, the asynchronous bipartite tracking can be achieved.

\emph{\underline{Necessity}}:
The necessity is proved in two steps as follows.

1) \emph{Necessity of condition} \textbf{C1}: Suppose that the condition \textbf{C1} is not satisfied. Then we consider three different situations. (I) All weights of digraph $\tilde{\mathscr{G}}$ are nonnegative. Under the condition \textbf{C2}, all followers will gradually converge to the leader's state. By Definition~\ref{definition:3.1}, the asynchronous bipartite tracking cannot be achieved. (II) All weights of digraph $\tilde{\mathscr{G}}$ are non-positive. Under the condition \textbf{C2}, the followers who can directly detect the leader's cooperative information will reach the leader's state, and the followers who can indirectly detect the leader's information will not all converge to the state opposite to the leader. By Definition~\ref{definition:3.1}, the asynchronous bipartite tracking also cannot be achieved. (III) The directed graph $\tilde{\mathscr{G}}$ contains both positive and negative weights. Under the condition \textbf{C2}, there may exist two directed paths $W_1$ and $W_2$ from the leader to one of the followers, denoted by $v_i$, where all edges' weights in $W_1$ are positive while one edge of $W_2$ contains negative weight and the other edges contain positive weights. That is, the follower $v_i$ is both cooperative and competitive with the leader. If the intensity of cooperation is the same as the intensity of competition, then follower $v_i$ will not reach the leader's state or the state opposite to the leader. By Definition~\ref{definition:3.1}, the asynchronous bipartite tracking also cannot be achieved. Through the analysis of the above three cases, if the condition \textbf{C1} is not satisfied, then the asynchronous bipartite tracking also cannot be achieved even if the condition \textbf{C2} holds.

2) \emph{Necessity of condition} \textbf{C2}: Under the condition \textbf{C1}, if the condition \textbf{C2} is not satisfied, then at least one follower will not find any state information from the leader, namely, this follower's position is always independent of the leader' position at all discrete-time instants $k\tau$, $k\in\mathbb{N}$. Consequently, it is impossible to achieve the asynchronous bipartite tracking, which is obviously a contradiction.
\end{proof}

Next, we show the dynamics of bipartite tracking for asynchronous first-order MASs through a simulation.

\begin{figure}[t]
  \centering
  \subfigure{
    \includegraphics[width=1.5in]{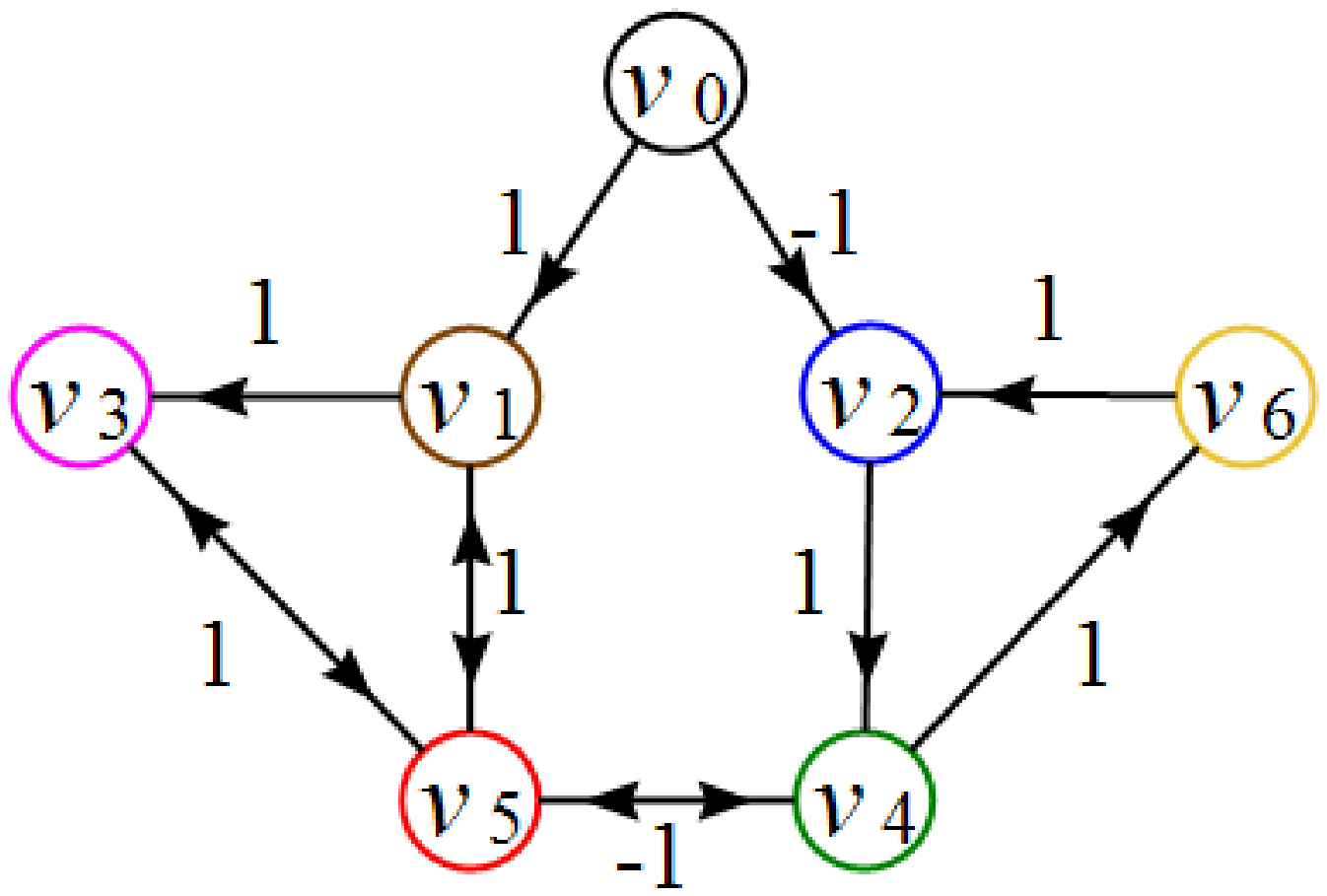}\label{fig1a}}
  \subfigure{
    \includegraphics[width=2.4in]{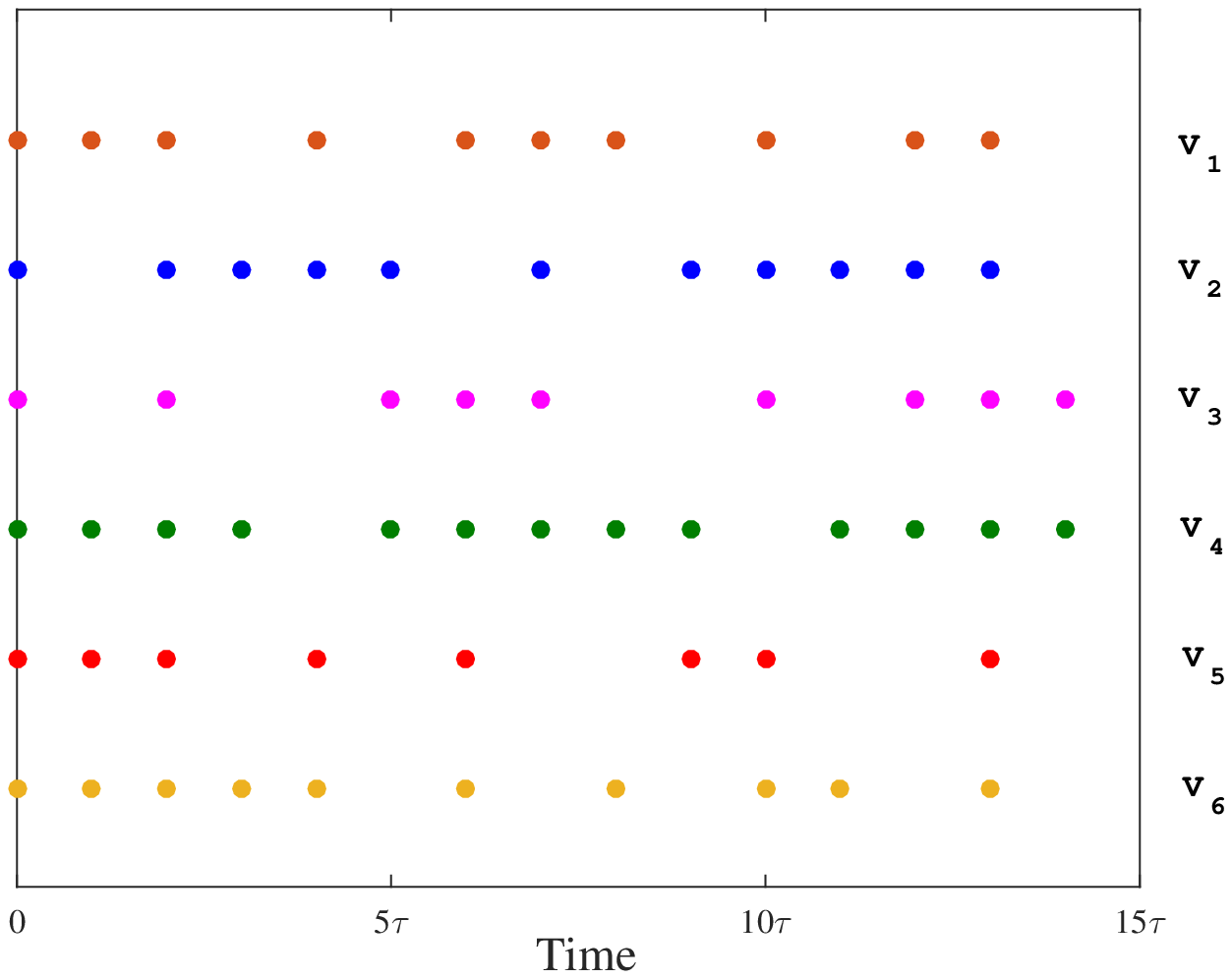}\label{fig1b}}
  \caption{Communication topology $\tilde{\mathscr{G}}_1$ and the asynchronous communication time instants of all the followers.}\label{fig1}
\end{figure}

\begin{figure}[t]
  \centering
    \includegraphics[width=2.4in]{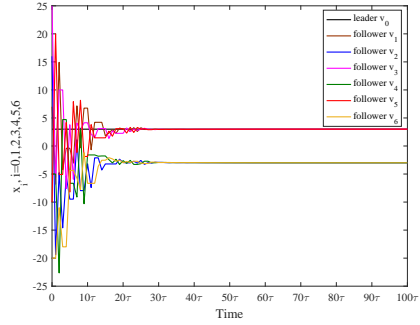}
  \caption{Position trajectories of first-order dynamic agents in Example~\ref{example:3.4}.}\label{fig2}
 \end{figure}

\begin{example}\label{example:3.4}
Consider a set of seven agents, including one leader (labelled by $v_0$) and six followers (labelled $v_1,v_2,\ldots,v_6$). The information exchange among the agents is described by a signed digraph $\tilde{\mathscr{G}}_1$ depicted in Fig.~\ref{fig1a}, from which we see $d_M=2$. Clearly, the digraph $\tilde{\mathscr{G}}_1$ satisfies the conditions \textbf{C1} and \textbf{C2}. Let $h=3$, and then the communication time instants of all followers are presented in Fig.~\ref{fig1b}. Choose $\tau=0.2$ and $\psi=2$ that satisfy the inequality (\ref{sys:3.7}). Finally, the agents' position trajectories are shown in Fig.~\ref{fig2}, from which it is seen that the asynchronous bipartite tracking for first-order MASs is realized.
\end{example}

\section{Bipartite tracking of asynchronous second-order MASs with a static leader}\label{section:4}

In this section, we first transform the asynchronous bipartite tracking issue of second-order MASs with a static leader into a product convergence issue of ISubSM, and then analyze the convergence based on the results of Theorems \ref{theorem:2.2} and \ref{theorem:2.3}.

Consider the discrete-time dynamic model for the leader
\begin{eqnarray}\label{sys:4.1}
x_{0}(k+1)=x_{0}(k),
\end{eqnarray}
and the dynamics for each follower $v_i\in\mathscr{V}$
\begin{eqnarray}\label{sys:4.2}
\begin{aligned}
&x_{i}(k+1)=x_{i}(k)+\tau\vartheta_{i}(k),\\
&\vartheta_{i}(k+1)=\vartheta_{i}(k)+\tau u_{i}(k),
\end{aligned}
\end{eqnarray}
where $\vartheta_{i}(k)\in \mathbb{R}^{p}$ represents agent $v_i$'s velocity at time instant $k\tau$. For each follower $v_i$, the following asynchronous distributed control input $u_{i}(k)$ is designed:
\begin{eqnarray}\label{sys:4.3}
\left\{
\begin{aligned}
u_{i}(k)=&-\gamma \vartheta_{i}(k)+\sum\limits_{v_j\in\mathscr{N}_i}|a_{ij}|\big[\sgn(a_{ij})x_{j}(k)-x_{i}(k)\big]\\
&\hspace{1.5cm}\!+|a_{i0}|\big[\sgn(a_{i0})x_{0}(k)-x_{i}(k)\big],\ {\rm if} \ k\tau\in\{s^{i}_{k}\tau\};\\
u_{i}(k)=&\, 0, \ {\rm if} \ k\tau\notin\{s^{i}_{k}\tau\}.
\end{aligned}
\right.
\end{eqnarray}
Observing (\ref{sys:4.3}), one finds that follower $v_i$ moves at a constant velocity during time intervals $[s^{i}_{k}\tau,s^{i}_{k+1}\tau)$, $k\in\mathbb{N}$, and changes the velocity only at time instants $s^{i}_{k}\tau$, $k\in\mathbb{N}$.

\begin{definition}\label{definition:4.1}
The bipartite tracking for systems (\ref{sys:4.1}) and (\ref{sys:4.2}) is said to be realized if the following conditions are satisfied:
\begin{eqnarray*}
\begin{split}
& \forall v_i\in\mathscr{V}_1,\ \lim_{k\rightarrow \infty}\left\|x_{i}(k)-x_{0}(k)\right\|=0,\\
&\forall v_i\in\mathscr{V}_{2},\ \lim_{k\rightarrow \infty}\left\|x_{i}(k)+x_{0}(k)\right\|=0,\\
&\forall v_i\in\mathscr{V},\ \lim_{k\rightarrow \infty}\left\|\vartheta_{i}(k)\right\|=0.
\end{split}
\end{eqnarray*}
\end{definition}

Based on the introduction of digraphs $\tilde{\mathscr{G}}(k)$, $k\in\mathbb{N}$ in the above section, control input (\ref{sys:4.3}) is equivalently representable as
\begin{equation}\label{sys:4.4}
\begin{aligned}
u_{i}(k)=-\gamma \vartheta_{i}(k)&+\sum\limits_{v_j\in\mathscr{N}_i(k)}|a_{ij}(k)|\big[\sgn(a_{ij}(k))x_{j}(k)-x_{i}(k)\big]\\
&+|a_{i0}(k)|\big[\sgn(a_{i0}(k))x_{0}(k)-x_{i}(k)\big].
\end{aligned}
\end{equation}

Denote
\begin{eqnarray}\label{sys:4.5}
\begin{aligned}
&\vartheta(k)=\left[\vartheta_1^{T}(k),\ldots,\vartheta_m^{T}(k),-\vartheta_{m+1}^{T}(k),\ldots,-\vartheta_n^{T}(k)\right]^{T},\\
&y(k)=\left[e_{x}^{T}(k), e_{x}^{T}(k)+\frac{2}{\gamma}\vartheta^T(k)\right]^T.\\
\end{aligned}
\end{eqnarray}
Using (\ref{sys:4.4}) and (\ref{sys:4.5}), systems (\ref{sys:4.1}) and (\ref{sys:4.2}) can be written as an error system
\begin{eqnarray}\label{sys:4.6}
\begin{aligned}
y(k+1)=\big[C(k)\otimes I_{p}\big]y(k),
\end{aligned}
\end{eqnarray}
where
\begin{eqnarray*}
    C(k)=\left(\begin{array}{@{\hspace{0.1em}}cc@{\hspace{0.1em}}}
                   I_{n}\!-\!\frac{\gamma\tau}{2}I_{n} & \frac{\gamma\tau}{2}I_{n} \\
                    \frac{\gamma\tau}{2}I_{n}\!-\!\frac{2\tau}{\gamma}\big(\mathscr{D}(k)\!+\!\mathscr{B}(k)\!-\!|\mathscr{A}(k)|\big) & I_{n}\!-\!\frac{\gamma\tau}{2}I_{n} \\
                  \end{array}
\right).
\end{eqnarray*}

Obviously, we need to prove that $\lim_{k\rightarrow\infty}\prod_{s=0}^kC(s)=\mathbf{0}$ for achieving the asynchronous bipartite tracking of second-order MASs with a static leader. Below we analyze the matrices $C(k)$, $k\in\mathbb{N}$. If the parameter $\gamma$ satisfies
\begin{eqnarray}\label{sys:4.7}
2\sqrt{d_M}\leq\gamma<\frac{2}{\tau},
\end{eqnarray}
where $d_M$ is defined in (\ref{sys:3.7}), then $C(k)$, $k\in\mathbb{N}$ are sub-stochastic matrices in which the diagonal elements are positive and the row sums satisfy:
\begin{eqnarray*}
\left\{
\begin{split}
&\Lambda_{i}\big[C(k)\big]=1,\\
&\Lambda_{i+n}\big[C(k)\big]<1 \ {\rm if} \ (v_0,v_i)\in\tilde{\mathscr{E}}(k),\\
&\Lambda_{i+n}\big[C(k)\big]=1 \ {\rm if} \ (v_0,v_i)\notin\tilde{\mathscr{E}}(k),
\end{split}
\right.
\end{eqnarray*}
where $i\in\{1,2,\ldots,n\}$. Thus, the asynchronous bipartite tracking issue is equivalently transformed into the product convergence issue of ISubSM.

\begin{lemma}\label{lemma:4.2}
Under the inequality (\ref{sys:4.7}), if the communication topology satisfies the conditions \textbf{C1} and \textbf{C2}, then for any $k\in\mathbb{N}$, there holds that
\begin{eqnarray}\label{sys:4.8}
\Big\|\prod_{s=k}^{k+2Ph-1}C(s)\Big\|_{\infty}<1.
\end{eqnarray}
\end{lemma}

\begin{proof}\
Since the inequality (\ref{sys:4.7}) holds, matrices $C(k)$, $k\in\mathbb{N}$ are sub-stochastic matrices with positive diagonal elements. Known from the conditions \textbf{C1} and \textbf{C2}, there is a directed path $W_z=v_0\rightarrow v_{\delta_{1}}\rightarrow v_{\delta_{2}}\rightarrow\cdots\rightarrow v_{\delta_{z}}$ in $\tilde{\mathscr{G}}$ for each vertex $v_{\delta_{z}}\in\mathscr{V}$, where $\delta_{1},\delta_{2},\ldots,\delta_{z}\in\{1,2,\ldots,n\}$. Assume that $k\tau+l_0\tau$ is a communication instant time of $v_{\delta_{1}}$ during $[k\tau,k\tau+h\tau)$, where $k\tau\leq k\tau+l_0\tau<k\tau+h\tau$. Then we have $(v_0,v_{\delta_{1}})\in\tilde{\mathscr{E}}(k+l_0)$. It follows that $\Lambda_{\delta_{1}+n}[C(k+l_0)]<1$, and further
\begin{eqnarray}\label{sys:4.9}
\begin{aligned}
\Lambda_{\delta_{1}+n}\left[\prod_{s=k}^{k+l_0}C(s)\right]
&=\sum_{j=1}^{2n}\big[C(k+l_0)\big]_{\delta_{1}+n,i}\Lambda_i\left[\prod_{s=k}^{k+l_0-1}C(s)\right]\\
&\leq\Lambda_{\delta_{1}+n}[C(k+l_0)]<1.
\end{aligned}
\end{eqnarray}
Based on (\ref{sys:4.9}) and the fact $[C(k+2h-1)]_{\delta_1,\delta_{1}+n}=\frac{\gamma\tau}{2}>0$, it can be derived based on the result in Theorem~\ref{theorem:2.4} that
\begin{eqnarray}\label{sys:4.10}
\Lambda_{\delta_{1}}\left[\prod_{s=k}^{k+2h-1}C(s)\right]<1, \ \Lambda_{\delta_{1}+n}\left[\prod_{s=k}^{k+2h-1}C(s)\right]<1.
\end{eqnarray}

Consider the edge $(v_{\delta_{1}},v_{\delta_{2}})$. Assume that $k\tau+2h\tau+l_1\tau$ is a communication instant time of $v_{\delta_{2}}$ during the interval $[k\tau+2h\tau,k\tau+2h\tau+h\tau)$, where $k\tau+2h\tau\leq k\tau+2h\tau+l_1\tau<k\tau+2h\tau+h\tau$. Then, we have $(v_{\delta_{1}},v_{\delta_{2}})\in\mathscr{E}(k+2h+l_1)$, which leads to  $[C(k+2h+l_1)]_{\delta_2+n,\delta_1}=\frac{\gamma\tau}{2}-\frac{2\tau}{\gamma}a_{\delta_2\delta_1}>0$. Based on (\ref{sys:4.10}) and the condition $[C(k+2h+l_1)]_{\delta_2+n,\delta_1}>0$, one gets depending on the result of Theorem~\ref{theorem:2.4} that $\Lambda_{\delta_{2}}\big[\prod_{s=k}^{k+2h+l_1}C(s)\big]<1$ and $\Lambda_{\delta_{2}+n}\big[\prod_{s=k}^{k+2h+l_1}C(s)\big]<1$, and further
\begin{eqnarray}\label{sys:4.11}
\Lambda_{\delta_{2}}\left[\prod_{s=k}^{k+4h-1}C(s)\right]<1, \ \Lambda_{\delta_{2}+n}\left[\prod_{s=k}^{k+4h-1}C(s)\right]<1.
\end{eqnarray}

Following the analysis of (\ref{sys:4.10}) and (\ref{sys:4.11}), the following result can be derived
\begin{eqnarray}\label{sys:4.12}
\Lambda_{\delta_{z}}\big[\prod_{s=k}^{k+2zh-1}C(s)\big]<1, \ \Lambda_{\delta_{z}+n}\big[\prod_{s=k}^{k+2zh-1}C(s)\big]<1.
\end{eqnarray}
According to the fact $P\geq d(v_0, v_{u_z})=z$, we can further get
\begin{eqnarray}\label{sys:4.13}
\Lambda_{\delta_{z}}\left[\prod_{s=k}^{k+2Ph-1}C(s)\right]<1, \ \Lambda_{\delta_{z}+n}\left[\prod_{s=k}^{k+2Ph-1}C(s)\right]<1.
\end{eqnarray}
Because of the arbitrariness of $\delta_z$ in the set $\{1,2,\ldots,n\}$, we have
\begin{eqnarray}\label{sys:4.14}
\begin{aligned}
\Big\|\prod_{s=k}^{k+2Ph-1}\!\!\!\!C(s)\Big\|_{\infty}\!\!=\!\!\max_{\delta_z\!=\!1,2,\!\ldots\!,n}\left\{\!\!\Lambda_{\delta_{z}}\!\!
\left[\prod_{s=k}^{k+2Ph-1}\!\!\!\!C(s)\!\right]\!\!,\Lambda_{\delta_{z}+n}\!\!
\left[\prod_{s=k}^{k+2Ph-1}\!\!\!\!C(s)\!\right]\!\right\}\!<\!1.
\end{aligned}
\end{eqnarray}
This completes the proof.
\end{proof}

Through all the above preparations, below we show the major result for the asynchronous bipartite tracking of second-order MASs with a static leader.

\begin{theorem}\label{theorem:4.3}
Suppose that gain parameter $\gamma$ meets the inequality (\ref{sys:4.7}). The asynchronous bipartite tracking for systems (\ref{sys:4.1}) and (\ref{sys:4.2}) with distributed protocol (\ref{sys:4.3}) can be realized if and only if the communication topology satisfies \textbf{C1} and  \textbf{C2}.
\end{theorem}

\begin{proof}
\emph{\underline{Sufficiency}}:\
We first divide the time axis into a series of time intervals $[2\theta Ph\tau,2\theta Ph\tau+2Ph\tau)$, $\theta\in\mathbb{N}$. Denote $C^*(\theta)=\prod_{s=2\theta Ph}^{2\theta Ph+2Ph-1}C(s)$. Then it is known from Lemma~\ref{lemma:4.2} that $\|C^*(\theta)\|_{\infty}<1$, $\theta\in\mathbb{N}$. Thus, we have
\begingroup
\allowdisplaybreaks
\begin{eqnarray}\label{sys:4.15}
\begin{aligned}
\lim_{k\rightarrow\infty}\|y(k+1)\|_{\infty}&=\lim_{k\rightarrow\infty}\left\|\left[\prod_{s=0}^kC(s)\otimes I_p\right]y(0)\right\|_{\infty}\\
&=\lim_{\theta\rightarrow\infty}\left\|\left[\prod_{s=0}^\theta C^*(s)\otimes I_p\right]y(0)\right\|_{\infty}\\
&\leq\lim_{\theta\rightarrow\infty}\prod_{s=0}^\theta\big\|C^*(\theta)\big\|_{\infty}\|y(0)\|_{\infty}=0.
\end{aligned}
\end{eqnarray}
\endgroup
This means that $\lim_{k\rightarrow\infty}e_x(k)=0$ and $\lim_{k\rightarrow\infty}\vartheta(k)=0$, namely,
$\lim_{k\rightarrow \infty}\|x_{i}(k)-x_{0}(k)\|_{\infty}=0$, $i=1,\ldots, m$; $\lim_{k\rightarrow \infty}\|x_{i}(k)+x_{0}(k)\|_{\infty}=0$, $i=m+1,\ldots,n$ and $\lim_{k\rightarrow \infty}\|\vartheta_{i}(k)\|_{\infty}=0$, $i=1,2,\ldots,n$. By Definition~\ref{definition:4.1}, the asynchronous bipartite tracking of second-order MASs with a static leader is achieved.

\emph{\underline{Necessity}}: The proof of the necessity is similar to the proof of Theorem~\ref{theorem:3.3}, so it is omitted here.
\end{proof}

\begin{remark}\label{remark:4.4}
In Theorems~\ref{theorem:3.3} and \ref{theorem:4.3}, the choices of parameters $\psi$ and $\gamma$ are related closely to the maximum degree $d_M$ of the vertices in the communication topology and the update step-size $\tau$. For any given topology, appropriate $\psi$, $\gamma$ and $\tau$ can always be found to ensure that the conditions (\ref{sys:3.7}) and (\ref{sys:4.7}) hold, and then the asynchronous bipartite tracking for first-order MASs and second-order MASs with a static leader can be realized.
\end{remark}

\begin{figure}[t]
  \centering
      \subfigure[position trajectories of second-order dynamic agents]{
    \includegraphics[width=2.4in]{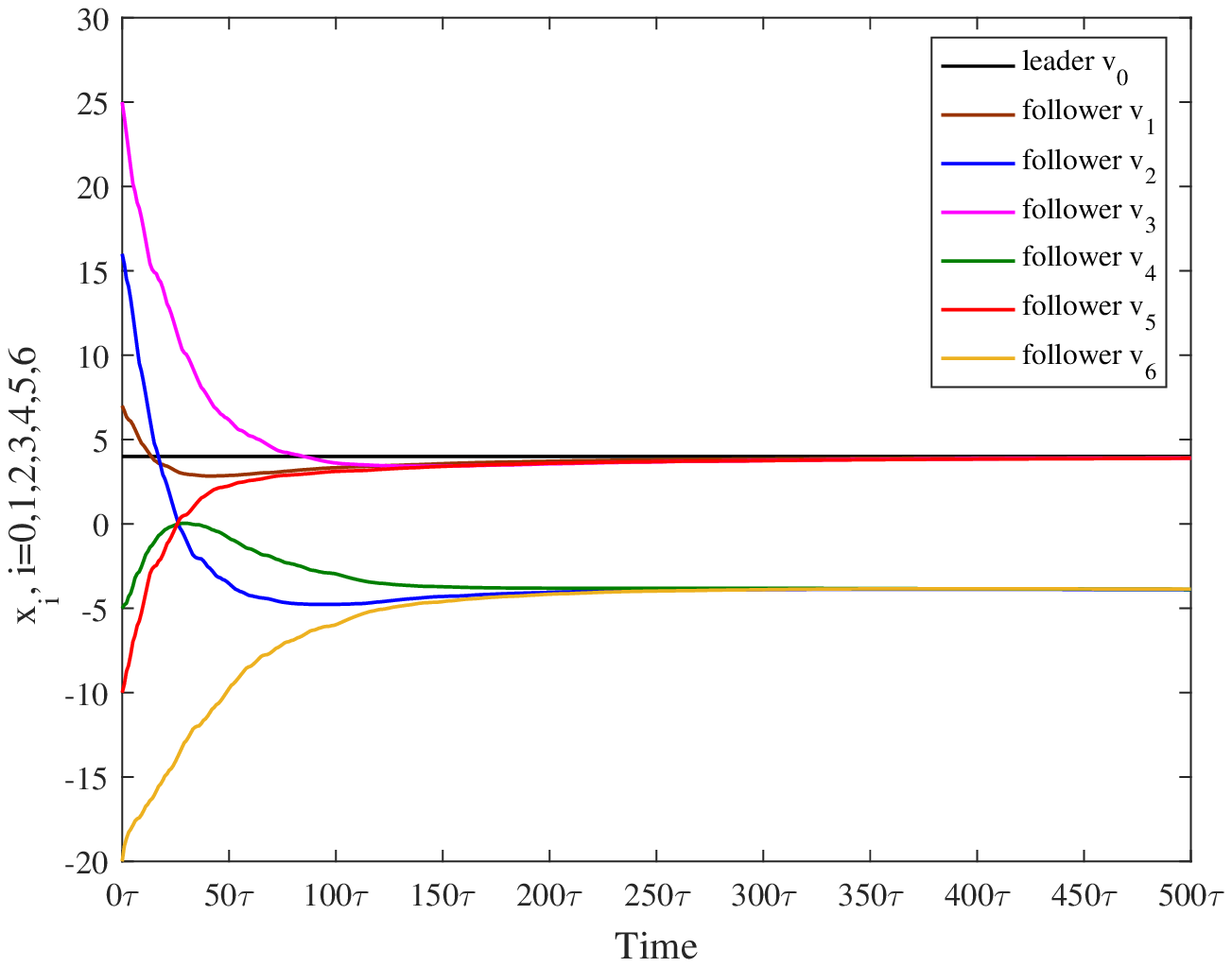}\label{fig3a}}
    \subfigure[velocity trajectories of second-order dynamic agents]{
    \includegraphics[width=2.4in]{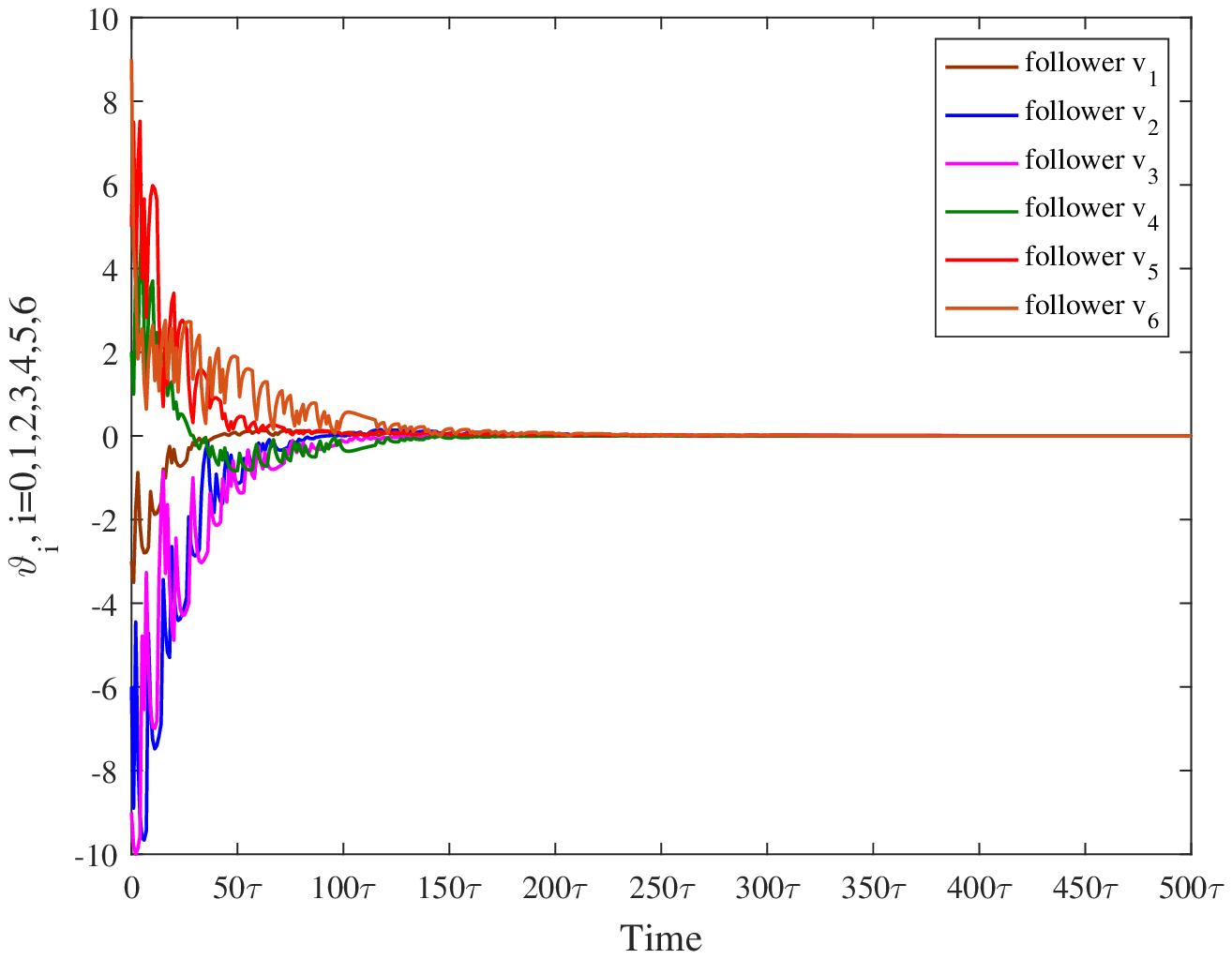}\label{fig3b}}
  \caption{State trajectories of all agents in Example~\ref{example:4.5}. }\label{fig3}
\end{figure}

\begin{example}\label{example:4.5}
Consider a group of agents interacting by the signed network $\tilde{\mathscr{G}}_1$ in Fig.~\ref{fig1a}. The communication time instants of all followers are presented in Fig.~\ref{fig1b}. Select $\tau=0.2$ and $\gamma=4$ that satisfy the inequality (\ref{sys:4.7}). Finally, the agents' position and velocity trajectories are shown in Fig.~\ref{fig3a} and Fig.~\ref{fig3b}, respectively, from which it can be observed that the asynchronous second-order bipartite tracking in MASs with a static leader is realized.
\end{example}

\section{Bipartite tracking of asynchronous second-order MASs with an active leader}\label{section:5}

In this section, we convert the asynchronous bipartite tracking issue of second-order MASs with an active leader to a convergence issue of ISupSM, and then discuss the convergence based on the results obtained in Section \ref{section:2.2}.

The agents in the system have the following estimates of the second-order information states:
\begin{eqnarray}\label{sys:5.1}
x_{0}(k+1)=x_{0}(k)+\tau\vartheta_0,
\end{eqnarray}
and
\begin{eqnarray}\label{sys:5.2}
\begin{aligned}
&x_{i}(k+1)=x_{i}(k)+\tau\vartheta_{i}(k),\\
&\vartheta_{i}(k+1)=\vartheta_{i}(k)+\tau u_{i}(k),
\end{aligned}
\end{eqnarray}
where $i=1,2,\ldots,n$, $v_0$ is the desired velocity of the leader and the asynchronous tracking protocol $u_{i}(k)$ is designed as
\begin{eqnarray}\label{sys:5.3}
\left\{
\begin{aligned}
u_{i}(k)=&\sum\limits_{v_j\in\mathscr{N}_i}|a_{ij}|\big[\sgn(a_{ij})x_{j}(k)-x_{i}(k)\big]\\
&+|a_{i0}|\big[\sgn(a_{i0})x_{0}(k)-x_{i}(k)\big]\\
&+\beta\sum\limits_{v_j\in\mathscr{N}_i}|a_{ij}|\big[\sgn(a_{ij})\vartheta_{j}(k)-\vartheta_{i}(k)\big]\\
&+\beta |a_{i0}|\big[\sgn(a_{i0})\vartheta_{0}\!-\!\vartheta_{i}(k)\big],\ \text{if} \ k\tau\in\{s^{i}_{k}\tau\};\\
u_{i}(k)=&0, \ \text{if} \ k\tau\notin\{s^{i}_{k}\tau\}.
\end{aligned}
\right.
\end{eqnarray}

\begin{definition}\label{definition:5.1}
The bipartite tracking for systems (\ref{sys:5.1}) and (\ref{sys:5.2}) is said to be realized if
\begin{eqnarray*}
\begin{aligned}
\forall v_i\in\mathscr{V}_1,\ & \lim_{k\rightarrow \infty}\left\|x_{i}(k)-x_{0}(k)\right\|=0,\
\lim_{k\rightarrow \infty}\left\|\vartheta_{i}(k)-\vartheta_{0}\right\|=0;\\
\forall v_i\in\mathscr{V}_{2},\ &\lim_{k\rightarrow \infty}\left\|x_{i}(k)+x_{0}(k)\right\|=0,\
\lim_{k\rightarrow \infty}\left\|\vartheta_{i}(k)+\vartheta_{0}\right\|=0.
\end{aligned}
\end{eqnarray*}
\end{definition}

Based on the structure of digraph $\tilde{\mathscr{G}}(k)$, which are described in Section~\ref{section:3}, the update of $u_{i}(k)$ reads as
\begin{eqnarray}\label{sys:5.4}
\begin{aligned}
u_{i}(k)=&\sum\limits_{v_j\in\mathscr{N}_i(k)}|a_{ij}(k)|\big[\sgn(a_{ij}(k))x_{j}(k)-x_{i}(k)\big]\\
&+|a_{i0}(k)|\big[\sgn(a_{i0}(k))x_{0}(k)-x_{i}(k)\big]\\
&+\beta\sum\limits_{v_j\in\mathscr{N}_i(k)}|a_{ij}(k)|\big[\sgn(a_{ij}(k))\vartheta_{j}(k)-\vartheta_{i}(k)\big]\\
&+\beta |a_{i0}(k)|\big[\sgn(a_{i0}(k))\vartheta_{0}-\vartheta_{i}(k)\big].
\end{aligned}
\end{eqnarray}

By constructing the following model transformations:
\begin{eqnarray}\label{sys:5.5}
\begin{aligned}
&e_{\vartheta1}(k)=\left[\vartheta^T_{1}(k)-\vartheta^T_{0},\ldots,\vartheta^T_{m}(k)-\vartheta^T_{0}\right]^T,\\
&e_{\vartheta2}(k)=\left[-\vartheta^T_{m+1}(k)-\vartheta^T_{0},\ldots,-\vartheta^T_{n}(k)-\vartheta^T_{0}\right]^T,\\
&e_{\vartheta}(k)=\left[e_{\vartheta1}^{T}(k),e_{\vartheta2}^{T}(k)\right]^{T}, \\
&\xi(k)=\big[e_{x}^{T}(k), \alpha e_{x}^{T}(k)+\alpha \beta e_{\vartheta}^{T}(k)\big]^{T},
\end{aligned}
\end{eqnarray}
where $\alpha>1$ is a constant, the update of $\xi(k)$ reads as
\begin{eqnarray}\label{sys:5.6}
\begin{aligned}
\xi(k+1)=[H(k)\otimes I_{p}]\xi(k),
\end{aligned}
\end{eqnarray}
where
\begin{eqnarray*}
    H(k)=\left(\begin{array}{c@{\hspace{0.4em}}c}
                   I_{n}-\frac{\tau}{\beta}I_{n} & \frac{\tau}{\alpha\beta}I_{n} \\
                    -\frac{\alpha\tau}{\beta}I_{n} & I_{n}+\frac{\tau}{\beta}I_{n}-\beta \tau \big(\mathscr{D}(k)+\mathscr{B}(k)-|\mathscr{A}(k)|\big) \\
                  \end{array}
\right).
\end{eqnarray*}
One can see that $H(k)$, $k\in\mathbb{N}$ are matrices with negative elements because of the existence of the sub-matrix $-\frac{\alpha\tau}{\beta}I_{n}$.

Apparently, the implementation of asynchronous bipartite tracking for systems (\ref{sys:5.1}) and (\ref{sys:5.2}) is equivalent to the achievement of the asymptotic stability of error system (\ref{sys:5.6}). Therefore, our task below is to prove $\lim_{k\rightarrow\infty}\big\|H(k)H(k-1)\cdots H(0)\big\|_{\infty}=0$. However, it is difficult to directly study this convergence using the nonnegative matrix theory since there are negative elements in $H(k)$. To solve this convergence issue, we construct a new matrix
\begin{eqnarray*}
    \Psi(k)=\left(\begin{array}{c@{\hspace{0.3em}}c}
                   I_{n}-\frac{\tau}{\beta}I_{n} & \frac{\tau}{\alpha\beta}I_{n} \\
                    \frac{\alpha\tau}{\beta}I_{n} & I_{n}+\frac{\tau}{\beta}I_{n}-\beta \tau \big(\mathscr{D}(k)+\mathscr{B}(k)-|\mathscr{A}(k)|\big) \\
                  \end{array}
\right)
\end{eqnarray*}
for each $H(k)$. Before proceeding, we first give the following lemma to present the properties of matrices $\Psi(k)$, $k\in\mathbb{N}$.

\begin{lemma}\label{lemma:5.2}
If the parameter $\beta$ satisfies:
\begin{eqnarray}\label{sys:5.7}
\sqrt{\frac{1+\alpha}{b_m}}<\beta\leq\frac{1}{\tau d_M},
\end{eqnarray}
where $b_m\!=\!\min\{|a_{i0}| \mid |a_{i0}|\!>\!0, i\!=\!1,2,\ldots,n\}$, then $\Psi(k)$, $k\in\mathbb{N}$ are super-stochastic in which the row sums satisfy:
\begin{eqnarray*}
\begin{aligned}
&\Lambda_{i}[\Psi(k)]=1-\frac{\tau}{\beta}+\frac{\tau}{\alpha\beta}<1;\\
&\Lambda_{n+i}[\Psi(k)]=1+\frac{\alpha \tau}{\beta}+\frac{\tau}{\beta}-\beta\tau |a_{i0}|<1, \ \text{if} \ |a_{i0}|>0;\\
&\Lambda_{n+i}[\Psi(k)]=1+\frac{\alpha \tau}{\beta}+\frac{\tau}{\beta}>1, \ \text{if} \ |a_{i0}|=0,
\end{aligned}
\end{eqnarray*}
where $i\in\{1,2,\ldots,n\}$.
\end{lemma}

\begin{proof}
From condition (\ref{sys:5.7}), 
we can obtain that $\tau<\frac{1}{d_M}\sqrt{\frac{b_m}{1+\alpha}}$. Then
\begin{align*}
\beta>\sqrt{\frac{1+\alpha}{b_m}}>\frac{1}{\sqrt{d_M}}>\frac{1}{\sqrt{d_M(1+\alpha)}}
=\frac{1}{d_M}\sqrt{\frac{d_M}{1+\alpha}}\geq\frac{1}{d_M}\sqrt{\frac{b_m}{1+\alpha}}>\tau.
\end{align*}
It follows that $I_{n}-\frac{\tau}{\beta}I_{n}$ is a nonnegative matrix in which the diagonal elements are positive. Under the condition $\frac{1}{\sqrt{d_M}}<\sqrt{\frac{1+\alpha}{b_m}}<\beta\leq\frac{1}{\tau d_M}$, it can be obtained that $I_{n}+\frac{\tau}{\beta}I_{n}-\beta \tau(\mathscr{D}(k)+\mathscr{B}(k)-|\mathscr{A}(k)|)$ is a super-stochastic matrix that contains positive diagonal elements. Besides, $\Lambda_{i}[\Psi(k)]=1-\frac{\tau}{\beta}+\frac{\tau}{\alpha\beta}<1$, $\Lambda_{n+i}[\Psi(k)]=1+\frac{\alpha \tau}{\beta}+\frac{\tau}{\beta}-\beta\tau |a_{i0}|<1$ if $|a_{i0}|>0$, and $\Lambda_{n+i}[\Psi(k)]=1+\frac{\alpha \tau}{\beta}+\frac{\tau}{\beta}>1$ if $|a_{i0}|=0$, where $i\in\{1,2,\ldots,n\}$. Therefore, the results hold.
\end{proof}

Under condition (\ref{sys:5.7}), $\Psi(k)$, $k\in\mathbb{N}$ are super-stochastic matrices, which means that $|H(k)|=\Psi(k)$. Then one gets
\begin{eqnarray}\label{sys:5.8}
\lim_{k\rightarrow\infty}\Big\|\prod_{s=0}^kH(s)\Big\|_{\infty}\leq\lim_{k\rightarrow\infty}\Big\|\prod_{s=0}^k\Psi(s)\Big\|_{\infty}
\end{eqnarray}
according to the properties of infinity norm. Thus, the achievement of asynchronous bipartite tracking is equivalent to having $\lim_{k\rightarrow\infty}\|\prod_{s=0}^k\Psi(s)\|_{\infty}=0$. To solve the product convergence issue of ISupSM, we first construct the following matrix
\begin{eqnarray}\label{sys:5.9}
\check{\mathscr{A}}(k)=\left(
        \begin{array}{cc}
          1 & 0_{1\times2n} \\
          q(k) & \Psi(k) \\
        \end{array}
      \right)
\end{eqnarray}
for each matrix $\Psi(k)$, where $q(k)=[q_1(k),q_2(k),\ldots,q_{2n}(k)]^T$ that satisfies: $q_j(k)=0$ if $\Lambda_j[\Psi(k)]\geq1$, and $q_j(k)=1-\Lambda_j[\Psi(k)]$ if $\Lambda_j[\Psi(k)]<1$. Then we construct a digraph $\check{\mathscr{G}}(k)=(\check{\mathscr{V}},\check{\mathscr{E}}(k),\check{\mathscr{A}}(k))$ with $\check{\mathscr{A}}(k)$ being the adjacency matrix, where $\check{\mathscr{V}}=\{0,1,\ldots,2n\}$ in which the sequential elements are the rows' indexes of matrix $\check{\mathscr{A}}(k)$. According to the structure of digraph $\check{\mathscr{G}}(k)$, we can draw the following conclusions:
\begin{enumerate}
\item [\emph{B1}.] $(0,u+n)\in\check{\mathscr{E}}(k)$ if $( v_0, v_u)\in\tilde{\mathscr{E}}(k)$, where $u\in\{1,2,$ $\ldots,n\}$;

\item [\emph{B2}.] $\big(u+n,u\big)\in\check{\mathscr{E}}(k)$, $u=1,2,\ldots,n$;

\item [\emph{B3}.] There exists a directed path from $u_1+n$ to $u_2+n$ in $\check{\mathscr{G}}(k)$ if a directed path from $v_{u_{1}}$ to $v_{u_{2}}$ exists in $\tilde{\mathscr{G}}(k)$, where $u_{1},u_{2}\in\{1,2,\ldots,n\}$;

\item [\emph{B4}.] Digraph $\check{\mathscr{G}}(k)$ has self-loops on all vertices.
\end{enumerate}

The following two lemmas play crucial roles for deriving the subsequent results.

\begin{lemma}\label{lemma:5.3}
Suppose that the inequality (\ref{sys:5.7}) holds. If the communication topology satisfies the conditions \textbf{C1} and \textbf{C2}, then for any $k\in\mathbb{N}$, the composition $\check{\mathscr{E}}(k)\circ\check{\mathscr{E}}(k+1)\circ\cdots\circ\check{\mathscr{E}}(k+Ph-1)$ associated with the set $\{n+1,n+2,\ldots,2n\}$ is rooted at the vertex $0$, where $P$ is defined in (\ref{sys:3.8}).
\end{lemma}

\begin{proof}
From the conditions \textbf{C1} and \textbf{C2}, there is a directed path $W_{z}=v_0\rightarrow v_{\delta_{1}}\rightarrow v_{\delta_{2}}\rightarrow\cdots\rightarrow v_{\delta_{z}}$ in $\tilde{\mathscr{G}}$ for each vertex $v_{\delta_{z}}$, where $ v_{\delta_{1}}, v_{\delta_{2}},\ldots, v_{\delta_{z}}\in\mathscr{V}$. First, we discuss the edge $( v_0, v_{\delta_{1}})$. Assume that $k\tau+l_{0}\tau$ is a communication time instant of $v_{\delta_{1}}$ during $[k\tau,k\tau+h\tau)$, where $k\tau\leq k\tau+l_{0}\tau< k\tau+h\tau$. By \emph{B1}, we can deduce that $( v_0, v_{\delta_{1}})\in\tilde{\mathscr{E}}(k+l_{0})\Rightarrow(0,\delta_{1}+n)\in\check{\mathscr{E}}(k+l_{0})$. Based on this result and the fact that vertices 0, $\delta_{1}+n$ have self-loops in $\check{\mathscr{G}}(q)$, $q\in\mathbb{N}$, it can be obtained from Theorem~\ref{theorem:2.11} that
\begin{eqnarray}\label{sys:5.10}
(0,\delta_{1}+n)\in\check{\mathscr{E}}(k)\circ\cdots\circ\check{\mathscr{E}}(k+h-1),
\end{eqnarray}
where $(0,0)\in\check{\mathscr{E}}(k),\cdots,(0,0)\in\check{\mathscr{E}}(k+l_{0}-1),(0,\delta_{1}+n)\in\check{\mathscr{E}}(k+l_{0}), (\delta_{1}+n,\delta_{1}+n)\in\check{\mathscr{E}}(k+l_{0}+1),\ldots,(\delta_{1}+n,\delta_{1}+n)\in\check{\mathscr{E}}(k+h-1)$.

Below we discuss the edges $( v_{\delta_{y}}, v_{\delta_{y+1}})$, $y=1,2,\ldots,z-1$. Because each follower communicates with its neighbors at least once in arbitrary time interval $[k\tau,k\tau+h\tau)$ with the length $h\tau$, there exists an instant $k\tau+yh\tau+l_y\tau$ such that $k\tau+yh\tau\leq k\tau+yh\tau+l_y\tau< k\tau+(y+1)h\tau$, and $( v_{\delta_{y}}, v_{\delta_{y+1}})\in\tilde{\mathscr{E}}(k+yh+l_y)$. By \emph{B3}, we have $(\delta_y+n,\delta_{y+1}+n)\in\check{\mathscr{E}}(k+yh+l_y)$. By means of this result and the fact that vertices $\delta_y+n$, $\delta_{y+1}+n$ have self-loops in $\check{\mathscr{G}}(q)$, $q\in\mathbb{N}$, it follows that
\begin{eqnarray}\label{sys:5.11}
(\delta_{y}+n,\delta_{y+1}+n)\in\check{\mathscr{E}}(k+yh)\circ\cdots \circ\check{\mathscr{E}}(k+yh+h-1),
\end{eqnarray}
in which $(\delta_{y}+n,\delta_{y}+n)\in\check{\mathscr{E}}(k+yh),\ldots,
(\delta_{y}+n,\delta_{y}+n)\in\check{\mathscr{E}}(k+yh+l_y-1),
(\delta_{y}+n,\delta_{y+1}+n)\in\check{\mathscr{E}}(k+yh+l_y),
(\delta_{y+1}+n,\delta_{y+1}+n)\in\check{\mathscr{E}}(k+yh+l_y+1),\ldots,
(\delta_{y+1}+n,\delta_{y+1}+n)\in\check{\mathscr{E}}(k+yh+h-1)$.

Based on (\ref{sys:5.10}) and (\ref{sys:5.11}), we can deduce that
\begin{eqnarray}\label{sys:5.12}
\begin{aligned}
&(0,\delta_{1}+n)\in\check{\mathscr{E}}(k)\circ\cdots \circ\check{\mathscr{E}}(k+h-1),\\
&(\delta_{1}+n,\delta_{2}+n)\in\check{\mathscr{E}}(k+h)\circ\cdots \circ\check{\mathscr{E}}(k+2h-1),\\
& \ \ \ \ \ \ \ \ \vdots\\
&(\delta_{z-1}+n,\delta_{z}+n)\in\check{\mathscr{E}}(k+zh-h)\circ\cdots\circ\check{\mathscr{E}}(k+zh-1).
\end{aligned}
\end{eqnarray}
This further results in
\begin{eqnarray}\label{sys:5.13}
\begin{aligned}
(0,\delta_{z}+n)\in\check{\mathscr{E}}(k)\circ\cdots\circ\check{\mathscr{E}}(k+zh-1).
\end{aligned}
\end{eqnarray}
Since $P$ is the farthest distance from the leader to the followers, we have $P\geq d(v_0, v_{u_z})=z$. According to \emph{B4}, the vertex $\delta_{z}+n$ has a self-loop in $\check{\mathscr{G}}(q), q\in\mathbb{N}$, which leads to
\begin{eqnarray}\label{sys:5.14}
\begin{aligned}
(\delta_{z}+n,\delta_{z}+n)\in\check{\mathscr{E}}(k+zh)\circ\cdots\circ\check{\mathscr{E}}(k+Ph-1).
\end{aligned}
\end{eqnarray}
Combine (\ref{sys:5.13}) and (\ref{sys:5.14}), one has
\begin{eqnarray}\label{sys:5.15}
\begin{aligned}
(0,\delta_{z}+n)\in\check{\mathscr{E}}(k)\circ\cdots\circ\check{\mathscr{E}}(k+Ph-1).
\end{aligned}
\end{eqnarray}
Because of the arbitrary selectivity of $\delta_{z}+n$ in the set $\{n+1,n+2,\ldots,2n\}$, we get that the composition $\check{\mathscr{E}}(k)\circ\check{\mathscr{E}}(k+1)\circ\cdots\circ\check{\mathscr{E}}(k+Ph-1)$ associated with the set $\{n+1,n+2,\ldots,2n\}$ is rooted at the vertex $0$. The proof is completed.
\end{proof}

\begin{lemma}\label{lemma:5.4}
Under the inequality (\ref{sys:5.7}), the result
\begin{eqnarray}\label{sys:5.16}
\begin{aligned}
\Big\|\prod_{s=k}^{k+Ph-1}\Psi(s)\Big\|_{\infty}<1
\end{aligned}
\end{eqnarray}
holds for any $k\in\mathbb{N}$ if the following conditions are met:
\begin{enumerate}[\!\!\!i)]
\item The composition $\check{\mathscr{E}}(k)\circ\check{\mathscr{E}}(k+1)\circ\cdots\circ\check{\mathscr{E}}(k+Ph-1)$ associated with the set $\{n+1,n+2,\ldots,2n\}$ is rooted at the vertex $0$;

\item The gain parameter $\beta$ satisfies
\begin{eqnarray}\label{sys:5.17}
\begin{aligned}
\beta>\frac{\alpha^{\frac{Ph}{Ph-1}}-1}{\tau b_m\varphi^{Ph-1}},
\end{aligned}
\end{eqnarray}
where $\alpha>(1+\frac{\tau}{\beta}+\frac{\alpha \tau}{\beta})^{Ph-1}$ and $\varphi=\min\big\{[\Phi^*]_{ij}\mid [\Phi^*]_{ij}>0, \ i,j=1,2,\ldots,n\big\}$ with $\Phi^*=I_{n}+\frac{\tau}{\beta}I_{n}-\beta \tau(\mathscr{D}+\mathscr{B}-|\mathscr{A}|)$.
\end{enumerate}
\end{lemma}

\begin{proof}
Firstly, we discuss the sum of the $j$th row of $\prod_{s=k}^{k+Ph-1}\Psi(s)$, where $j\in\{n+1,n+2,\ldots,2n\}$. Since the composition $\check{\mathscr{E}}(k)\circ\check{\mathscr{E}}(k+1)\circ\cdots\circ\check{\mathscr{E}}(k+Ph-1)$ associated with the set $\{n+1,n+2,\ldots,2n\}$ is rooted at the vertex $0$, there exist vertices $j_1,j_2,\ldots,j_{Ph}\in\{n+1,n+2,\ldots,2n\}$ such that $(0,j_{1})\in\check{\mathscr{E}}(k), (j_{1},j_{2})\in\check{\mathscr{E}}(k+1),\ldots,(j_{Ph-1},j_{Ph})\in\check{\mathscr{E}}(k+Ph-1)$, where $j_{Ph}=j$. With no loss of generality, assume that $j_{f}\neq0$ and $j_{f-1}=j_{f-2}=\cdots=j_{1}=0$, where $f\in\{1,2,\ldots,Ph-1\}$. Then, we have
\begin{eqnarray}\label{sys:5.18}
\begin{aligned}
\Lambda_{j_{f}}\big[\Psi(k+f-1)\big]\leq1+\frac{\tau}{\beta}+\frac{\alpha \tau}{\beta}-\beta \tau b_m<1,
\end{aligned}
\end{eqnarray}
and
\begin{eqnarray}\label{sys:5.19}
\begin{aligned}
&\Lambda_{i}\big[\Psi(k+f-1)\big]\leq 1+\frac{\tau}{\beta}+\frac{\alpha \tau}{\beta},\\
&\Lambda_{i}\big[\Psi(k+f-2)\big]\leq 1+\frac{\tau}{\beta}+\frac{\alpha \tau}{\beta},\\
& \ \vdots\\
&\Lambda_{i}\big[\Psi(k)\big]\leq 1+\frac{\tau}{\beta}+\frac{\alpha \tau}{\beta}, \ i=n+1,\ldots,2n.
\end{aligned}
\end{eqnarray}
For convenience, let $c=1+\frac{\tau}{\beta}+\frac{\alpha \tau}{\beta}-\beta\tau b_m$ and $g=1+\frac{\tau}{\beta}+\frac{\alpha \tau}{\beta}$. It can be seen that $c<1$ and $g>1$ under condition (\ref{sys:5.7}). Thus, according to Theorem~\ref{theorem:2.8}, the row sum of the $j_{f+1}$th row of matrix $\prod_{s=k+f-1}^{k+f}\Psi(s)$ satisfies:
\begin{eqnarray}\label{sys:5.20}
\begin{aligned}
\Lambda_{j_{f+1}}\left[\prod_{s=k+f-1}^{k+f}\Psi(s)\right]\leq g^{2}-(g-c)\varphi,
\end{aligned}
\end{eqnarray}
and further we arrive at
\begin{eqnarray}\label{sys:5.21}
\begin{aligned}
\Lambda_{j_{f+2}}\left[\prod_{s=k+f-1}^{k+f+1}\Psi(s)\right]\leq g^{3}-(g-c)\varphi^2.
\end{aligned}
\end{eqnarray}
In the same way, it can be obtained that
\begin{eqnarray}\label{sys:5.22}
\begin{aligned}
\Lambda_{j_{Ph}}\left[\prod_{s=k+f-1}^{k+Ph-1}\Psi(s)\right]\leq g^{Ph-f+1}-(g-c)\varphi^{Ph-f}.
\end{aligned}
\end{eqnarray}
It thus follows that
\begin{eqnarray}\label{sys:5.23}
\begin{aligned}
\Lambda_{j_{Ph}}\left[\prod_{s=k}^{k+Ph-1}\Psi(s)\right]&=\!\sum_{i=1}^{2n}\left[\prod_{s=k+f-1}^{k+Ph-1}\Psi(s)\right]_{j_{Ph}i}
\Lambda_{i}\left[\prod_{s=k}^{k+f-2}\Psi(s)\right]\\
&\leq g^{f-1}\Lambda_{j_{Ph}}\left[\prod_{s=k+f-1}^{k+Ph-1}\Psi(s)\right]\\
&\leq g^{Ph}-(g-c)\varphi^{Ph-1}<1.
\end{aligned}
\end{eqnarray}
This implies that $\Lambda_{j}\big[\prod_{s=k}^{k+Ph-1}\Psi(s)\big]<1$ for any $j\in\{n+1,n+2,\ldots,2n\}$.

Next, we consider the sum of the $s$th row of $\prod_{s=k}^{k+Ph-1}\Psi(s)$, where $s\!\in\!\{1,2,\ldots,n\}$. According to the given condition $\alpha>(1+\frac{\tau}{\beta}+\frac{\alpha \tau}{\beta})^{Ph-1}=g^{Ph-1}$ and the fact $g>1$,
we can derive that $\alpha>g$, $\alpha>g^2$, $\ldots$, $\alpha>g^{Ph-1}$. Based on these results, we have $1-\frac{\tau}{\beta}+\frac{ \tau}{\alpha\beta}g<1$, $1-\frac{\tau}{\beta}+\frac{ \tau}{\alpha\beta}g^2<1$, $\ldots$, $1-\frac{\tau}{\beta}+\frac{ \tau}{\alpha\beta}g^{Ph-1}<1$. It thus follows that
\begin{eqnarray}\label{sys:5.24}
\begin{aligned}
\Lambda_{s}\left[\prod_{s=k}^{k+1}\Psi(s)\right]&=\sum_{i=1}^{n}\big[\Psi(k+1)\big]_{si}\Lambda_{i}\big[\Psi(k)\big]
+\sum_{i=n+1}^{2n}\big[\Psi(k+1)\big]_{si}\Lambda_{i}\big[\Psi(k)\big] \\
&\leq \sum_{i=1}^{n}\big[\Psi(k+1)\big]_{si}+\sum_{i=n+1}^{2n}\big[\Psi(k+1)\big]_{si}g \\
&\leq 1-\frac{\tau}{\beta}+\frac{ \tau}{\alpha\beta}g<1.
\end{aligned}
\end{eqnarray}
Furthermore,
\begin{eqnarray}\label{sys:5.25}
\begin{aligned}
\Lambda_{s}\left[\prod_{s=k}^{k+2}\Psi(s)\right]&=\sum_{i=1}^{n}\big[\Psi(k\!+\!2)\big]_{si}\Lambda_{i}\left[\prod_{s=k}^{k+1}\Psi(s)\right]
+\!\sum_{i=n+1}^{2n}\big[\Psi(k\!+\!2)\big]_{si}\Lambda_{i}\left[\prod_{s=k}^{k+1}\Psi(s)\right] \\
&\leq \sum_{i=1}^{n}\big[\Psi(k\!+\!2)\big]_{si}\!+\!\sum_{i=n+1}^{2n}\big[\Psi(k\!+\!2)\big]_{si}\sum_{i_1=1}^{2n}\big[\Psi(k\!+\!1)\big]_{ii_1}g\\
&\leq \sum_{i=1}^{n}\big[\Psi(k\!+\!2)\big]_{si}\!+\!\sum_{i=n+1}^{2n}\big[\Psi(k\!+\!2)\big]_{si}g^2\\
&\leq 1-\frac{\tau}{\beta}+\frac{ \tau}{\alpha\beta}g^2<1.
\end{aligned}
\end{eqnarray}
In the same way, one can obtain that
\begin{align}\label{sys:5.26}
\Lambda_{s}\left[\prod_{s=k}^{k+Ph-1}\Psi(s)\right]\leq 1\!-\!\frac{\tau}{\beta}\!+\!\frac{\tau}{\alpha\beta}g^{Ph-1}<1.
\end{align}
This implies that $\Lambda_{s}\big[\prod_{s=k}^{k+Ph-1}\Psi(s)\big]<1$ for any $s\in\{1,2,\ldots,n\}$.

Summarizing the above analysis, we arrive at $\big\|\prod_{s=k}^{k+Ph-1}\Psi(s)\big\|_{\infty}<1$ for any $k\in\mathbb{N}$. This completes the proof.
\end{proof}

With all the above preparations, we are now ready to give a major result for the asynchronous bipartite tracking of MASs with an active leader.

\begin{theorem}\label{theorem:5.5}
Suppose that the gain parameter $\beta$ satisfies conditions (\ref{sys:5.7}) and (\ref{sys:5.8}). The asynchronous bipartite tracking for systems (\ref{sys:5.1}) and (\ref{sys:5.2}) with distributed protocol (\ref{sys:5.3}) can be implemented if and only if the topology graph meets the conditions \textbf{C1} and \textbf{C2}.
\end{theorem}

\begin{proof}\emph{\underline{Sufficiency}}:
We first divide the time axis into a series of time intervals $[\theta Ph\tau,\theta Ph\tau+Ph\tau)$, $\theta\in\mathbb{N}$. Let
\begin{eqnarray}\label{sys:5.27}
\begin{aligned}
H^*(\theta)=\prod_{s=\theta Ph}^{\theta Ph+Ph-1}H(s), \ \Psi^*(\theta)=\prod_{s=\theta Ph}^{\theta Ph+Ph-1}\Psi(s).
\end{aligned}
\end{eqnarray}
Then by Lemma~\ref{lemma:5.4}, we have $\big\|\Psi^*(\theta)\|_{\infty}<1$, $\theta\in\mathbb{N}$ under conditions (\ref{sys:5.7}) and (\ref{sys:5.8}). Equivalently, from system (\ref{sys:5.6}) it can be obtained that
\begin{eqnarray}\label{sys:5.28}
\begin{aligned}
\lim_{k\rightarrow\infty}\|\xi(k+1)\|_{\infty}
&=\lim_{k\rightarrow\infty}\left\|\left[\prod_{s=0}^kH(s)\otimes I_{p}\right]\xi(0)\right\|_{\infty}\\
&=\lim_{\theta\rightarrow\infty}\left\|\left[\prod_{s=0}^\theta H^*(s)\otimes I_{p}\right]\xi(0)\right\|_{\infty}\\
&\leq\lim_{\theta\rightarrow\infty}\prod_{s=0}^\theta\big\|H^*(s)\big\|_{\infty}\big\|\xi(0)\big\|_{\infty}\\
&\leq\lim_{\theta\rightarrow\infty}\prod_{s=0}^\theta\big\|\Psi^*(s)\big\|_{\infty}\big\|\xi(0)\big\|_{\infty}=0,
\end{aligned}
\end{eqnarray}
which implies that $\lim_{k\rightarrow\infty}e_{x}(k)=0$ and $\lim_{k\rightarrow\infty}e_{\vartheta}(k)=0$. Therefore, we have $\lim_{k\rightarrow \infty}x_{i}(k)=x_{0}(k)$, $\lim_{k\rightarrow \infty}\vartheta_{i}(k)=\vartheta_{0}$, $i=1,2,\ldots,m$ and
$\lim_{k\rightarrow\infty}x_{i}(k)=-x_{0}(k)$, $\lim_{k\rightarrow \infty}\vartheta_{i}(k)=-\vartheta_{0}$, $i=m+1,m+2,\ldots,n$. By Definition~\ref{definition:5.1}, the asynchronous bipartite tracking with an active leader can be achieved.

\emph{\underline{Necessity}}: The proof of the necessity is similar to the proof of Theorem~\ref{theorem:3.3}, so it is omitted here.
\end{proof}

\begin{remark}\label{remark:5.6}
In Theorem~\ref{theorem:5.5}, the choice of the parameter $\beta$ is related closely to the edge weighs of communication topology, the maximum distance from the leader to the followers, the upper bound of the length of asynchronous update intervals and the update step-size $\tau$. For any given communication topology and the designed asynchronous communication rules, $d_M$, $b_m$, $P$, $\varphi$ and $h$ are determined. Then we can always find a proper constant $\alpha$ in the model transformations and the update step-size $\tau$ to ensure that the parameter $\beta$ is existent under conditions (\ref{sys:5.7}) and (\ref{sys:5.8}).
\end{remark}

\begin{figure}[t]
  \centering
    \includegraphics[width=1.7in]{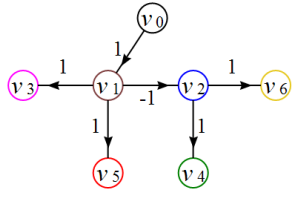}
  \caption{Communication topology $\tilde{\mathscr{G}}_2$.}\label{fig4}
 \end{figure}

 \begin{figure}[t]
  \centering
    \subfigure[position trajectories of the agents]{
    \includegraphics[width=2.4in]{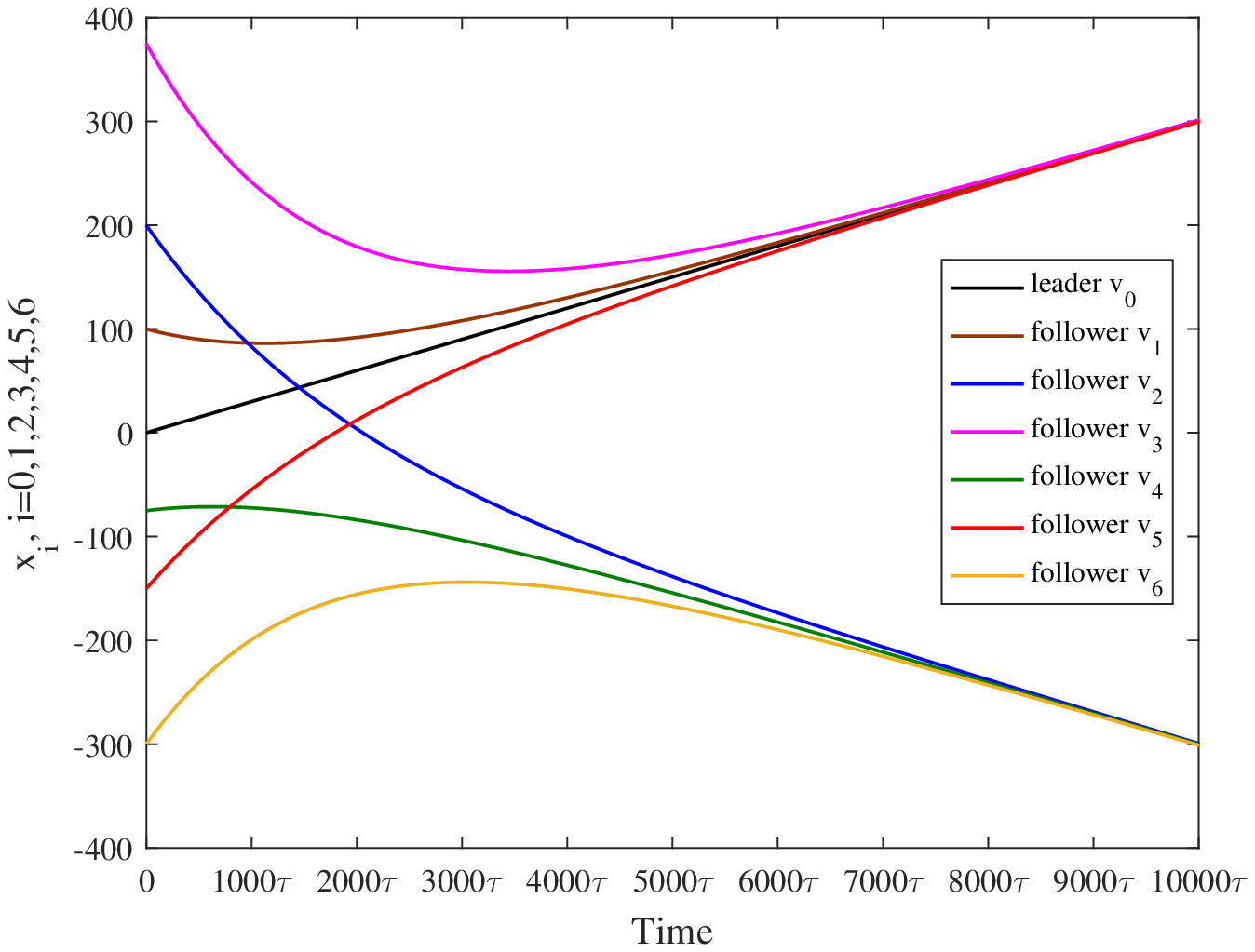}\label{fig5a}}
      \subfigure[velocity trajectories of the agents]{
    \includegraphics[width=2.4in]{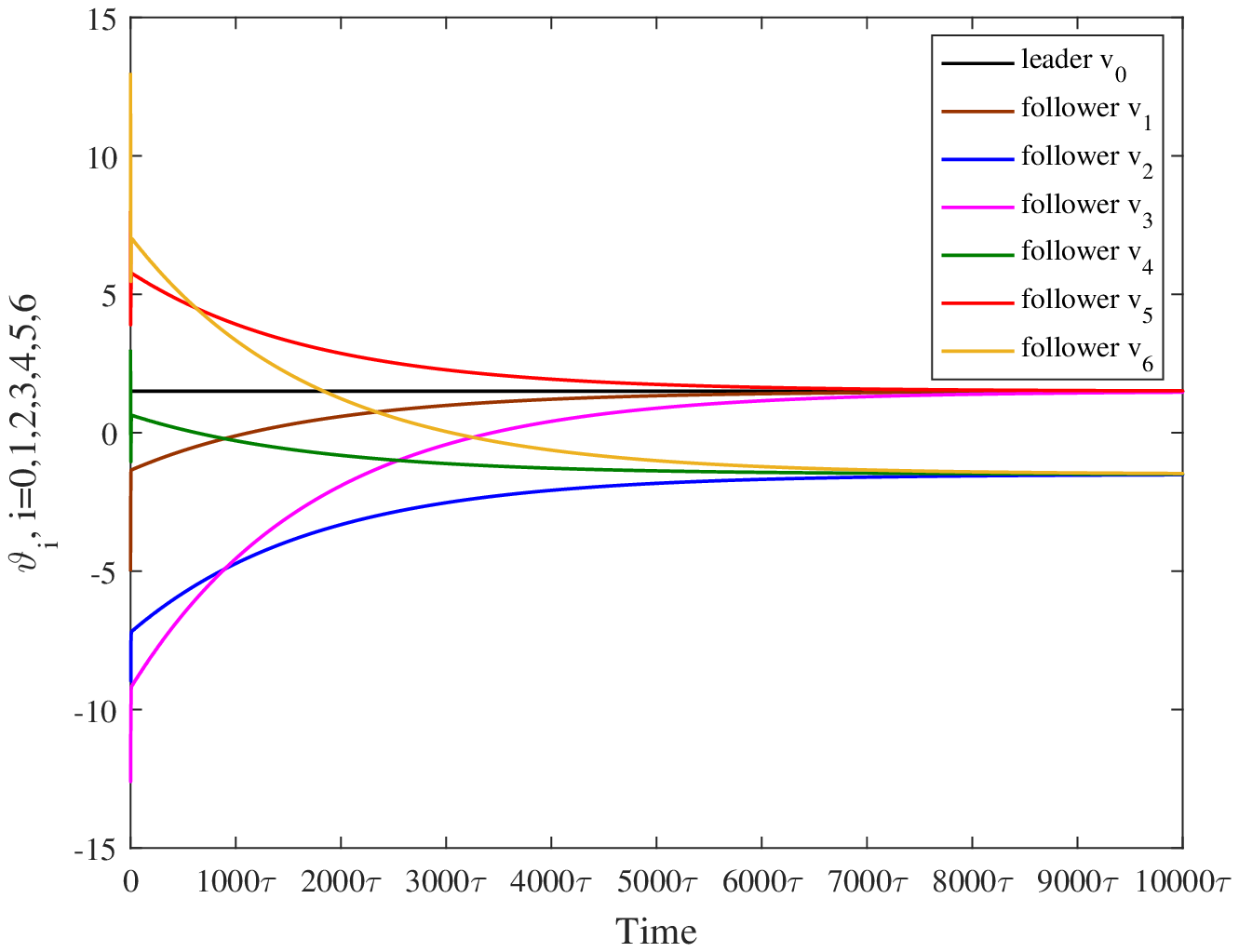}\label{fig5b}}
  \caption{State trajectories of all agents in Example~\ref{example:5.7}. }\label{fig5}
\end{figure}

\begin{example}\label{example:5.7}
Consider the asynchronous bipartite tracking control for second-order MASs with an active leader. The information exchange among one leader (labelled by $v_0$) and six followers (labelled $v_1,v_2,\ldots,v_6$) is described by a signed digraph $\tilde{\mathscr{G}}_2$ depicted in Fig.~\ref{fig4}, which results in $d_M=1$, $b_m=1$ and $P=3$. It is clear that the digraph $\tilde{\mathscr{G}}_2$ satisfies the conditions \textbf{C1} and \textbf{C2}. The update time instants of all followers meet the condition $t_{k+1}^i-t_{k}^i\leq h=2$ for any $k\in\mathbb{N}$ and $i=1,2,\ldots,n$. Choose $\tau=0.02$, $\beta=35$ and $\alpha=1.01$ that satisfy the inequalities (\ref{sys:5.7}) and (\ref{sys:5.8}). Finally, the agents' position and velocity trajectories are displayed in Fig.~\ref{fig5a} and Fig.~\ref{fig5b}, respectively, which indicate that the asynchronous bipartite tracking for second-order MASs with an active leader is realized.
\end{example}

\section{Bipartite tracking of asynchronous general linear MASs}\label{section:6}

In this section, the bipartite tracking behavior of general linear MASs is analyzed. In the discrete-time setting, the state evolution of agents can be expressed as
\begin{eqnarray}\label{sys:6.1}
\begin{aligned}
&x_{0}(k+1)=Ax_{0}(k),\\
&x_{i}(k+1)=Ax_{i}(k)+Bu_{i}(k), \ i=1,\ldots,n, \\
\end{aligned}
\end{eqnarray}
where $x_{i}(k)=[x^{(1)}_{i}(k),x^{(2)}_{i}(k),\cdots,x^{(p)}_{i}(k)]^T\in\mathbb{R}^{p}$ is the state of agent $v_i$ at time instant $k\tau$, $A\in\mathbb{R}^{p\times p}$ and $B\in\mathbb{R}^{p\times q}$ denote the system matrix and input matrix, respectively. The asynchronous control protocol $u_{i}(k)$ can be designed as
\begin{equation}\label{sys:6.2}
\left\{
\begin{aligned}
u_{i}(k)=&\, K\sum_{v_j\in\mathscr{N}_i(k)}|a_{ij}|\left[\sgn(a_{ij})x_{j}(k)-x_{i}(k)\right]\\
&+K|b_i|\left[\sgn(b_i)x_{0}(k)-x_{i}(k)\right], \ {\rm if} \ k\tau\in\{s^{i}_{k}\tau\};\\
u_{i}(k)=&\, 0, \ {\rm if} \ k\tau\notin\{s^{i}_{k}\tau\},
\end{aligned}
\right.
\end{equation}
where $K$ is the gain matrix.

By substituting (\ref{sys:6.2}) into system (\ref{sys:6.1}), the following error system can be obtained
\begin{eqnarray}\label{sys:6.3}
\begin{aligned}
e(k+1)=\big[I_n\otimes A-(\mathscr{D}+\mathscr{B}-|\mathscr{A}|)\otimes BK\big]e(k).
\end{aligned}
\end{eqnarray}

Before proceeding, the following lemma need to be introduced.

\begin{lemma}[\cite{Shi2019Containment}]\label{lemma:6.1}\
Consider a real matrix $S\in \mathbb{R}^{n\times n}$ in which the spectral radius is expressed as $\rho$. There exists $z\geq0$ so that  $\|S^k\|_{\infty}\leq z k^{n-1}\rho^k$ for every $k\geq n$.
\end{lemma}

In the next Theorem \ref{theorem:6.1}, we examine the stability of error system (\ref{sys:6.3}) in detail and present a sufficient condition.

\begin{theorem}\label{theorem:6.1}
Let matrix $B$ be of full row rank. If the communication topology $\tilde{\mathscr{G}}$ meets the conditions \textbf{C1} and \textbf{C2}, then there is a feedback matrix $K\!=\!\psi^* B^T(BB^T)^{-1}A$ with
\begin{align}\label{sys:6.4}
\psi^*<\frac{1}{d_{M}}
\end{align}
such that the bipartite tracking for system (\ref{sys:6.1}) can be realized, where the system matrix $A$ is allowed to be strictly unstable and its spectral radius satisfies:
\begin{eqnarray}\label{sys:6.5}
\begin{aligned}
\rho(A)<\frac{1}{\sqrt[Ph]{1-(1-\zeta)\kappa^{Ph-1}}},
\end{aligned}
\end{eqnarray}
in which
\begin{align*}
\zeta&=\max\big\{\Lambda_i\big[Q\big]\mid \Lambda_i\big[Q\big]<1,\, i=1,\ldots,n\big\},\\
\kappa&=\min\big\{\big[Q\big]_{ij} \mid \big[Q\big]_{ij}>0, i,j=1,\ldots,n, \ i\neq j\big\}.
\end{align*}
where $Q=I_n\!-\!\psi^*\mathscr{D}\!-\!\psi^*\mathscr{B}\!+\!\psi^*|\mathscr{A}|$.
\end{theorem}

\begin{proof}
Since matrix $B$ is of full row rank and $K=\psi^* B^T(BB^T)^{-1}A$, system (\ref{sys:6.3}) can be equivalently reexpressed as
\begin{eqnarray}\label{sys:6.6}
\begin{aligned}
e(k+1)=\big[Q(k)\otimes A\big]e(k),
\end{aligned}
\end{eqnarray}
where $Q(k)=I_n\!-\!\psi^*\mathscr{D}(k)\!-\!\psi^*\mathscr{B}(k)\!+\!\psi^*|\mathscr{A}(k)|$ is a sub-stochastic matrix with positive diagonal elements. Apparently, the bipartite tracking for system (\ref{sys:6.1}) can be implemented if error system (\ref{sys:6.6}) has asymptotic stability. Divide the time axis into ordered intervals $[\theta Ph\tau,\theta Ph\tau+Ph\tau)$, $\theta\in\mathbb{N}$. Let $Q^*(\theta)=\prod_{s=\theta Ph}^{\theta Ph+Ph-1}Q(s)$. Similar to the analysis in Lemma~\ref{lemma:3.2}, we can deduce that $\|Q^*(\theta)\|_{\infty}\leq1-(1-\zeta)\kappa^{Ph-1}<1$.
Thus, we have
\begin{eqnarray}\label{sys:6.7}
\begin{aligned}
\lim_{k\rightarrow\infty}\!\left\|\prod_{s=0}^{k}Q(s)\right\|_{\infty}&\leq\lim_{\theta\rightarrow\infty}\prod_{s=0}^{\theta}
\left\|Q^*(s)\right\|_{\infty}\\
&\leq\lim_{\theta\rightarrow\infty}\big[1-(1-\zeta)\kappa^{Ph-1}\big]^\theta.
\end{aligned}
\end{eqnarray}

Now we proceed to prove the following equality firstly for guaranteeing the stability of error system (\ref{sys:6.6}):
\begin{eqnarray}\label{sys:6.8}
\begin{aligned}
\lim_{k\rightarrow\infty}\left\|\prod_{s=0}^{k}Q(s)\otimes A^k\right\|_{\infty}=0.
\end{aligned}
\end{eqnarray}
By Lemma~\ref{lemma:6.1}, one knows that there exists $z>0$ such that $\|A^k\|_{\infty}\leq z k^{p-1}\rho^{k}(A)$. Thus,
\begin{eqnarray}\label{sys:6.9}
\begin{aligned}
\lim_{k\rightarrow\infty}\left\|A^k\right\|_{\infty}=\lim_{\theta\rightarrow\infty}\left\|(A^{Ph})^\theta\right\|_{\infty}=\lim_{\theta\rightarrow\infty}z (Ph\theta)^{ph-1}\rho^{Ph\theta}(A).
\end{aligned}
\end{eqnarray}
Based on (\ref{sys:6.7}) and (\ref{sys:6.9}), to prove (\ref{sys:6.8}), it suffices to prove that
\begin{eqnarray}\label{sys:6.10}
\begin{aligned}
\lim_{c\rightarrow\infty}z (Ph\theta)^{ph-1}\big[\big(1-(1-\zeta)\kappa^{Ph-1}\big)\rho^{Ph}(A)\big]^\theta=0.
\end{aligned}
\end{eqnarray}
By (\ref{sys:6.4}), we get
\begin{align*}
0<\big(1-(1-\zeta)\kappa^{Ph-1}\big)\rho^{Ph}(A)<1.
\end{align*}
According to the fact that the exponential decay dominates the polynomial increase, the expression (\ref{sys:6.8}) holds. This implies that the bipartite tracking can be achieved.
\end{proof}

\begin{remark}\label{remark:6.3}
In the general linear MASs considered in this section, the input matrix $B$ needs to be row full rank in order to ensure that the whole system can still converge when the system matrix $A$ is strictly unstable. This means that the general linear MASs in this section does not contain the second-order MASs in Section \ref{section:4} and Section \ref{section:5} as special cases, in which $B=[0,1]^T$ is not row full rank. Therefore, second-order MASs and general linear MASs need to be discussed respectively for the integrity of study content.
\end{remark}

\begin{figure}[t]
  \centering
    \subfigure[position trajectories of the agents]{
    \includegraphics[width=2.4in]{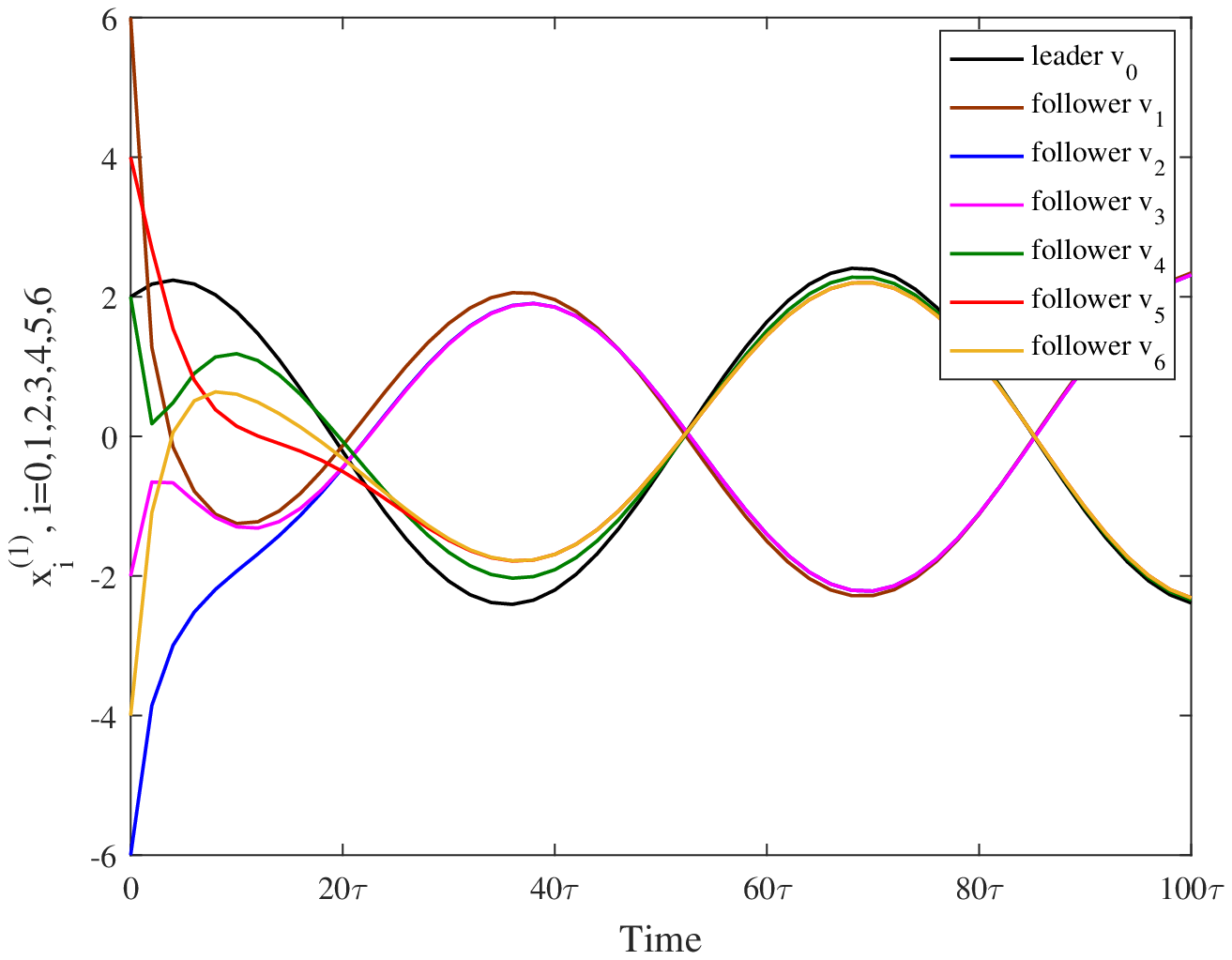}\label{fig6a}}
      \subfigure[velocity trajectories of the agents]{
    \includegraphics[width=2.4in]{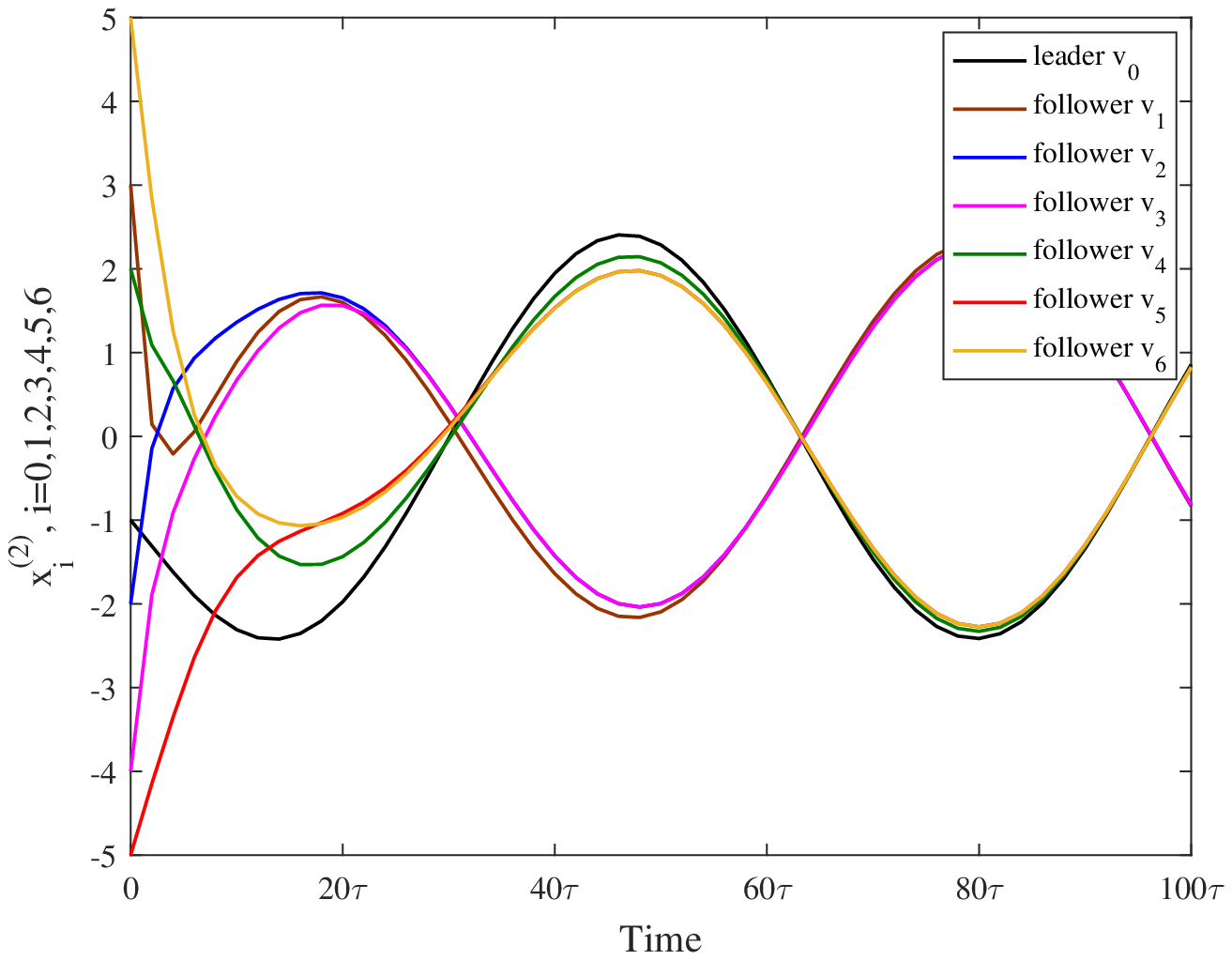}\label{fig6b}}
    \subfigure[acceleration trajectories of the agents]{
    \includegraphics[width=2.4in]{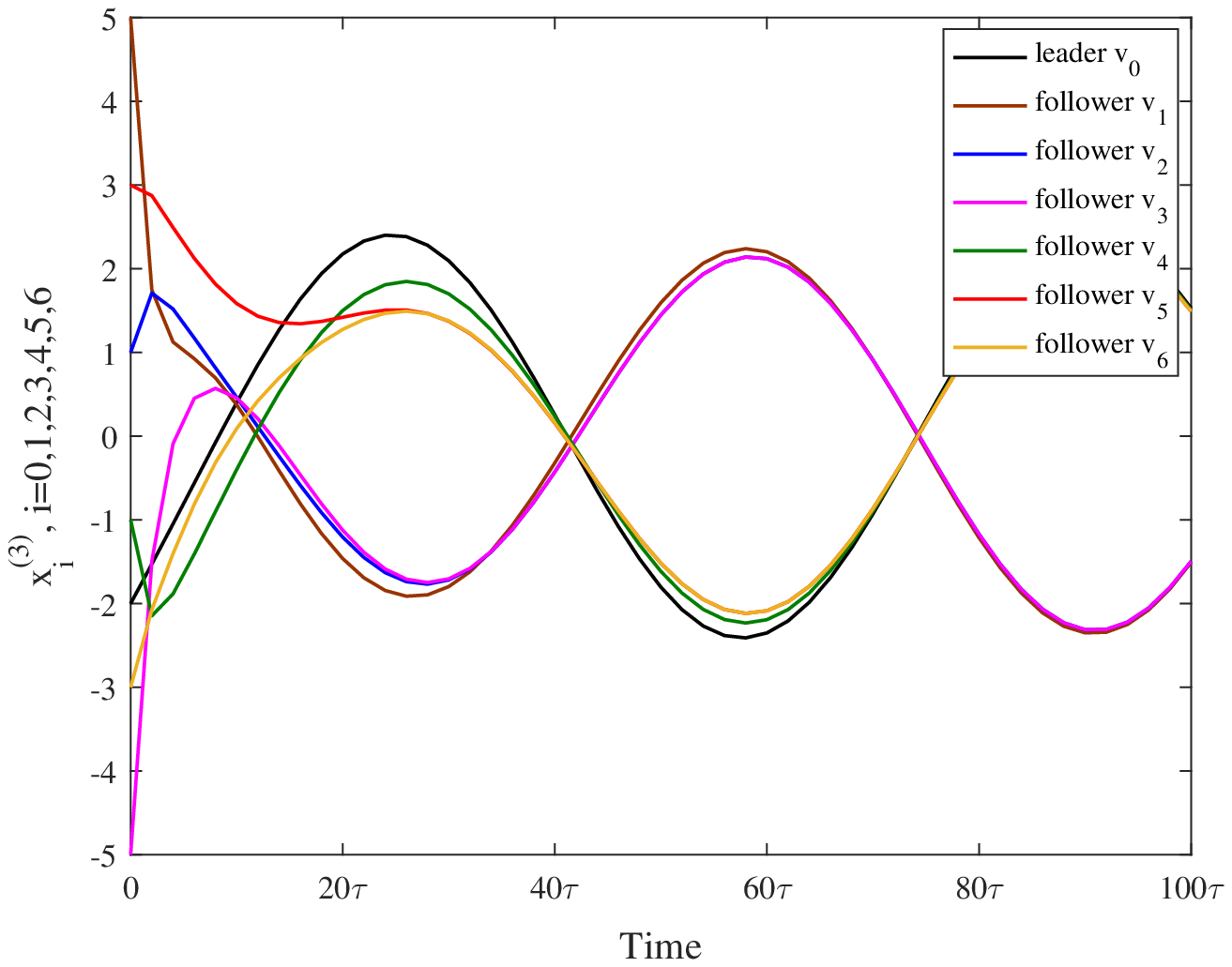}\label{fig6c}}
  \caption{State trajectories of all agents in Example \ref{example:6.4}. }\label{fig6}
\end{figure}

\begin{example}\label{example:6.4}
Consider a third-order multi-agent network on the communication topology $\tilde{\mathscr{G}}_2$ shown in Fig.~\ref{fig4}. The longest distance from the leader to the followers is $P=2$. The asynchronous communication is presented in Fig.~\ref{fig1b}, where $h=3$. Let $\psi^*=\frac{1}{3}$ that satisfies (\ref{sys:6.5}), then we have $\zeta=\frac{2}{3}$ and $\kappa=\frac{1}{3}$. The input matrix $B$ and the system matrix $A$ are expressed respectively by
\begin{align*}
B\!=\!\!\left(
    \begin{array}{@{\hspace{0em}}c@{\hspace{0.6em}}c@{\hspace{0.6em}}c@{\hspace{0.6em}}c@{\hspace{0em}}}
      1 & 0 & 0 & 1 \\
      0 & 1 & 0 & 0 \\
      0 & 0 & 1 & 0 \\
    \end{array}
    \right)\!,\
A\!=\!\!\left(
    \begin{array}{@{\hspace{0em}}r@{\hspace{0.6em}}r@{\hspace{0.6em}}r@{\hspace{0em}}}
       0.8730  &       0 &  -0.2182\\
   -0.2182 &   0.8730   &      0\\
         0  & -0.2182  &  0.8730\\
    \end{array}
    \right)\!.
\end{align*}
Obviously, $B$ is of full row rank and $A$ is strictly unstable. Besides, the spectral radius of $A$ satisfies
\begin{align*}
\rho(A)=1.0001<1.0002=\frac{1}{\sqrt[Ph]{1-(1-\zeta)\kappa^{Ph-1}}}.
\end{align*}
Finally, the agents' state trajectories are exhibited in Fig.~\ref{fig6}, from which we can observe that the bipartite tracking is realized.
\end{example}

\begin{remark}\label{remark:6.5}
In the existing literature, researchers have mainly focused on the investigations of bipartite consensus \cite{Meng2016Interval,Zhao2017Adaptive,Guo2018Bipartite,Zhang2016Bipartite,Qin2016On} and bipartite tracking \cite{Wen2018Bipartite,Ma2018Necessary,Liu2018Robust,Yu2018Prescribed} under the synchronous setting. In general, the synchronous setting is a very special case of the asynchronous setting. Based on this consideration, this paper studies the asynchronous bipartite tracking issue of discrete-time MASs for the first time, and establishes some necessary and sufficient conditions. Therefore, the results obtained in this paper are very important extensions to the study of bipartite consensus and bipartite tracking in MASs.
\end{remark}

\begin{remark}\label{remark:6.6}
In this paper, based on some suitable model transformations, the asynchronous bipartite tracking issues are converted into the convergence issues of ISubSM and ISupSM. According to the definitions of sub-stochastic matrix and super-stochastic matrix, it is known that the existing methods used to study the product of row-stochastic matrices cannot be utilized to handle these two convergence issues. Due to this reason, some new techniques, including the properties of matrix product and the compositions of directed edge sets, are developed to study these two convergence issues in detail in this paper.
\end{remark}

\section*{Part II}

In Section \ref{section:7}-Section \ref{section:13} below, the product properties of ISubSM and ISupSM will be utilized to establish the algebraic conditions of realizing bipartite tracing or bipartite bounded tracking in different practical environments, such as time delays, switching topologies, random networks, lossy links, matrix disturbance, external noise disturbance, and a leader of unmeasurable velocity and acceleration.

\section{Bipartite tracking of second-order MASs with time delays}\label{section:7}

In general, unmodelled delay effects in the feedback mechanism may destabilize the original stable network system. In multi-agent networks, due to the mobility of network agents, the congestion of communication channels and the physical characteristics of the information transmitted by the media, time delays may naturally occur. Based on this consideration, we discuss the bipartite tracking problem of MASs with time delays through the product convergence of ISubSM in this paper.

Consider second-order MASs (\ref{sys:4.1}) and (\ref{sys:4.2}). The distributed control input $u_{i}(k)$ with time delays is given by
\begin{eqnarray}\label{sys:7.1}
\begin{aligned}
u_{i}(k)=-\gamma \vartheta_{i}(k)&+\sum\limits_{v_j\in\mathscr{N}_i}|a_{ij}|\big[\sgn(a_{ij})x_{j}\big(k-\sigma_{ij}(k)\big)-x_{i}(k)\big]\\
&+|a_{i0}|\big[\sgn(a_{i0})x_{0}\big(k-\sigma_{i0}(k)\big)-x_{i}(k)\big],
\end{aligned}
\end{eqnarray}
where $\gamma>0$ is a gain parameter, $\sigma_{ij}(k)\leq\sigma_{max}$ is the delay when information is transmitted from agent $v_j$ to agent $v_i$ at time instant $k\tau$, and $\sigma_{max}\in\mathbb{N}$ is the upper bound of time delays.

Using protocol (\ref{sys:7.1}), systems (\ref{sys:4.1}) and (\ref{sys:4.2}) can be written as the following error form:
\begin{eqnarray}\label{sys:7.2}
\begin{aligned}
y(k+1)=\big[D^*\otimes I_{p}\big]y(k)+\sum_{s=0}^{\sigma_{max}}D_s(k)y(k-s),
\end{aligned}
\end{eqnarray}
where $y(k)$ is defined in (\ref{sys:4.5}), and
\begin{eqnarray*}
   & D^*=\left(\begin{array}{@{\hspace{0.1em}}cc@{\hspace{0.1em}}}
                   I_{n}\!-\!\frac{\gamma\tau}{2}I_{n} & \frac{\gamma\tau}{2}I_{n} \\
                    \frac{\gamma\tau}{2}I_{n}\!-\!\frac{2\tau}{\gamma}\big(\mathscr{D}(k)\!+\!\mathscr{B}(k)\big) & I_{n}\!-\!\frac{\gamma\tau}{2}I_{n} \\
                  \end{array}
\right),\\
&D_s(k)=\left(\begin{array}{@{\hspace{0.1em}}cc@{\hspace{0.1em}}}
                   \mathbf{0} & \mathbf{0} \\
                    \frac{2\tau}{\gamma}|\mathscr{A}_s(k)| & \mathbf{0} \\
                  \end{array}
\right), \ s=0,1,\ldots,\sigma_{max},
\end{eqnarray*}
in which $\mathscr{A}_s(k)$, $s=0,1,\ldots,\sigma_{max}$ are nonnegative matrices satisfying: 1) if $(v_j,v_i)\in\tilde{\mathscr{E}}$ and $\sigma_{ij}(k)=s'\in\{0,1,\ldots,\sigma_{max}\}$, then $[\mathscr{A}_s'(k)]_{ij}=[\mathscr{A}]_{ij}$ and $[\mathscr{A}_s(k)]_{ij}=0$, $s\neq s'$; 2) if $(v_j,v_i)\notin\tilde{\mathscr{E}}$, then $[\mathscr{A}_s(k)]_{ij}=0$, $h=0,1,\ldots,\sigma_{max}$. According to the descriptions of the nonnegative matrices $\mathscr{A}_s(k)$, $s=0,1,\ldots,\sigma_{max}$, the following result holds
\begin{equation*}
\begin{aligned}
\sum_{s=0}^{\sigma_{max}}\mathscr{A}_s(k)=\mathscr{A}.
\end{aligned}
\end{equation*}
In order to analyze the stability of error system (\ref{sys:7.2}), we need further model transformation. Denote
\begin{eqnarray*}
\begin{aligned}
Y(k)=\big[Y^T(k),Y^T(k-1),\ldots,Y^T(k-\sigma_{max})\big]^T.
\end{aligned}
\end{eqnarray*}
Then (\ref{sys:7.2}) can be expressed as
\begin{equation}\label{sys:7.3}
\begin{aligned}
Y(k+1)=E(k)Y(k),
\end{aligned}
\end{equation}
where
\begin{equation*}
\begin{aligned}
E(k)=\left(
       \begin{array}{ccccc}
         D^*+D_0(k) & D_1(k) & \ldots & D_{\sigma_{max}-1}(k) & D_{\sigma_{max}}(k) \\
         I_{2n} & \mathbf{0} & \ldots & \mathbf{0} & \mathbf{0} \\
         \mathbf{0} & I_{2n} & \ldots & \mathbf{0} & \mathbf{0} \\
         \vdots & \vdots & \ddots & \vdots & \vdots \\
         \mathbf{0} & \mathbf{0} & \ldots & \mathbf{0} & \mathbf{0} \\
         \mathbf{0} & \mathbf{0} & \ldots & I_{2n} & \mathbf{0} \\
       \end{array}
     \right).
\end{aligned}
\end{equation*}

If the parameter $\gamma$ satisfies condition (\ref{sys:4.7}), then $E(k)$, $k\in\mathbb{N}$ are sub-stochastic matrices in which $[E(k)]_{ii}>0$, $i=1,2,\ldots,2n$ and the row sums satisfy
\begin{eqnarray*}
\left\{
\begin{aligned}
&\Lambda_{i}\big[E(k)\big]=1, \\
&\Lambda_{i+sn}\big[E(k)\big]=1, s=2,3,\ldots,2\sigma_{max}+1,\\
&\Lambda_{i+n}\big[E(k)\big]<1 \ {\rm if} \ (v_0,v_i)\in\tilde{\mathscr{E}},\\
&\Lambda_{i+n}\big[E(k)\big]=1,  \ {\rm if} \ (v_0,v_i)\notin\tilde{\mathscr{E}},
\end{aligned}
\right.
\end{eqnarray*}
where $i\in\{1,2,\ldots,n\}$. Thus, the bipartite tracking issue is equivalently transformed into the product convergence issue of ISubSM. In order to analyze the convergence of ISubSM, we construct the following matrix
\begin{eqnarray}\label{sys:7.4}
\hat{\mathscr{A}}(k)=\left(
        \begin{array}{cc}
          1 & \mathbf{0} \\
          z(k) & E(k) \\
        \end{array}
      \right),
\end{eqnarray}
for each matrix $E(k)$, where $\sigma^*=\sigma_{max}+1$ and $z(k)=[z_1(k),z_2(k),\ldots,z_{2\sigma^*n}(k)]^T$ in which $z_i(k)=1-\Lambda_i[E(k)]$. Obviously, $\hat{\mathscr{A}}(k)$ is a row stochastic matrix. Using $\hat{\mathscr{A}}(k)$ as the adjacent matrix, we construct a corresponding digraph $\hat{\mathscr{G}}(k)=\big(\hat{ \mathscr{V}},\hat{\mathscr{E}}(k),\hat{\mathscr{A}}(k)\big)$, in which $\hat{\mathscr{V}}=\{0,1,2,\cdots,2\sigma^*n\}$ whose elements are the indexes of the rows of $E(k)$ except the first one which is labeled as 0. Some results concerning the digraphs $\tilde{\mathscr{G}}$ and $\hat{\mathscr{G}}(k)$ are presented for later use.
\begin{enumerate}
\item[1)] $(0,u+n)\in\hat{\mathscr{E}}(k)$ if $(v_{0},v_{u})\in\tilde{\mathscr{E}}$;
\item[2)] $(u+n,u)\in\hat{\mathscr{E}}(k)$, $(u,u+n)\in\hat{\mathscr{E}}(k)$, $u=1,2,\ldots,n$;
\item[3)] $(u+(2s-2)n,u+2sn)\in\hat{\mathscr{E}}(k)$, $(u+(2s-1)n,u+(2s+1)n)\in\hat{\mathscr{E}}(k)$ for any $u=1,2,\ldots,n$ and $s=1,2,\ldots,\sigma^*-1$;
\item[4)] The nodes $0,1,\cdots,2n$ have self-loops in $\hat{\mathscr{G}}(k)$.
\end{enumerate}

Based on the construction of digraph $\hat{\mathscr{G}}(k)$, the following conclusion can be obtained.

\begin{lemma}\label{lemma:7.1}
Suppose that the parameter $\gamma$ satisfies condition (\ref{sys:4.7}). If the communication topology satisfies the conditions \textbf{C1} and \textbf{C2}, then for any $k\in\mathbb{N}$, the composition $\hat{\mathscr{E}}(k)\circ\hat{\mathscr{E}}(k+1)\circ\cdots\circ\hat{\mathscr{E}}(k+2P\sigma^*-1)$ associated with the set $\{1,2,\ldots,2\sigma^*n\}$ is rooted at the vertex $0$.
\end{lemma}

\begin{proof}
By \textbf{C1} and \textbf{C2}, there exists a path $v_0\rightarrow v_{u_{1}}\rightarrow v_{u_{2}}\rightarrow\cdots\rightarrow v_{u_{z}}$ in $\tilde{\mathscr{G}}$ for each vertex $v_{u_{z}}$. First of all, it can be derived that $(0,u_1+n)\in\hat{\mathscr{E}}(k)$ because $(v_0,v_{u_{1}})\in\tilde{\mathscr{E}}$, which together with the fact that $(u_1+n,u_1)\in\hat{\mathscr{E}}(k+\sigma^*-1)$, guarantees that
\begin{equation}\label{sys:7.5}
\begin{aligned}
(0,u_1)\in\hat{\mathscr{E}}(k)\circ\cdots\circ\hat{\mathscr{E}}(k+\sigma^*-1),
\end{aligned}
\end{equation}
where $(0,0)\in\hat{\mathscr{E}}(k),\ldots,(0,0)\in\hat{\mathscr{E}}(k+\sigma^*-3),(0,u_1)\in\hat{\mathscr{E}}(k+\sigma^*-2),
(u_1+n,u_1)\in\hat{\mathscr{E}}(k+\sigma^*-1)$.

For convenience, the time delay of edge $(v_{u_1},v_{u_2})$ at time $(k+2\sigma^*-1)\tau$ is denoted by $\tilde{\sigma}$, that is,  $\sigma_{u_2u_1}(k+2\sigma^*-1)=\tilde{\sigma}$. Then we can deduce the following result
\begin{equation}\label{sys:7.6}
\begin{aligned}
(u_1,u_2)\in\hat{\mathscr{E}}(k+\sigma^*)\circ\cdots\circ\hat{\mathscr{E}}(k+3\sigma^*-1),
\end{aligned}
\end{equation}
where $(u_1,u_1)\in\hat{\mathscr{E}}(k\!+\!\sigma^*),\ldots,(u_1,u_1)\in\hat{\mathscr{E}}\big(k\!+\!2\sigma^*\!-\!2\!-\!\tilde{\sigma}\big), (u_1,u_1+2n)\in\hat{\mathscr{E}}\big(k\!+\!2\sigma^*\!-\!1\!-\!\tilde{\sigma}\big),
(u_1+2n,u_1+4n)\in\hat{\mathscr{E}}\big(k\!+\!2\sigma^*\!-\!\tilde{\sigma}\big),\ldots,
(u_1+2(\tilde{\sigma}-1)n,u_1+2\tilde{\sigma}n)
\in\hat{\mathscr{E}}\big(k\!+\!2\sigma^*\!-\!2\big),(u_1+2\tilde{\sigma}n,u_2+n)\in\hat{\mathscr{E}}\big(k\!+\!2\sigma^*\!-\!1\big),
(u_2+n,u_2)\in\hat{\mathscr{E}}\big(k\!+\!2\sigma^*\big),(u_2,u_2)\in\hat{\mathscr{E}}\big(k\!+\!2\sigma^*\big),\ldots,
(u_2,u_2)\in\hat{\mathscr{E}}\big(k\!+\!3\sigma^*-1\big)$.

Using the method of analyzing edge $(v_{u_1},v_{u_2})$ to analyze the edges $(v_{u_2},v_{u_3}),\ldots$, $(v_{u_{z-1}},v_{u_z})$, we can get
\begin{equation}\label{sys:7.7}
\begin{aligned}
&(u_2,u_3)\in\hat{\mathscr{E}}(k+3\sigma^*)\circ\cdots\circ\hat{\mathscr{E}}(k+5\sigma^*-1),\\
&(u_3,u_4)\in\hat{\mathscr{E}}(k+5\sigma^*)\circ\cdots\circ\hat{\mathscr{E}}(k+7\sigma^*-1),\\
& \ \ \ \ \vdots\\
&(u_{z-1},u_z)\in\hat{\mathscr{E}}(k+(2z-3)\sigma^*)\circ\cdots\circ\hat{\mathscr{E}}(k+(2z-1)\sigma^*-1).
\end{aligned}
\end{equation}
Combining (\ref{sys:7.5}), (\ref{sys:7.6}) and (\ref{sys:7.7}), we have
\begin{equation}\label{sys:7.8}
\begin{aligned}
(0,u_z)\in\hat{\mathscr{E}}(k)\circ\cdots\circ\hat{\mathscr{E}}(k+(2z-1)\sigma^*-1).
\end{aligned}
\end{equation}

Since $P\geq z$ and the fact that vertex $u_1$ has self-loop, where $P$ is the longest distance from the leader to the followers and $z$ is the distance from the leader to follower $v_{u_z}$, we can further derive that
\begin{equation}\label{sys:7.9}
\begin{aligned}
(0,u_z)\in\hat{\mathscr{E}}(k)\circ\cdots\circ\hat{\mathscr{E}}(k+(2P-1)\sigma^*-1).
\end{aligned}
\end{equation}
And since $(u_z,u_z+n)\in\hat{\mathscr{E}}(k+(2P-1)\sigma^*)$, we have
\begin{equation}\label{sys:7.10}
\begin{aligned}
(0,u_z)\in\hat{\mathscr{E}}(k)\circ\cdots\circ\hat{\mathscr{E}}(k+(2P-1)\sigma^*),\\
(0,u_z+n)\in\hat{\mathscr{E}}(k)\circ\cdots\circ\hat{\mathscr{E}}(k+(2P-1)\sigma^*),
\end{aligned}
\end{equation}
which together with the facts $(u_z+(2s-2)n,u_z+2sn)\in\hat{\mathscr{E}}(k)$ and $(u_z+(2s-1)n,u_z+(2s+1)n)\in\hat{\mathscr{E}}(k)$ for any $s=1,2,\ldots,\sigma^*-1$ and $k\in\mathbb{N}$, ensures that
\begin{equation}\label{sys:7.11}
\begin{aligned}
&(0,u_z+2sn)\in\hat{\mathscr{E}}(k)\circ\cdots\circ\hat{\mathscr{E}}(k+2P\sigma^*-1),\\
&(0,u_z+(2s+1)n)\in\hat{\mathscr{E}}(k)\circ\cdots\circ\hat{\mathscr{E}}(k+2P\sigma^*-1),
\end{aligned}
\end{equation}
where $s=0,1,\ldots,\sigma^*-1$. This also implies that the composition $\hat{\mathscr{E}}(k)\circ\hat{\mathscr{E}}(k+1)\circ\cdots\circ\hat{\mathscr{E}}(k+2P\sigma^*-1)$ associated with the set $\{1,2,\ldots,2\sigma^*n\}$ is rooted at the vertex $0$.
\end{proof}

Based on the result of Lemma \ref{lemma:7.1}, we now present a necessary and sufficient condition for bipartite tracking of MASs with time delays.

\begin{theorem}\label{theorem:7.2}
Suppose that gain parameter $\gamma$ meets condition (\ref{sys:4.7}). The bipartite tracking for systems (\ref{sys:4.1}) and (\ref{sys:4.2}) with distributed protocol (\ref{sys:7.1}) can be realized if and only if the communication topology satisfies \textbf{C1} and  \textbf{C2}.
\end{theorem}

\begin{proof}
\emph{\underline{Sufficiency}}:\
Divide the time axis into a series of intervals $[2\theta P\sigma^*\tau, 2(\theta+1)P\sigma^*h\tau)$, $\theta\in\mathbb{N}$. Denote
\[E^*(\theta)=\prod_{s=2\theta P\sigma^*}^{2(\theta+1)P\sigma^*h-1}E(s).\]
Now we analyze $\|E^*(\theta)\|_{\infty}$. By Lemma \ref{lemma:7.1}, for any $u\in\{1,2,\ldots,2\sigma^*n\}$, there holds $(0,r)\in\hat{\mathscr{E}}(2\theta P\sigma^*)\circ\cdots\circ\hat{\mathscr{E}}(2(\theta+1)P\sigma^*h-1)$. Without loss of generality, assume that $(0,r_1)\in\hat{\mathscr{E}}(2\theta P\sigma^*),(r_1,r_2)\in\hat{\mathscr{E}}(2\theta P\sigma^*+1),\ldots,(r_{2P\sigma^*-1},r_{2P\sigma^*})\in\hat{\mathscr{E}}(2(\theta+1)P\sigma^*h-1)$, where $r_1,r_2,\ldots,r_{2P\sigma^*}\in\{1,2,\ldots,2\sigma^*n\}$ and $r_{2P\sigma^*}=u$. Firstly, we can get by $(0,u_1)\in\hat{\mathscr{E}}(2\theta P\sigma^*)$ that
\begin{equation}\label{sys:7.12}
\begin{aligned}
\Lambda_{r_1}\big[E(2\theta P\sigma^*)\big]<1,
\end{aligned}
\end{equation}
which combines the condition $(r_1,r_2)\in\hat{\mathscr{E}}(2\theta P\sigma^*+1)$, guarantees that
\begin{equation}\label{sys:7.13}
\begin{aligned}
\Lambda_{r_2}\left[\prod_{s=2\theta P\sigma^*}^{2\theta P\sigma^*+1}E(s)\right]&=\sum_{j\neq r_1}\big[E(2\theta P\sigma^*+1)\big]_{r_2j}\Lambda_{j}\big[E(2\theta P\sigma^*)\big]\\
& \ \ \ \ +\big[E(2\theta P\sigma^*+1)\big]_{r_2r_1}\Lambda_{r_1}\big[E(2\theta P\sigma^*)\big]<1.
\end{aligned}
\end{equation}
Furthermore, based on the condition $(r_2,r_3)\in\hat{\mathscr{E}}(2\theta P\sigma^*+1)$, we have
\begin{equation}\label{sys:7.14}
\begin{aligned}
\Lambda_{r_3}\left[\prod_{s=2\theta P\sigma^*}^{2\theta P\sigma^*+2}E(s)\right]
&=\sum_{j\neq r_2}\big[E(2\theta P\sigma^*+2)\big]_{r_3j}\Lambda_{j}\left[\prod_{s=2\theta P\sigma^*}^{2\theta P\sigma^*+1}E(s)\right]\\
& \ \ \ \ +\big[E(2\theta P\sigma^*+1)\big]_{r_3r_2}\Lambda_{r_2}\left[\prod_{s=2\theta P\sigma^*}^{2\theta P\sigma^*+1}E(s)\right]<1.
\end{aligned}
\end{equation}
Continue to analyze the edges $(r_3,r_4),\ldots,(r_{2P\sigma^*-1},r_{2P\sigma^*})$, and we can finally deduce that
\begin{equation}\label{sys:7.15}
\begin{aligned}
\Lambda_{r}\left[\prod_{s=2\theta P\sigma^*}^{2(\theta+1)P\sigma^*h-1}E(s)\right]<1.
\end{aligned}
\end{equation}
This means that $\|E^*(\theta)\|_{\infty}<1$ for any $\theta\in\mathbb{N}$. Equivalently, from error system (\ref{sys:7.3}), it can be obtained that
\begin{equation}\label{sys:7.16}
\begin{aligned}
\lim_{k\rightarrow\infty}\|Y(k)\|_{\infty}&\leq\lim_{k\rightarrow\infty}\Big\|\prod_{s=0}^kE(s)\Big\|_{\infty}\|Y(0)\|_{\infty}\\
&\leq\lim_{\theta\rightarrow\infty}\prod_{s=0}^{\theta}\|E^*(s)\|_{\infty}\|Y(0)\|_{\infty}=0.
\end{aligned}
\end{equation}
That is to say, the bipartite tracking is achieved.

\emph{\underline{Necessity}}: The proof of the necessity is similar to the proof of Theorem~\ref{theorem:3.3}, so it is omitted here.
\end{proof}

\section{Bipartite tracking of Second-order MASs with arbitrary switching topologies}\label{section:8}

In the actual MASs, locally visible neighbors of each agent will likely change over time because the communication links between agents may be unreliable owing to disturbances and communication range limitations. This shifts our attention to the problem of bipartite tracking of MASs with switching topologies.

In the discrete-time setting, the information interactions among the followers at time instant $k\tau$ are described by a structurally balanced signed digraph $\mathscr{G}_k=(\mathscr{V},\mathscr{E}_k)$. Let $\mathscr{A}_k=[a^k_{ij}]$ stand for the weighted adjacency matrix of $\mathscr{G}_k$. Let $a^k_{i0}$ describe the directed information interaction from the leader to follower $v_i$ at time instant $k\tau$. Let $\tilde{\mathscr{G}}_k=(\tilde{\mathscr{V}}_k,\tilde{\mathscr{E}}_k)$ be a digraph being composed of digraph $\mathscr{G}_k$, vertex $v_0$ and the directed edges from $v_0$ to another vertices. In addition, assume that all the modulus of weighting factors have uniform lower and upper bound, i.e. $|a^k_{ij}|,|a^k_{i0}|\in[\underline{\varrho},\overline{\varrho}]$, where $0<\underline{\varrho}<\overline{\varrho}$.

For realizing bipartite tracking, the following conditions based on the communication topologies are required:
\begin{enumerate}
\item [\textbf{C3}.] All signed digraphs $\tilde{\mathscr{G}}_k$, $k\in\mathbb{N}$ are structurally balanced. In each $\tilde{\mathscr{G}}_k$, the vertices are divided into two subsets $\mathscr{V}_1$ and $\mathscr{V}_2$ with intra-subset cooperation and inter-subset competition.
\item [\textbf{C4}.] All discrete times $k\tau$, $k\in\mathbb{N}$ are divided into a sequence of continuous, uniformly bounded and nonempty time intervals $[s_j\tau,s_{j+1}\tau)$, $j\in\mathbb{N}$, starting at $s_0=0$. In the union digraph $\bigcup_{k=s_j}^{s_{j+1}-1}\tilde{\mathscr{G}}_k$ related to each interval, for each follower, there is a directed path from the leader and to that follower.
\end{enumerate}

\subsection{The case with a static leader}\label{section:8.1}

Consider second-order MASs (\ref{sys:4.1}) and (\ref{sys:4.2}). The distributed state feedback controller is given by
\begin{eqnarray}\label{sys:8.1}
\begin{aligned}
u_{i}(k)=-\gamma \vartheta_{i}(k)&+\sum\limits_{v_j\in\mathscr{N}^k_i}|a^k_{ij}|\big[\sgn(a^k_{ij})x_{j}(k)-x_{i}(k)\big]\\
&+|a^k_{i0}|\big[\sgn(a^k_{i0})x_{0}(k)-x_{i}(k)\big],
\end{aligned}
\end{eqnarray}
where $\gamma>0$ is a gain parameter and $\mathscr{N}^k_i$ is the set of the neighbors of agent $v_i$ at time instant $k\tau$.

Substituting controller (\ref{sys:8.1}) into systems (\ref{sys:4.1}) and (\ref{sys:4.2}), we have
\begin{eqnarray}\label{sys:8.2}
\begin{aligned}
y(k+1)=\big[\Phi_k\otimes I_{p}\big]y(k),
\end{aligned}
\end{eqnarray}
where $y(k)$ is defined in (\ref{sys:4.5}) and
\begin{eqnarray*}
    \Phi_k=\left(\begin{array}{@{\hspace{0.1em}}cc@{\hspace{0.1em}}}
                   I_{n}\!-\!\frac{\gamma\tau}{2}I_{n} & \frac{\gamma\tau}{2}I_{n} \\
                    \frac{\gamma\tau}{2}I_{n}\!-\!\frac{2\tau}{\gamma}\big(\mathscr{D}_k\!+\!\mathscr{B}_k\!-\!|\mathscr{A}_k|\big) & I_{n}\!-\!\frac{\gamma\tau}{2}I_{n} \\
                  \end{array}
\right)
\end{eqnarray*}
with
\begin{equation*}
\begin{aligned}
&\mathscr{D}_k=\diag\Big\{\sum_{v_j\in\mathscr{N}^k_1}|a^k_{1j}|,\sum_{v_j\in\mathscr{N}^k_2}|a^k_{2j}|,\ldots,\sum_{v_j\in\mathscr{N}^k_n}|a^k_{nj}|\Big\},\\
&\mathscr{B}_k=\diag\big\{|a^k_{10}|,|a^k_{20}|,\ldots,|a^k_{n0}|\big\}.
\end{aligned}
\end{equation*}
When the parameter $\gamma$ is selected as
\begin{eqnarray}\label{sys:8.3}
2\sqrt{n\overline{\varrho}}\leq\gamma<\frac{2}{\tau},
\end{eqnarray}
$\Phi_k$, $k\in\mathbb{N}$ are sub-stochastic matrices with positive diagonal elements. Thus, the bipartite tracking issue is equivalently transformed into the product convergence issue of ISubSM, which is equivalent to prove $\lim_{k\rightarrow\infty}\prod_{s=0}^k\Phi_s=\mathbf{0}$.

A sufficient condition of realizing second-order bipartite tracking with a static leader and switching topologies is established below.

\begin{theorem}\label{theorem:8.1}
Consider systems (\ref{sys:4.1}) and (\ref{sys:4.2}). Suppose that gain parameter $\gamma$ meets condition (\ref{sys:8.3}). The distributed controller (\ref{sys:8.1}) is said to solve the problem of bipartite tracking if the communication topologies satisfy \textbf{C3} and  \textbf{C4}.
\end{theorem}

\begin{proof}\
In order to get the main conclusion, we first need to obtain the following inequality
\begin{eqnarray}\label{sys:8.4}
\Big\|\prod_{s=s_j}^{s_{j+2n}}\Phi_s\Big\|_{\infty}<1, \ j\in\mathbb{N}.
\end{eqnarray}
Under condition (\ref{sys:8.3}), the matrices $\Phi_k$, $k\in\mathbb{N}$ are sub-stochastic and contain positive diagonal elements. Below we use mathematical induction to prove (\ref{sys:8.4}).

\emph{\textbf{Step I.}} Firstly, consider the time interval $[s_j\tau,s_{j+2}\tau)$. From the conditions \textbf{C3} and \textbf{C4}, it is known that for each follower, there is a directed path from the leader $v_0$ to that follower in the union digraph $\bigcup_{s=s_j}^{s_{j+1}-1}\tilde{\mathscr{G}}_s$. Thus, we can find a vertex $v_{\delta_1}$ such that $(v_0,v_{\delta_1})\in\bigcup_{s=s_j}^{s_{j+1}-1}\tilde{\mathscr{E}}_s$. Without losing generality, we assume that $(v_0,v_{\delta_1})\in\tilde{\mathscr{E}}_{l_j}$, where $s_j\leq l_j<s_{j+1}$, then it can be obtained that
\begin{eqnarray}\label{sys:8.5}
\Lambda_{\delta_1+n}\big[\Phi_{l_j}\big]<1,
\end{eqnarray}
and further
\begin{eqnarray}\label{sys:8.6}
\begin{aligned}
\Lambda_{\delta_1+n}\left[\prod_{s=s_j}^{l_j}\Phi_{s}\right]
&=\sum_{i=1}^{2n}\big[\Phi_{l_j}\big]_{\delta_1+n,i}\Lambda_i\left[\prod_{s=s_j}^{l_j-1}\Phi_{s}\right]\\
&\leq\Lambda_{\delta_1+n}\big[\Phi_{l_j}\big]<1.
\end{aligned}
\end{eqnarray}
According to the result of Theorem \ref{theorem:2.4}, one knows that (\ref{sys:8.6}) and the condition $\big[\Phi_{s_{j+1}-1}\big]_{\delta_1+n,\delta_1+n}>0$ can guarantee
\begin{eqnarray}\label{sys:8.7}
\begin{aligned}
\Lambda_{\delta_1+n}\left[\prod_{s=s_j}^{s_{j+1}-1}\Phi_{s}\right]<1,
\end{aligned}
\end{eqnarray}
which together with the conditions $\big[\Phi_{s_{j+2}-1}\big]_{\delta_1+n,\delta_1+n}>0$ and $\big[\Phi_{s_{j+2}-1}\big]_{\delta_1,\delta_1+n}>0$, ensures
\begin{eqnarray}\label{sys:8.8}
\begin{aligned}
\Lambda_{\delta_1+n}\left[\prod_{s=s_j}^{s_{j+2}-1}\Phi_{s}\right]<1, \ \Lambda_{\delta_1}\left[\prod_{s=s_j}^{s_{j+2}-1}\Phi_{s}\right]<1.
\end{aligned}
\end{eqnarray}

\emph{\textbf{Step II.}} In the interval $[s_{j}\tau,s_{j+2\theta}\tau)$, $\theta=1,2,\ldots,n-1$, we assume that
\begin{eqnarray}\label{sys:8.9}
\begin{aligned}
\Lambda_{\delta_{\theta}+n}\left[\prod_{s=s_{j}}^{s_{j+2\theta}-1}\Phi_{s}\right]<1, \ \Lambda_{\delta_{\theta}}\left[\prod_{s=s_{j}}^{s_{j+2\theta}-1}\Phi_{s}\right]<1.
\end{aligned}
\end{eqnarray}
Consider the interval $[s_{j}\tau,s_{j+2\theta+2}\tau)$. Since there is a directed path from the leader $v_0$ to that follower in the union digraph $\bigcup_{s=s_{j+2\theta}}^{s_{j+2\theta+1}-1}\tilde{\mathscr{G}}_s$, we can always find a vertex $v_{\delta_{\theta+1}}\neq v_{\delta_{\theta}}$ such that $(v_0,v_{\delta_{\theta+1}})\in\bigcup_{s=s_{j+2\theta}}^{s_{j+2\theta+1}-1}\tilde{\mathscr{E}}_s$ or $(v_{\delta_\theta},v_{\delta_{\theta+1}})\in\bigcup_{s=s_{j+2\theta}}^{s_{j+2\theta+1}-1}\tilde{\mathscr{E}}_s$. Without loss of generality, we assume that $(v_0,v_{\delta_{\theta+1}})\in\tilde{\mathscr{E}}_{l_{j+2\theta}}$ or $(v_{\delta_\theta},v_{\delta_{\theta+1}})\in\tilde{\mathscr{E}}_{l_{j+2\theta}}$, where $s_{j+2\theta}\leq l_{j+2\theta}<s_{j+2\theta+1}$. Below we discuss two situations.

1) If $(v_0,v_{\delta_{\theta+1}})\in\tilde{\mathscr{E}}_{l_{j+2\theta}}$, then similar to the analysis for the interval $[s_j\tau,s_{j+2}\tau)$, we can obtain that
\begin{eqnarray}\label{sys:8.10}
\begin{aligned}
\Lambda_{\delta_{\theta+1}+n}\left[\prod_{s=s_{j+2\theta}}^{s_{j+2\theta+1}-1}\Phi_{s}\right]<1, \ \Lambda_{\delta_{\theta+1}}\left[\prod_{s=s_{j+2\theta}}^{s_{j+2\theta+1}-1}\Phi_{s}\right]<1.
\end{aligned}
\end{eqnarray}
Since $[\prod_{s=s_{j}}^{s_{j+2\theta}-1}\Phi_{s}]_{\delta_{\theta+1},\delta_{\theta+1}}>0$ and $[\prod_{s=s_{j}}^{s_{j+2\theta}-1}\Phi_{s}]_{\delta_{\theta+1}+n,\delta_{\theta+1}+n}>0$, we can derive from the result in Theorem \ref{theorem:2.3} that
\begin{eqnarray}\label{sys:8.11}
\begin{aligned}
\Lambda_{\delta_{\theta+1}+n}\left[\prod_{s=s_{j}}^{s_{j+2\theta+1}-1}\Phi_{s}\right]<1, \ \Lambda_{\delta_{\theta+1}}\left[\prod_{s=s_{j}}^{s_{j+2\theta+1}-1}\Phi_{s}\right]<1.
\end{aligned}
\end{eqnarray}
Furthermore, since $[\prod_{s=s_{j+2\theta+1}}^{s_{j+2\theta+2}-1}\Phi_{s}]_{\delta_{\theta+1},\delta_{\theta+1}}>0$ and $[\prod_{s=s_{j+2\theta+1}}^{s_{j+2\theta+2}-1}\Phi_{s}]_{\delta_{\theta+1}+n,\delta_{\theta+1}+n}>0$, we get
\begin{eqnarray}\label{sys:8.12}
\begin{aligned}
\Lambda_{\delta_{\theta+1}+n}\left[\prod_{s=s_{j}}^{s_{j+2\theta+2}-1}\Phi_{s}\right]<1, \ \Lambda_{\delta_{\theta+1}}\left[\prod_{s=s_{j}}^{s_{j+2\theta+2}-1}\Phi_{s}\right]<1.
\end{aligned}
\end{eqnarray}

2) If $(v_{\delta_\theta},v_{\delta_{\theta+1}})\in\tilde{\mathscr{E}}_{l_{j+2\theta}}$, then we have
\begin{eqnarray}\label{sys:8.13}
\begin{aligned}
\big[\Phi_{s_{j+2\theta}}\big]_{\delta_{\theta+1}+n,\delta_\theta}>0, \ \big[\Phi_{s_{j+2\theta}}\big]_{\delta_{\theta+1},\delta_{\theta+1}+n}>0.
\end{aligned}
\end{eqnarray}
Combining (\ref{sys:8.9}) and (\ref{sys:8.13}), we can also derive the inequalities in (\ref{sys:8.12}) according to Theorem \ref{theorem:2.4}.

In addition, by the result in (\ref{sys:8.9}) and the conditions $[\prod_{s=s_{j+2\theta}}^{s_{j+2\theta+1}-1}\Phi_{s}]_{\delta_{\theta},\delta_{\theta}}>0$ and $[\prod_{s=s_{j+2\theta}}^{s_{j+2\theta+1}-1}\Phi_{s}]_{\delta_{\theta}+n,\delta_{\theta}+n}>0$, one has
\begin{eqnarray}\label{sys:8.14}
\begin{aligned}
\Lambda_{\delta_{\theta}+n}\left[\prod_{s=s_{j}}^{s_{j+2\theta+2}-1}\Phi_{s}\right]<1, \ \Lambda_{\delta_{\theta}}\left[\prod_{s=s_{j}}^{s_{j+2\theta+2}-1}\Phi_{s}\right]<1.
\end{aligned}
\end{eqnarray}

Summarizing (\ref{sys:8.12}) and (\ref{sys:8.14}), we can observe that in the interval $[s_{j}\tau,s_{j+2\theta+2}\tau)$, there exist four integers $\delta_{\theta}, \delta_{\theta}+n, \delta_{\theta+1}, \delta_{\theta+1}+n$ such that $\Lambda_{i}\big[\prod_{s=s_{j}}^{s_{j+2\theta+2}-1}\Phi_{s}\big]<1$, $i\in\{\delta_{\theta}, \delta_{\theta}+n, \delta_{\theta+1}, \delta_{\theta+1}+n\}$.

Based on the results obtained by the above mathematical induction, we can get that
\begin{eqnarray}\label{sys:8.15}
\begin{aligned}
\Lambda_{i}\big[\prod_{s=s_{j}}^{s_{j+2n}-1}\Phi_{s}\big]<1, \ i=1,2,\ldots,2n.
\end{aligned}
\end{eqnarray}
This means that the result (\ref{sys:8.4}) holds.

Let $\Phi^*_{\theta}=\prod_{s=s_{2\theta n}}^{s_{2(\theta+1)n}-1}\Phi_s$, $\theta\in\mathbb{N}$. Then it is known from (\ref{sys:8.4}) that $\big\|\Phi^*_{\theta}\big\|_{\infty}<1$. It follows that
\begin{eqnarray}\label{sys:8.16}
\begin{aligned}
\lim_{k\rightarrow\infty}\|y(k+1)\|_{\infty}&=\lim_{k\rightarrow\infty}\left\|\left[\prod_{s=0}^k\Phi_s\otimes I_p\right]y(0)\right\|_{\infty}\\
&=\lim_{\theta\rightarrow\infty}\left\|\left[\prod_{s=0}^\theta \Phi^*_s\otimes I_p\right]y(0)\right\|_{\infty}\\
&\leq\lim_{\theta\rightarrow\infty}\prod_{s=0}^\theta\big\|\Phi^*(\theta)\big\|_{\infty}\|y(0)\|_{\infty}=0.
\end{aligned}
\end{eqnarray}
Consequently, bipartite tracking with a static leader is realized.
\end{proof}

\begin{figure}[t]
  \centering
    \subfigure[$\tilde{\mathscr{G}}_a$]{
    \includegraphics[width=1.5in]{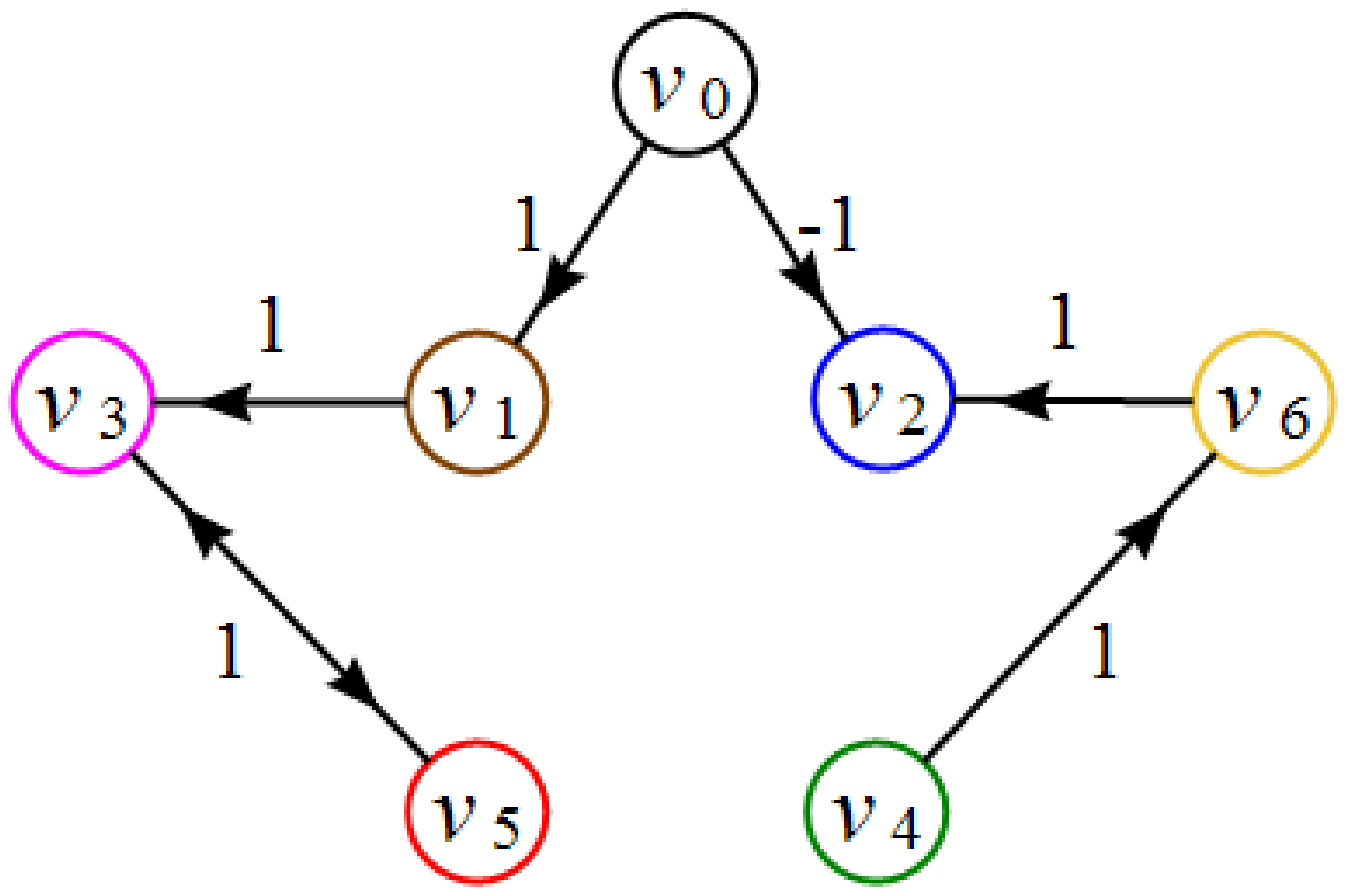}}
      \subfigure[$\tilde{\mathscr{G}}_b$]{
    \includegraphics[width=1.5in]{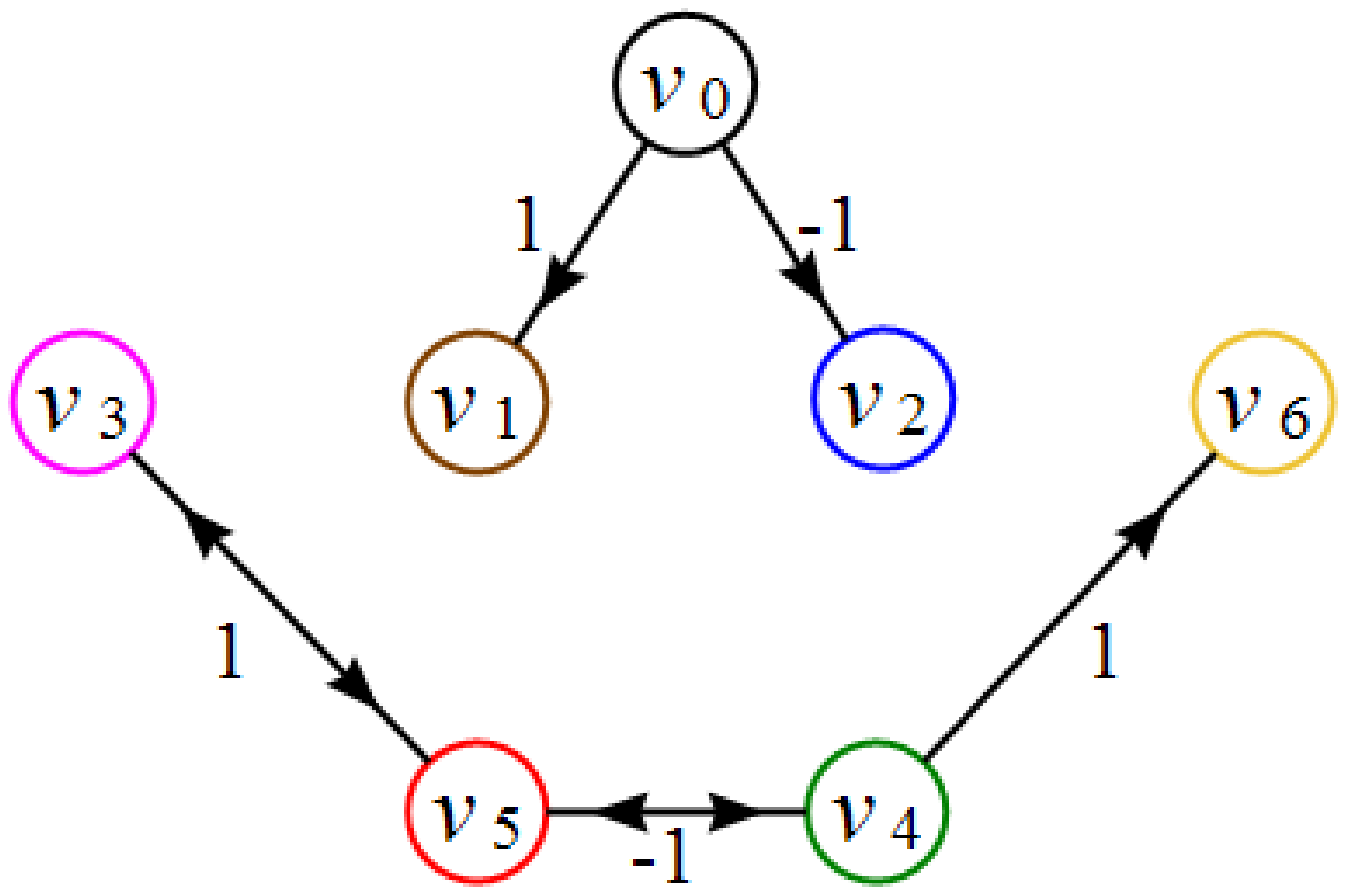}}
    \subfigure[$\tilde{\mathscr{G}}_c$]{
    \includegraphics[width=1.5in]{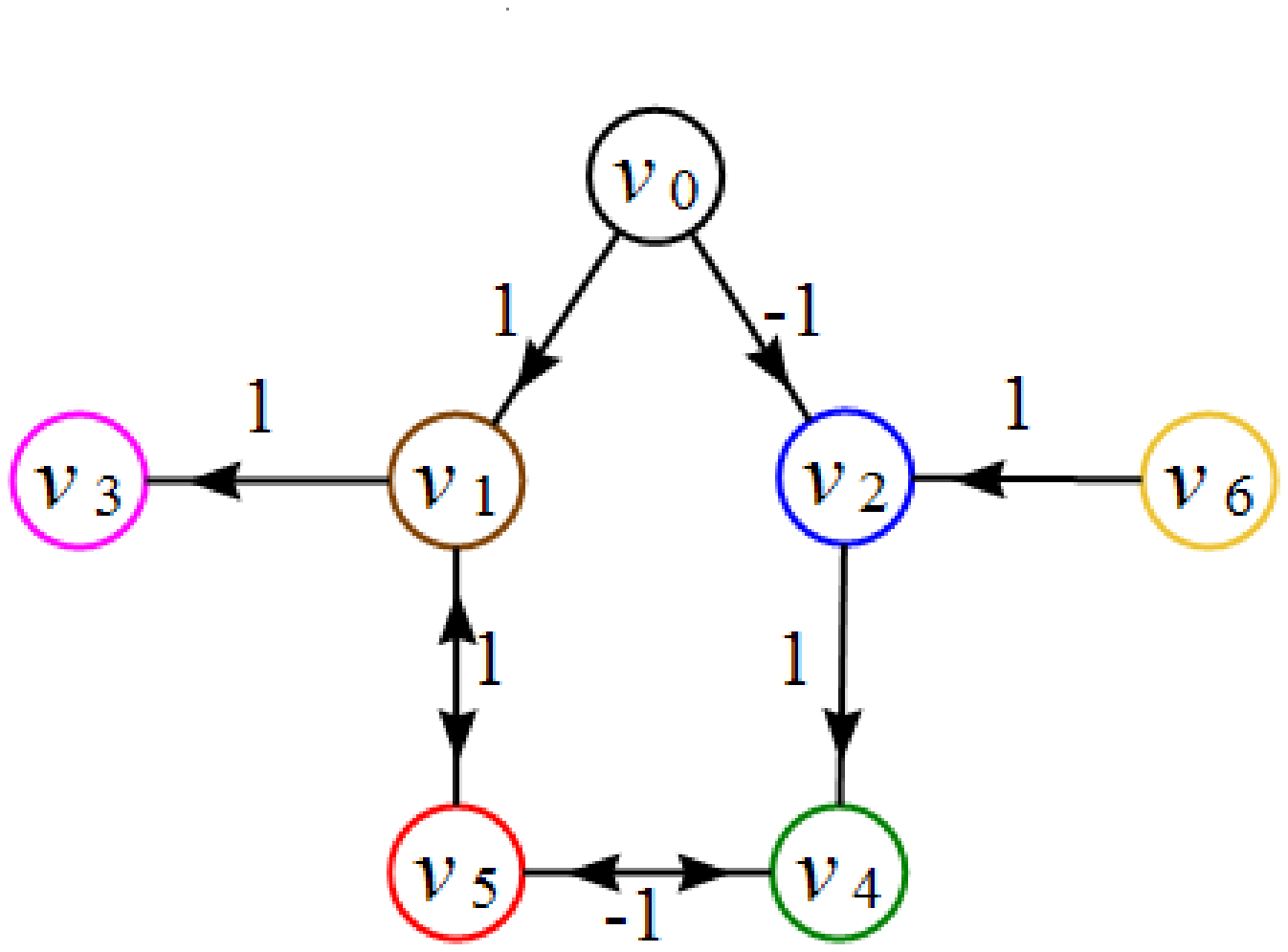}}
  \caption{Switching topologies. }\label{fig7}
\end{figure}

\begin{figure}[t]
  \centering
      \subfigure[position trajectories of second-order dynamic agents]{
    \includegraphics[width=2.4in]{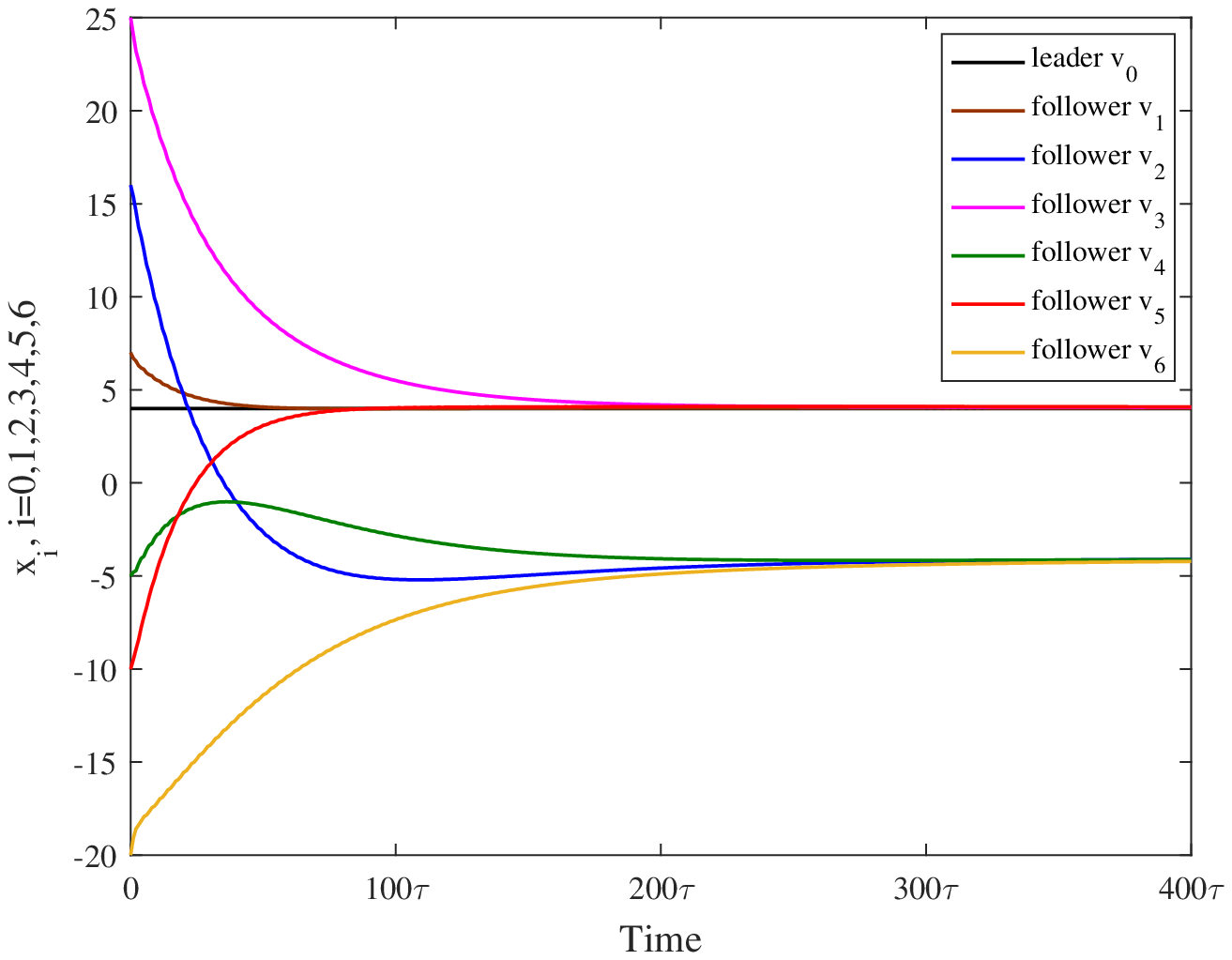}}
    \subfigure[velocity trajectories of second-order dynamic agents]{
    \includegraphics[width=2.4in]{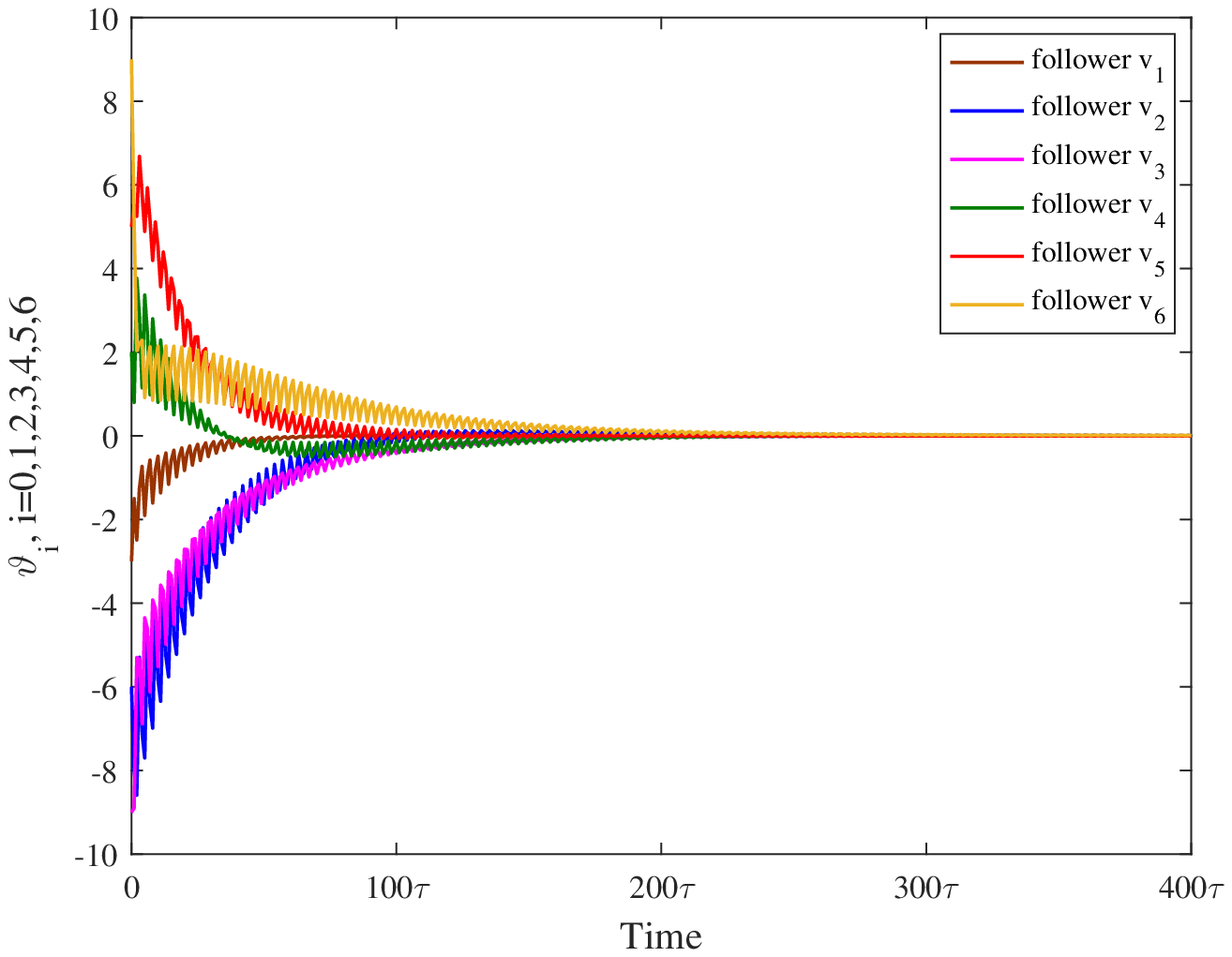}}
  \caption{State trajectories of all agents in Example~\ref{example:8.2}. }\label{fig8}
\end{figure}

\begin{example}\label{example:8.2}
Consider a team of seven agents interacting by switching topologies $\tilde{\mathscr{G}}_a, \tilde{\mathscr{G}}_b$ and $\tilde{\mathscr{G}}_c$,  shown in Fig.~\ref{fig7}. The agents are divided in to two subgroups, i.e.,  $\mathscr{V}_1=\{v_0,v_1,v_2,v_3\}$ and $\mathscr{V}_2=\{v_4,v_5,v_6\}$. Divid all discrete time instants $0,\tau,\ldots,k\tau,\ldots$, into a series of time intervals $\big[3\theta\tau, 3(\theta+1)\tau\big)$, $\theta\in\mathbb{N}$. Suppose that the network topologies undergo a periodic switching every $\tau=0.1$, from $\tilde{\mathscr{G}}_a$ to $\tilde{\mathscr{G}}_b$, from $\tilde{\mathscr{G}}_b$ to $\tilde{\mathscr{G}}_c$ and, then $\tilde{\mathscr{G}}_c$ to $\tilde{\mathscr{G}}_a$. Obviously, the conditions \textbf{C3} and \textbf{C4} are satisfied. Because the modulus of all edge weights is equal to 1, so $\overline{\varrho}=1$. Select $\gamma=6$ that satisfies (\ref{sys:8.3}). Lastly, Fig.~\ref{fig7} plots the agents' state trajectories, from which the realizing of the bipartite tracking with a static leader in the case of switching topologies can be observed.
\end{example}

\subsection{The case with an active leader}\label{section:8.2}

Consider second-order MASs (\ref{sys:5.1}) and (\ref{sys:5.2}). Each agent updates its states depending on the following distributed controller
\begin{eqnarray}\label{sys:8.17}
\begin{aligned}
u_{i}(k)=&\sum\limits_{v_j\in\mathscr{N}^k_i}|a^k_{ij}|\big[\sgn(a^k_{ij})x_{j}(k)-x_{i}(k)\big]\\
&+|a^k_{i0}|\big[\sgn(a^k_{i0})x_{0}(k)-x_{i}(k)\big]\\
&+\beta\sum\limits_{v_j\in\mathscr{N}^k_i}|a^k_{ij}|\big[\sgn(a^k_{ij})\vartheta_{j}(k)-\vartheta_{i}(k)\big]\\
&+\beta |a^k_{i0}|\big[\sgn(a^k_{i0})\vartheta_{0}-\vartheta_{i}(k)\big],
\end{aligned}
\end{eqnarray}
where $\beta>0$ is a gain parameter.

Based on the vector constructions in (\ref{sys:5.5}), controller (\ref{sys:8.17}) is used to convert systems (\ref{sys:5.1}) and (\ref{sys:5.2}) to the following error system
\begin{eqnarray}\label{sys:8.18}
\begin{aligned}
\xi(k+1)=[N_k\otimes I_{p}]\xi(k),
\end{aligned}
\end{eqnarray}
where
\begin{eqnarray*}
    N_k=\left(\begin{array}{c@{\hspace{0.4em}}c}
                   I_{n}-\frac{\tau}{\beta}I_{n} & \frac{\tau}{\alpha\beta}I_{n} \\
                    -\frac{\alpha\tau}{\beta}I_{n} & I_{n}+\frac{\tau}{\beta}I_{n}-\beta \tau \big(\mathscr{D}_k+\mathscr{B}_k-|\mathscr{A}_k|\big) \\
                  \end{array}
\right).
\end{eqnarray*}
The implementation of bipartite tracking for systems (\ref{sys:5.1}) and (\ref{sys:5.2}) is equivalent to the achievement of the asymptotic stability of error system (\ref{sys:8.18}), that is, to prove $\lim_{k\rightarrow\infty}\prod_{s=0}^kN_s=\mathbf{0}$. In general, it is difficult to directly study this convergence using the nonnegative matrix theory since there are negative elements in $N_k$. To solve this convergence issue, we construct a new matrix
\begin{eqnarray*}
    \Gamma_k=\left(\begin{array}{c@{\hspace{0.4em}}c}
                   I_{n}-\frac{\tau}{\beta}I_{n} & \frac{\tau}{\alpha\beta}I_{n} \\
                    \frac{\alpha\tau}{\beta}I_{n} & I_{n}+\frac{\tau}{\beta}I_{n}-\beta \tau \big(\mathscr{D}_k+\mathscr{B}_k-|\mathscr{A}_k|\big) \\
                  \end{array}
\right)
\end{eqnarray*}
for each $N_k$. To ensure that $\Gamma_k$ is a super-stochastic matrix, we present the following lemma.

\begin{lemma}\label{lemma:8.3}
If the parameter $\beta$ satisfies the following inequality
\begin{eqnarray}\label{sys:8.19}
\sqrt{\frac{1+\alpha}{\underline{\varrho}}}<\beta<\frac{1}{\tau n\bar{\varrho}},
\end{eqnarray}
then $\Gamma_k$ is super-stochastic matrix with positive diagonal elements and its row sums satisfy:
\begin{eqnarray*}
\begin{aligned}
&\Lambda_{i}[\Gamma_k]=1-\frac{\tau}{\beta}+\frac{\tau}{\alpha\beta}<1;\\
&\Lambda_{i+n}[\Gamma_k]\leq1+\frac{\alpha \tau}{\beta}+\frac{\tau}{\beta}-\beta\tau \underline{\varrho}<1, \ \text{if} \ |a^k_{i0}|>0;\\
&\Lambda_{i+n}[\Gamma_k]=1+\frac{\alpha \tau}{\beta}+\frac{\tau}{\beta}>1, \ \text{if} \ |a^k_{i0}|=0,
\end{aligned}
\end{eqnarray*}
where $i\in\{1,2,\ldots,n\}$.
\end{lemma}

\begin{proof}
Firstly, according to condition (\ref{sys:8.19}), we can deduce that
\begin{eqnarray*}
\beta>\sqrt{\frac{1+\alpha}{\underline{\varrho}}}>\frac{1}{\sqrt{n\bar{\varrho}}}
>\frac{1}{\sqrt{n\bar{\varrho}(1+\alpha)}}.
\end{eqnarray*}
In addition, we can also see $\sqrt{\frac{1+\alpha}{\underline{\varrho}}}<\frac{1}{\tau n\bar{\varrho}}$ through (\ref{sys:8.19}). It follows that $\tau<\frac{1}{n\bar{\varrho}}\sqrt{\frac{\underline{\varrho}}{1+\alpha}}$. Consequently,
\begin{eqnarray*}
\tau<\frac{1}{n\bar{\varrho}}\sqrt{\frac{n\bar{\varrho}}{1+\alpha}}
=\frac{1}{\sqrt{n\bar{\varrho}(1+\alpha)}}.
\end{eqnarray*}
This means that $\beta>\tau$. And thus, $I_{n}-\frac{\tau}{\beta}I_{n}$ is a diagonal matrix in which all elements lie on $(0,1)$. Next we analyze the diagonal elements of matrix $\Gamma^*_k=I_{n}+\frac{\tau}{\beta}I_{n}-\beta \tau \big(\mathscr{D}_k+\mathscr{B}_k-|\mathscr{A}_k|\big)$. It can be observed that the minimum diagonal element of $\Gamma_k$ is $1+\frac{\tau}{\beta}-\beta\tau n\bar{\varrho}$. According to the condition $\beta<\frac{1}{\tau n\bar{\varrho}}$, it can be deduced that $1+\frac{\tau}{\beta}-\beta\tau n\bar{\varrho}>0$. Thus, the matrix $\Gamma^*_k$ contains positive diagonal elements. In addition, since $|\mathscr{A}_k|$ is a nonnegative matrix, the non-diagonal elements in $\Gamma^*_k$ are nonnegative. Consequently, $\Gamma^*_k$ is a nonnegative matrix with positive diagonal elements. This implies that $\Gamma_k$ is a nonnegative matrix with positive diagonal elements.

Now we analyze the row sums of matrix $\Gamma_k$. Obviously, the sums of the first $n$ rows of matrix $\Gamma_k$ satisfy $\Lambda_{i}[\Gamma_k]=1-\frac{\tau}{\beta}+\frac{\tau}{\alpha\beta}<1$, $i=1,2,\ldots,n$. For the sums of the last $n$ rows of matrix $\Gamma_k$, we can obtain the following results: if $|a^k_{i0}|=0$, then $\Lambda_{n+i}[\Gamma_k]=1+\frac{\alpha \tau}{\beta}+\frac{\tau}{\beta}>1$, and if $|a^k_{i0}|>0$, then $\Lambda_{n+i}[\Gamma_k]\leq1+\frac{\alpha \tau}{\beta}+\frac{\tau}{\beta}-\beta\tau \underline{\varrho}<1$ based on the fact $\underline{\varrho}>\frac{1+\alpha}{\beta^2}$ that is deduced through (\ref{sys:8.19}). This completes the proof.
\end{proof}

With the above preparations, we have the following result.

\begin{theorem}\label{theorem:8.4}
Consider systems (\ref{sys:5.1}) and (\ref{sys:5.2}) with distributed protocol (\ref{sys:8.17}). Suppose that the communication topologies satisfy \textbf{C3} and \textbf{C4}. The bipartite tracking can be achieved if condition (\ref{sys:8.19}) and the following conditions hold
\begin{subequations}\label{sys:8.20}
\begin{align}
&\frac{\tau}{\beta}<\frac{\sqrt[n\varsigma-1]{\alpha}-1}{1+\alpha}, \label{sys:8.20a}\\
&g^{n\varsigma}-\beta \tau \underline{\varrho}\varpi^{n\varsigma-1}<1,\label{sys:8.20b}
\end{align}
\end{subequations}
where $g=1+\frac{\tau}{\beta}+\frac{\alpha \tau}{\beta}$, $\alpha>g^{n\varsigma-1}$, $\varpi=\min\big\{\beta \tau\underline{\varrho},1+\frac{\tau}{\beta}-\beta \tau n\overline{\varrho}\big\}$ and $\varsigma$ is the upper bound of the number of time instants in the intervals $[s_j\tau,s_{j+1}\tau), j\in\mathbb{N}$.
\end{theorem}

\begin{proof}
We first analyze the infinite norm of the product of matrices $\Gamma_k$ in a bounded time interval $[s_j\tau,s_{j+n}\tau)$. Under condition (\ref{sys:8.19}), the matrices $\Gamma_k$, $k\in\mathbb{N}$ are super-stochastic. Below we have two steps to prove the inequality
\begin{eqnarray}\label{sys:8.21}
\Big\|\prod_{s=s_j}^{s_j+n}\Gamma_s\Big\|_{\infty}<1.
\end{eqnarray}

\textbf{\emph{Step I.}} We first analyze the sums of the last $n$ rows of the matrix product $\prod_{s=s_j}^{s_j+n}\Gamma_s$. Consider the interval $[s_j\tau,s_{j+1}\tau)$. By \textbf{C3} and \textbf{C4}, it is known that for each follower, there is a directed path from the leader $v_0$ to that follower in the union digraph $\bigcup_{s=s_j}^{s_{j+1}-1}\tilde{\mathscr{G}}_s$. Consequently, we can find a vertex $v_{\delta_1}$ such that $(v_0,v_{\delta_1})\in\bigcup_{s=s_j}^{s_{j+1}-1}\tilde{\mathscr{E}}_s$. Without losing generality, it is assumed that $(v_0,v_{\delta_1})\in\tilde{\mathscr{E}}_{l_j}$, where $s_j\leq l_j<s_{j+1}$, then it can be obtained that
\begin{eqnarray}\label{sys:8.22}
\begin{aligned}
\Lambda_{\delta_1+n}[\Gamma_{l_j}]\leq g-\beta\tau \underline{\varrho}<1.
\end{aligned}
\end{eqnarray}
Since $[\Gamma_{s_{j+1}-1}]_{\delta_1+n,\delta_1+n}\geq\varpi$, we can obtain from the result in Theorem \ref{theorem:2.8} that
\begin{eqnarray}\label{sys:8.23}
\begin{aligned}
\Lambda_{\delta_1+n}\left[\prod_{s=l_j}^{s_{j+1}-1}\Gamma_{s}\right]\leq g^{s_{j+1}-l_j}-\beta\tau \underline{\varrho}\varpi^{s_{j+1}-l_j-1}.
\end{aligned}
\end{eqnarray}
Since $\Lambda_{i+n}[\Gamma_{k}]\leq g$ for any $i=1,2,\ldots,n$ and $k\in\mathbb{N}$, we have
\begin{eqnarray}\label{sys:8.24}
\begin{aligned}
\Lambda_{i+n}\left[\prod_{s=s_j}^{l_j-1}\Gamma_{s}\right]\leq g^{l_j-s_j}, \ i=1,2,\ldots,n.
\end{aligned}
\end{eqnarray}
According to (\ref{sys:8.23}) and (\ref{sys:8.24}), we can deduce that
\begin{eqnarray}\label{sys:8.25}
\begin{aligned}
\Lambda_{\delta_1+n}\left[\prod_{s=s_j}^{s_{j+1}-1}\Gamma_{s}\right]&=
\sum_{i=1}^{2n}\left[\prod_{s=l_j}^{s_{j+1}-1}\Gamma_{s}\right]_{\delta_1+n,i}\Lambda_{i}\left[\prod_{s=s_j}^{l_j-1}\Gamma_{s}\right]\\
&\leq \big(g^{s_{j+1}-l_j}-\beta\tau \underline{\varrho}\varpi^{s_{j+1}-l_j-1}\big)g^{l_j-s_j}\\
&=g^{s_{j+1}-s_j}-\beta\tau \underline{\varrho}\varpi^{s_{j+1}-s_j-1}\\
&\leq g^{\varsigma}-\beta\tau \underline{\varrho}\varpi^{\varsigma-1}.
\end{aligned}
\end{eqnarray}

Consider the time interval $[s_{j}\tau,s_{j+2}\tau)$. Since there is a directed path from the leader $v_0$ to that follower in the union digraph $\bigcup_{s=s_{j+1}}^{s_{j+2}-1}\tilde{\mathscr{G}}_s$, we can always find a vertex $v_{\delta_{2}}\neq v_{\delta_{1}}$ such that $(v_0,v_{\delta_{2}})\in\bigcup_{s=s_{j+1}}^{s_{j+2}-1}\tilde{\mathscr{E}}_s$ or $(v_{\delta_1},v_{\delta_{2}})\in\bigcup_{s=s_{j+1}}^{s_{j+2}-1}\tilde{\mathscr{E}}_s$. Without loss of generality, we assume that $(v_0,v_{\delta_{2}})\in\tilde{\mathscr{E}}_{l_{j+1}}$ or $(v_{\delta_1},v_{\delta_{2}})\in\tilde{\mathscr{E}}_{l_{j+1}}$, where $s_{j+1}\leq l_{j+1}<s_{j+2}$. We discuss two situations below. For the case $(v_0,v_{\delta_{2}})\in\tilde{\mathscr{E}}_{l_{j+1}}$, we have $\Lambda_{\delta_{2}+n}\big[\Gamma_{l_{j+1}}\big]\leq g-\beta\tau \underline{\varrho}<1$.
By Theorem \ref{theorem:2.8}, it can be derived that
\begin{eqnarray}\label{sys:8.26}
\begin{aligned}
\Lambda_{\delta_{2}+n}\left[\prod_{s=l_{j+1}}^{s_{j+2}-1}\Gamma_{s}\right]\leq g^{s_{j+2}-l_{j+1}}-\beta\tau \underline{\varrho}\varpi^{s_{j+2}-l_{j+1}-1},
\end{aligned}
\end{eqnarray}
which combines the result $\Lambda_{i+n}\big[\prod_{s=s_j}^{l_{j+1}-1}\Gamma_{s}\big]\leq g^{l_{j+1}-s_j}$, $i=1,2,\ldots,2n$, guarantees that
\begin{eqnarray}\label{sys:8.27}
\begin{aligned}
\Lambda_{\delta_2+n}\left[\prod_{s=s_j}^{s_{j+2}-1}\Gamma_{s}\right]\leq g^{2\varsigma}-\beta\tau \underline{\varrho}\varpi^{2\varsigma-1}.
\end{aligned}
\end{eqnarray}
For the case $(v_{\delta_1},v_{\delta_{2}})\in\tilde{\mathscr{E}}_{l_{j+1}}$, we can still get the inequality in (\ref{sys:8.27}) based on the result in Theorem \ref{theorem:2.8}. In addition, for (\ref{sys:8.25}), we can further deduce
\begin{eqnarray}\label{sys:8.28}
\begin{aligned}
\Lambda_{\delta_1+n}\left[\prod_{s=s_j}^{s_{j+2}-1}\Gamma_{s}\right]\leq g^{2\varsigma}-\beta\tau \underline{\varrho}\varpi^{2\varsigma-1}.
\end{aligned}
\end{eqnarray}
Consequently, in the time interval $[s_{j}\tau,s_{j+2}\tau)$, there always exist two different vertices $v_{\delta_{1}}, v_{\delta_{2}}\!\in\!\mathscr{V}$ such that $\Lambda_{i+n}\big[\prod_{s=s_j}^{s_{j+2}-1}\Gamma_{s}\big]\leq g^{2\varsigma}-\beta\tau \underline{\varrho}\varpi^{2\varsigma-1}$, $i\in\{\delta_1,\delta_2\}$.

By using the analysis method similar to the interval $[s_{j}\tau,s_{j+2}\tau)$, we can derive that in the interval $[s_{j}\tau,s_{j+n}\tau)$, there exist $n$ different vertices $v_{\delta_{1}},v_{\delta_{2}},\ldots, v_{\delta_{n}}\in\mathscr{V}$ with $\{\delta_1,\delta_2,\ldots,\delta_n\}=\{1,2,\ldots,n\}$ such that
\begin{eqnarray}\label{sys:8.29}
\begin{aligned}
\Lambda_{i+n}\left[\prod_{s=s_j}^{s_{j+n}-1}\Gamma_{s}\right]\leq g^{n\varsigma}-\beta\tau \underline{\varrho}\varpi^{n\varsigma-1}, \ i=1,2,\ldots,n.
\end{aligned}
\end{eqnarray}
By condition (\ref{sys:8.20b}), we can see that the sums of the last $n$ rows of the matrix product $\prod_{s=s_j}^{s_j+n}\Gamma_s$ are less than 1.

\textbf{\emph{Step 2.}} Now we discuss the sums of the first $n$ rows of the matrix product $\prod_{s=s_j}^{s_j+n}\Gamma_s$. Since $\alpha>g^{n\varsigma-1}$ and the fact $g>1$, we have $1-\frac{\tau}{\beta}+\frac{ \tau}{\alpha\beta}g^{n\varsigma-1}<1$. According to the expressions of the inequalities in (\ref{sys:5.24}) and (\ref{sys:5.25}), one has
\begin{eqnarray}\label{sys:8.30}
\begin{aligned}
\Lambda_{i}\left[\prod_{s=s_j}^{s_{j+n}}\Gamma_s\right]\leq 1-\frac{\tau}{\beta}+\frac{ \tau}{\alpha\beta}g^{n\varsigma-1}<1, \ i=1,2,\ldots,n.
\end{aligned}
\end{eqnarray}
Obviously, (\ref{sys:8.29}) and (\ref{sys:8.30}) mean that (\ref{sys:8.21}) has to be established.

Divide the time axis into a series of intervals $[s_{\theta n}\tau, s_{(\theta+1) n}\tau)$, $\theta\in\mathbb{N}$. Let $\Gamma^*_{\theta}=\prod_{s=s_{\theta n}}^{s_{(\theta+1)n}-1}\Gamma_s$, $\theta\in\mathbb{N}$. Then it is known from (\ref{sys:8.21}) that $\big\|\Gamma^*_{\theta}\big\|_{\infty}<1$. It follows that
\begin{eqnarray}
\begin{aligned}
\lim_{k\rightarrow\infty}\|\xi(k+1)\|_{\infty}&\leq\lim_{k\rightarrow\infty}\left\|\left[\prod_{s=0}^kN_s\otimes I_p\right]\right\|_{\infty}\|\xi(0)\|_{\infty}\\
&\leq\lim_{k\rightarrow\infty}\left\|\left[\prod_{s=0}^k\Gamma_s\otimes I_p\right]\right\|_{\infty}\|\xi(0)\|_{\infty}\\
&\leq\lim_{\theta\rightarrow\infty}\prod_{s=0}^\theta\big\|\Gamma^*(\theta)\big\|_{\infty}\|\xi(0)\|_{\infty}=0.
\end{aligned}
\end{eqnarray}
That is to say, bipartite tracking with an active leader is realized.
\end{proof}

\section{Bipartite tracking of second-order MASs with random packet losses}\label{section:9}

In practice, one of the most common phenomena in networked control systems is data packet losses caused by interference or congestion. In this section, we consider the bipartite tracking behavior of second-order MASs with random packet losses from the controller to the actuator.

The phenomenon of packet losses is characterized by a Bernoulli stochastic variable $\theta_k\in\{0,1\}$, $k\in\mathbb{N}$. More specifically, if $\theta_k=1$, then the control packets from the controller are successfully transmitted to the actuator, and if $\theta_k=0$, then the control packets are lost. Let the Bernoulli stochastic variable $\theta_k$ satisfies:
\begin{equation*}
\begin{aligned}
&\emph{Prob}\{\theta_k=1\}=\mathbb{E}\big\{\theta_k\big\}=\bar{\theta},\\
&\emph{Prob}\{\theta_k=0\}=1-\mathbb{E}\big\{\theta_k\big\}=1-\bar{\theta},
\end{aligned}
\end{equation*}
where $\bar{\theta}\in[0,1]$ is known as the successful transmission rate. Based on the above description of lossy links, the control input to the actuator is designed as
\begin{equation}\label{sys:9.1}
u_{i}(k)=\left\{
\begin{aligned}
&u^c_{i}(k), \ \theta_k=1,\\
&0, \ \ \ \ \ \ \ \theta_k=0,
\end{aligned}
\right.
\end{equation}
where the linear feedback adopted in remote controller is given by
\begin{eqnarray}\label{sys:9.2}
\begin{aligned}
u^c_{i}(k)=&\sum_{v_j\in\mathscr{N}_i}|a_{ij}|\big[\sgn(a_{ij})x_{j}(k)-x_{i}(k)\big]\\
&+|a_{i0}|\big[\sgn(a_{i0})x_{0}(k)-x_{i}(k)\big]\\
&+\beta\sum_{v_j\in\mathscr{N}_i(k)}|a_{ij}|\big[\sgn(a_{ij})\vartheta_{j}(k)-\vartheta_{i}(k)\big]\\
&+\beta |a_{i0}|\big[\sgn(a_{i0})\vartheta_{0}-\vartheta_{i}(k)\big],
\end{aligned}
\end{eqnarray}
in which $\beta>0$ is a gain parameter. The transmission process of control packets from the controller to the actuator under lossy links is described in Fig. \ref{fig9}.

\begin{figure}[t]
  \centering
    \includegraphics[width=2in]{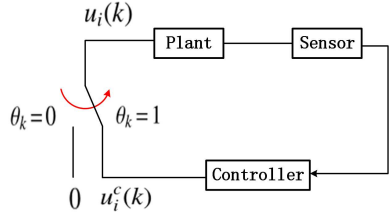}
  \caption{The transmission process of control packets from the controller to the actuator under lossy links.}\label{fig9}
\end{figure}

\begin{definition}\label{definition:9.1}
The bipartite tracking for systems (\ref{sys:5.1}) and (\ref{sys:5.2}) in expectation is said to be realized if
\begin{eqnarray*}
\begin{aligned}
\forall v_i\in\mathscr{V}_1,\ & \lim_{k\rightarrow \infty}\mathbb{E}\big[\|x_{i}(k)-x_{0}(k)\|\big]=0,\
\lim_{k\rightarrow \infty}\mathbb{E}\big[\|\vartheta_{i}(k)-\vartheta_{0}\|\big]=0;\\
\forall v_i\in\mathscr{V}_{2},\ &\lim_{k\rightarrow \infty}\mathbb{E}\big[\|x_{i}(k)+x_{0}(k)\|\big]=0,\
\lim_{k\rightarrow \infty}\mathbb{E}\big[\|\vartheta_{i}(k)+\vartheta_{0}\|\big]=0.
\end{aligned}
\end{eqnarray*}
\end{definition}

Under control protocol (\ref{sys:9.2}), systems (\ref{sys:5.1}) and (\ref{sys:5.2}) can be written as
\begin{eqnarray}\label{sys:9.3}
\mathbb{E}\big[\xi(k+1)\big]=
\left\{
\begin{aligned}
&\big[H\otimes I_{p}\big]\mathbb{E}[\xi(k)], \ \ \theta_k=1,\\
&\big[H^*\otimes I_{p}\big]\mathbb{E}[\xi(k)], \ \theta_k=0,
\end{aligned}
\right.
\end{eqnarray}
where $\xi(k)$ is defined in (\ref{sys:5.5}), and
\begin{eqnarray*}
\begin{aligned}
    &H=\left(\begin{array}{c@{\hspace{0.4em}}c}
                   I_{n}-\frac{\tau}{\beta}I_{n} & \frac{\tau}{\alpha\beta}I_{n} \\
                    -\frac{\alpha\tau}{\beta}I_{n} & I_{n}+\frac{\tau}{\beta}I_{n}-\beta \tau \big(\mathscr{D}+\mathscr{B}-|\mathscr{A}|\big) \\
                  \end{array}
\right),\\
&H^*=\left(\begin{array}{c@{\hspace{0.4em}}c}
                   I_{n}-\frac{\tau}{\beta}I_{n} & \frac{\tau}{\alpha\beta}I_{n} \\
                    -\frac{\alpha\tau}{\beta}I_{n} & I_{n}+\frac{\tau}{\beta}I_{n}\\
                  \end{array}
\right).
\end{aligned}
\end{eqnarray*}

The main result is given below.

\begin{theorem}\label{theorem:9.2}
Consider systems (\ref{sys:5.1}) and (\ref{sys:5.2}), where the gain parameter $\beta$ satisfies the the following inequalities
\begin{eqnarray}\label{sys:9.4}
\begin{aligned}
\max\left\{\sqrt{\frac{1+\alpha}{b_m}},\frac{\alpha^{\frac{P}{P-1}}-1}{\tau b_m\varphi^{P-1}}\right\}<\beta\leq\frac{1}{\tau d_M},
\end{aligned}
\end{eqnarray}
where $\alpha>(1+\frac{\tau}{\beta}+\frac{\alpha \tau}{\beta})^{P-1}$, and $d_M,b_m,\varphi$ are defined in (\ref{sys:3.7}), (\ref{sys:5.7}) and (\ref{sys:5.9}), respectively. Suppose that the successful transmission rate satisfies $\bar{\theta}>0$. Using distributed protocol (\ref{sys:9.1}), the bipartite tracking in expectation is achieved  if and only if the communication topology meets the conditions \textbf{C1} and \textbf{C2}.
\end{theorem}

\begin{proof}
\emph{\underline{Sufficiency}}:\
For the matrix $H$, we construct the following matrix
\begin{eqnarray*}
\begin{aligned}
    &\Psi=\left(\begin{array}{c@{\hspace{0.4em}}c}
                   I_{n}-\frac{\tau}{\beta}I_{n} & \frac{\tau}{\alpha\beta}I_{n} \\
                    \frac{\alpha\tau}{\beta}I_{n} & I_{n}+\frac{\tau}{\beta}I_{n}-\beta \tau \big(\mathscr{D}+\mathscr{B}-|\mathscr{A}|\big) \\
                  \end{array}
\right).
\end{aligned}
\end{eqnarray*}
We can deduce that under the condition $\sqrt{\frac{1+\alpha}{b_m}}<\beta\leq\frac{1}{\tau d_M}$, $\Psi$ is a super-stochastic matrix and its row sums satisfy: 1) $\Lambda_{i}[\Psi]\!=\!1\!-\!\frac{\tau}{\beta}\!+\!\frac{\tau}{\alpha\beta}\!<\!1$; 2) $\Lambda_{n+i}[\Psi]\!=\!1\!+\!\frac{\alpha \tau}{\beta}\!+\!\frac{\tau}{\beta}\!-\!\beta\tau |a_{i0}|\!<\!1$ if $|a_{i0}|\!>\!0$; 3) $\Lambda_{n+i}[\Psi]\!=\!1\!+\!\frac{\alpha \tau}{\beta}\!+\!\frac{\tau}{\beta}\!>\!1$ if $|a_{i0}|\!=\!0$, where $i\in\{1,2,\ldots,n\}$. Furthermore, under the condition $\beta>\frac{\alpha^{\frac{P}{P-1}}-1}{\tau b_m\varphi^{P-1}}$, we can derive by the proof method of Theorem \ref{theorem:5.5} that $\lim_{k\rightarrow\infty}\Psi^k=\mathbf{0}$. And since the successful transmission rate satisfies $\bar{\theta}>0$, we have
\begin{eqnarray}\label{sys:9.5}
\begin{aligned}
\lim_{k\rightarrow\infty}\mathbb{E}\big[\|\xi(k+1)\|_{\infty}\big]
&\leq\lim_{k\rightarrow\infty}\mathbb{E}\big[\|H^k\|_{\infty}\big]\mathbb{E}\big[\|\xi(0)\big\|_{\infty}\big]\\
&\leq\lim_{k\rightarrow\infty}\mathbb{E}\big[\|\Psi^k\|_{\infty}\big]\mathbb{E}\big[\|\xi(0)\big\|_{\infty}\big]=0.
\end{aligned}
\end{eqnarray}
This implies that error system (\ref{sys:9.3}) converges gradually to zeros when $\bar{\theta}>0$. Therefore, the bipartite tracking in expectation is achieved.

\emph{\underline{Necessity}}:
The proof of the necessity is similar to the proof of Theorem~\ref{theorem:3.3}, so it is omitted here.
\end{proof}

\section{Bipartite tracking of second-order MASs over a random signed network}\label{section:10}

As a matter of fact, the quality of wireless communication channels in actual networks tends to change with time due to congestion and random fading, which makes random networks more widely used. Therefore, we consider the problem of bipartite tracking of MASs over a random signed network.

Assume the information interactions among the followers are described by a random signed network $\mathscr{G}=\{\mathscr{V},\mathscr{E}\}$, an edge $(i,j)$ is determined randomly with probability $p_{ij}\in[0,1]$. The weight matrix of random network $\mathscr{G}$ is denoted by $\mathscr{W}=[w_{ij}]$, where $w_{ij}\neq0$ if $(j,i)\in\mathscr{E}$, and otherwise, $w_{ij}=0$. Let $\bar{\mathscr{A}}=[\bar{a}_{ij}]$ represent the random adjacency matrix of $\mathscr{G}$, where
\begin{equation*}
\bar{a}_{ij}=\left\{
\begin{aligned}
&w_{ij} \ \ {\rm with \ probability} \ p_{ij}\\
&0 \ \ \ \ \ {\rm with \ probability} \ 1-p_{ij}
\end{aligned}
\right..
\end{equation*}
Thus the expectation of the adjacency matrix can be written as
\begin{equation*}
\mathbb{E}[\bar{\mathscr{A}}]=
             \begin{bmatrix}
               p_{11}w_{11} & \cdots & p_{1n}w_{1n} \\
               \vdots & \ddots & \vdots \\
               p_{n1}w_{n1} & \cdots & p_{nn}w_{nn} \\
             \end{bmatrix}.
\end{equation*}
In addition, for any follower $v_i$, the edge $(v_0,v_i)$ is determined randomly with probability $p_{i0}$ and its weight is denoted by $w_{i0}$. The adjacency element from the leader to follower $v_i$ satisfies: $\bar{a}_{i0}=w_{i0}$ with probability $p_{i0}$, and $\bar{a}_{i0}=0$ with probability $1-p_{i0}$. The information interactions among all agents are described by a random network $\tilde{\mathscr{G}}=\{\tilde{\mathscr{V}},\tilde{\mathscr{E}}\}$.

Consider second-order MASs (\ref{sys:5.1}) and (\ref{sys:5.2}). The control input is designed as
\begin{eqnarray}\label{sys:10.1}
\begin{aligned}
u_{i}(k)=&\sum_{j=1}^n|\bar{a}_{ij}|\big[\sgn(\bar{a}_{ij})x_{j}(k)-x_{i}(k)\big]\\
&+|\bar{a}_{i0}|\big[\sgn(\bar{a}_{i0})x_{0}(k)-x_{i}(k)\big]\\
&+\beta\sum_{j=1}^n|\bar{a}_{ij}|\big[\sgn(\bar{a}_{ij})\vartheta_{j}(k)-\vartheta_{i}(k)\big]\\
&+\beta |\bar{a}_{i0}|\big[\sgn(\bar{a}_{i0})\vartheta_{0}-\vartheta_{i}(k)\big],
\end{aligned}
\end{eqnarray}
in which $\beta>0$ is a gain parameter.

Using control protocol (\ref{sys:10.1}), systems (\ref{sys:5.1}) and (\ref{sys:5.2}) can be expressed as
\begin{eqnarray}\label{sys:10.2}
\mathbb{E}\big[\xi(k+1)\big]=\mathbb{E}\big[\bar{H}\big]\mathbb{E}\big[\xi(k)\big],
\end{eqnarray}
where $\xi(k)$ is defined in (\ref{sys:5.5}), and
\begin{eqnarray*}
\begin{aligned}
\mathbb{E}[\bar{H}]=\left(\begin{array}{c@{\hspace{0.4em}}c}
                   I_{n}-\frac{\tau}{\beta}I_{n} & \frac{\tau}{\alpha\beta}I_{n} \\
                    -\frac{\alpha\tau}{\beta}I_{n} & I_{n}+\frac{\tau}{\beta}I_{n}-\beta \tau \big(\mathbb{E}[\bar{\mathscr{D}}]+\mathbb{E}[\bar{\mathscr{B}}]-\big|\mathbb{E}[\bar{\mathscr{A}}]\big|\big) \\
                  \end{array}
\right),
\end{aligned}
\end{eqnarray*}
with
\begin{eqnarray*}
\begin{aligned}
&\mathbb{E}[\bar{\mathscr{D}}]=\diag\big\{\sum_{j=1}^np_{1j}|w_{1j}|,\sum_{j=1}^np_{2j}|w_{2j}|,\ldots,\sum_{j=1}^np_{nj}|w_{nj}|\big\},\\
&\mathbb{E}[\bar{\mathscr{B}}]=\diag\left\{p_{10}|w_{10}|,p_{20}|w_{20}|,\ldots,p_{n0}|w_{n0}|\right\}.
\end{aligned}
\end{eqnarray*}

\begin{theorem}\label{theorem:10.2}
Consider systems (\ref{sys:5.1}) and (\ref{sys:5.2}) under control protocol (\ref{sys:10.1}). Suppose that the gain parameter $\beta$ satisfies
\begin{eqnarray}\label{sys:10.3}
\begin{aligned}
\max\left\{\sqrt{\frac{1+\alpha}{\bar{b}_m}},\frac{\alpha^{\frac{P}{P-1}}-1}{\tau \bar{b}_m\bar{\varphi}^{P-1}}\right\}<\beta\leq\frac{1}{\tau \bar{d}_M},
\end{aligned}
\end{eqnarray}
where $\alpha>(1+\frac{\tau}{\beta}+\frac{\alpha \tau}{\beta})^{P-1}$, and
\begin{eqnarray*}
\begin{aligned}
&\bar{d}_M=\max\big\{\sum_{j=1}^np_{ij}|w_{ij}|+p_{i0}|w_{i0}| \mid i=1,2,\ldots,n\big\},\\
&\bar{b}_m=\min\big\{p_{i0}|w_{i0}| \mid i=1,2,\ldots,n\big\},\\
&\bar{\varphi}=\min\big\{[\Phi']_{ij}\mid [\Phi']_{ij}>0, \ i,j=1,2,\ldots,n\big\},
\end{aligned}
\end{eqnarray*}
in which $\Phi'=I_{n}+\frac{\tau}{\beta}I_{n}-\beta \tau \big(\mathbb{E}[\bar{\mathscr{D}}]+\mathbb{E}[\bar{\mathscr{B}}]-\big|\mathbb{E}[\bar{\mathscr{A}}]\big|\big)$. The bipartite tracking in expectation is achieved if and only if the expected graph of communication topology meets the conditions \textbf{C1} and \textbf{C2}.
\end{theorem}

\begin{proof}
\emph{\underline{Sufficiency}}:\
For the matrix $\mathbb{E}[\bar{H}]$, we construct the following matrix
\begin{eqnarray*}
\begin{aligned}
    &\mathbb{E}[\bar{\Psi}]=\left(\begin{array}{c@{\hspace{0.4em}}c}
                   I_{n}-\frac{\tau}{\beta}I_{n} & \frac{\tau}{\alpha\beta}I_{n} \\
                    \frac{\alpha\tau}{\beta}I_{n} & I_{n}+\frac{\tau}{\beta}I_{n}-\beta \tau \big(\mathbb{E}[\bar{\mathscr{D}}]+\mathbb{E}[\bar{\mathscr{B}}]-\big|\mathbb{E}[\bar{\mathscr{A}}]\big|\big) \\
                  \end{array}
\right).
\end{aligned}
\end{eqnarray*}
Then it can be deduced that under the condition $\sqrt{\frac{1+\alpha}{\bar{b}_m}}<\beta\leq\frac{1}{\tau \bar{d}_M}$, $\mathbb{E}[\bar{\Psi}]$ is a super-stochastic matrix. Furthermore, when $\beta>\frac{\alpha^{\frac{P}{P-1}}-1}{\tau \bar{b}_m\bar{\varphi}^{P-1}}$, we can derive by the proof method of Theorem \ref{theorem:5.5} that $\lim_{k\rightarrow\infty}\mathbb{E}^k[\bar{\Psi}]=0_{2n\times 2n}$. It thus follows that
\begin{eqnarray}\label{sys:10.4}
\begin{aligned}
\lim_{k\rightarrow\infty}\mathbb{E}\big[\|\xi(k+1)\|_{\infty}\big]
&\leq\lim_{k\rightarrow\infty}\big\|\mathbb{E}\big[\bar{H}]\big\|^k_{\infty}\mathbb{E}\big[\|\xi(0)\big\|_{\infty}\big]\\
&\leq\lim_{k\rightarrow\infty}\big\|\mathbb{E}\big[\bar{\Psi}]\big\|^k_{\infty}\mathbb{E}\big[\|\xi(0)\big\|_{\infty}\big]=0.
\end{aligned}
\end{eqnarray}
Therefore, the bipartite tracking in expectation is achieved.

\emph{\underline{Necessity}}:
The proof of the necessity is similar to the proof of Theorem~\ref{theorem:3.3}, so it is omitted here.
\end{proof}

\section{Bipartite tracking of second-order MASs with matrix disturbance}\label{section:11}

\subsection{The case with a static leader}\label{section:11.1}

In actual systems, external interference is almost inevitable, which causes errors in the designed controller. Due to the "butterfly effect", small errors may cause serious deterioration of system performance. In addition, the matrix to be decomposed in practical application fields is generally a measurement matrix mixed with disturbances. Naturally, we want to give sufficient conditions to implement bipartite tracking when there is disturbance on the adjacency matrix $\mathscr{A}$.

Consider second-order MASs (\ref{sys:4.1}) and (\ref{sys:4.2}). Let $\Delta\mathscr{A}=[\Delta a_{ij}]_{n\times n}$ denote the linear disturbance of matrix $\mathscr{A}$. It is assumed that the matrices $\mathscr{A}$ and $\Delta\mathscr{A}$ are the same type. We design the following distributed protocol:
\begin{eqnarray}\label{sys:11.1}
\begin{aligned}
u_{i}(k)=-\gamma\vartheta_i(k)&+\sum\limits_{v_j\in\mathscr{N}_i}|a_{ij}+\Delta a_{ij}|\big[\sgn(a_{ij}+\Delta a_{ij})x_{j}(k)-x_{i}(k)\big]\\
&+|a_{i0}+\Delta a_{i0}|\big[\sgn(a_{i0}+\Delta a_{i0})x_{0}(k)-x_{i}(k)\big],
\end{aligned}
\end{eqnarray}
where $\gamma>0$ is a gain parameter. Using protocol (\ref{sys:11.1}), systems (\ref{sys:4.1}) and (\ref{sys:4.2}) can be written as
\begin{eqnarray}\label{sys:11.2}
\begin{aligned}
y(k+1)=\big[\big(C+\Delta C\big)\otimes I_{p}\big]y(k),
\end{aligned}
\end{eqnarray}
where $y(k)$ is defined in (\ref{sys:4.5}), and
\begin{eqnarray*}
   & C=\left(\begin{array}{@{\hspace{0.1em}}cc@{\hspace{0.1em}}}
                   I_{n}\!-\!\frac{\gamma\tau}{2}I_{n} & \frac{\gamma\tau}{2}I_{n} \\
                    \frac{\gamma\tau}{2}I_{n}\!-\!\frac{2\tau}{\gamma}\big(\mathscr{D}\!+\!\mathscr{B}\!-\!|\mathscr{A}|\big) & I_{n}\!-\!\frac{\gamma\tau}{2}I_{n} \\
                  \end{array}
\right),\\
&\Delta C=\left(\begin{array}{@{\hspace{0.1em}}cc@{\hspace{0.1em}}}
                   \mathbf{0} & \mathbf{0} \\
                    -\frac{2\tau}{\gamma}\big(\Delta\mathscr{D}\!+\!\Delta\mathscr{B}\!-\!|\Delta\mathscr{A}|\big) & \mathbf{0} \\
                  \end{array}
\right),
\end{eqnarray*}
in which $\Delta C$ is the disturbance of matrix $C$, and
\begin{equation*}
\begin{aligned}
&\Delta\mathscr{D}=\diag\Big\{\sum_{v_j\in\mathscr{N}_1}|\Delta a_{1j}|,\sum_{v_j\in\mathscr{N}_2}|\Delta a_{2j}|,\ldots,\sum_{v_j\in\mathscr{N}_n}|\Delta a_{nj}|\Big\},\\
&\Delta\mathscr{B}=\diag\Big\{|\Delta a_{10}|,|\Delta a_{20}|,\ldots,|\Delta a_{n0}|\Big\}.
\end{aligned}
\end{equation*}

In the following, we show the result of bipartite tracking in the presence of matrix disturbance.

\begin{theorem}\label{theorem:11.1}
Consider systems (\ref{sys:4.1}) and (\ref{sys:4.2}), where the gain parameter $\psi$ satisfies the the following inequality
\begin{eqnarray}\label{sys:11.3}
2\sqrt{\tilde{d}_M}\leq\gamma<\frac{2}{\tau},
\end{eqnarray}
where $\tilde{d}_M=\max\big\{\sum_{v_j\in\mathscr{N}_i}|a_{ij}\!+\!\Delta a_{ij}|\!+\! |a_{i0}\!+\!\Delta a_{i0}|\mid i=1,2,\ldots,n\big\}$. The distributed protocol (\ref{sys:10.1}) solves the problem of bipartite tracking if and only if the communication topology meets the conditions \textbf{C1} and \textbf{C2}.
\end{theorem}

\begin{proof}
\emph{\underline{Sufficiency}}:\
Under condition (\ref{sys:11.3}), we can deduce that $C+\Delta C$ is a sub-stochastic matrix in which the diagonal elements are positive and the row sums satisfy:
\begin{eqnarray*}
\left\{
\begin{aligned}
&\Lambda_{i}\big[C+\Delta C\big]<1 \ {\rm if} \ (v_0,v_i)\in\tilde{\mathscr{E}},\\
&\Lambda_{i}\big[C+\Delta C\big]=1 \ {\rm if} \ (v_0,v_i)\notin\tilde{\mathscr{E}},
\end{aligned}
\right.
\end{eqnarray*}
where $i\in\{1,2,\ldots,n\}$. Thus, similar to the proof method of Theorem \ref{theorem:4.3}, we get
\begin{eqnarray*}
\begin{aligned}
\lim_{k\rightarrow\infty}\big[C+\Delta C\big]^k=\mathbf{0}.
\end{aligned}
\end{eqnarray*}
This implies that error system (\ref{sys:11.2}) converges gradually to zeros. That is to say, the bipartite tracking under matrix disturbance is realized.

\emph{\underline{Necessity}}:
The proof of the necessity is similar to the proof of Theorem~\ref{theorem:3.3}, so it is omitted here.
\end{proof}

\begin{remark}\label{remark:10.2}
Under the conditions \textbf{C1} and \textbf{C2}, the selection range of the parameter $\gamma$ given in (\ref{sys:4.7}) can guarantee that the matrices $C(k)$, $k\in\mathbb{N}$ are sub-stochastic and $\lim_{k\rightarrow\infty}\prod_{s=0}^kC(s)=\mathbf{0}$, namely, the asynchronous bipartite tracking in the absence of matrix disturbance is realized. However, condition (\ref{sys:4.7}) cannot guarantee that bipartite tracking can still be achieved if there exists the disturbance $\Delta C$ in system (\ref{sys:11.2}). This is because under condition (\ref{sys:4.7}), matrix $C+\Delta C$ may not be a sub-stochastic matrix, and thus error system (\ref{sys:11.2}) is likely to fail to converge. Therefore, it is required to give a more strict parameter selection range (\ref{sys:11.3}) to ensure that $C+\Delta C$ is still a sub-stochastic matrix. Thus, the convergence of error system (\ref{sys:11.2}) can be proved by using the properties of sub-stochastic matrix.
\end{remark}

\subsection{The case with an active leader}\label{section:11.2}

In the following, we consider the problem of bipartite tracking for second-order MASs (\ref{sys:5.1}) and (\ref{sys:5.2}) with the matrix disturbance. Assume that the adjacency matrix $\mathscr{A}$ and its disturbance $\Delta\mathscr{A}$ are the same type. The distributed protocol with disturbance is designed as:
\begin{eqnarray}\label{sys:11.4}
\begin{aligned}
u_{i}(k)=&\sum\limits_{v_j\in\mathscr{N}_i}|a_{ij}+\Delta a_{ij}|\big[\sgn(a_{ij}+\Delta a_{ij})x_{j}(k)-x_{i}(k)\big]\\
&+|a_{i0}+\Delta a_{i0}|\big[\sgn(a_{i0}+\Delta a_{i0})x_{0}(k)-x_{i}(k)\big]\\
&+\beta\sum\limits_{v_j\in\mathscr{N}_i}|a_{ij}+\Delta a_{ij}|\big[\sgn(a_{ij}+\Delta a_{ij})\vartheta_{j}(k)-\vartheta_{i}(k)\big]\\
&+\beta |a_{i0}+\Delta a_{i0}|\big[\sgn(a_{i0}+\Delta a_{i0})\vartheta_{0}\!-\!\vartheta_{i}(k)\big],
\end{aligned}
\end{eqnarray}
where $\beta>0$ is a fixed gain parameter. Under control protocol (\ref{sys:11.4}), the error system takes the following form
\begin{eqnarray}\label{sys:11.5}
\begin{aligned}
\xi(k+1)=\big[\big(H+\Delta H\big)\otimes I_{p}\big]\xi(k),
\end{aligned}
\end{eqnarray}
where $\xi(k)$ is defined in (\ref{sys:5.5}), and
\begin{eqnarray*}
\begin{aligned}
&H=\left(\begin{array}{c@{\hspace{0.4em}}c}
                   I_{n}-\frac{\tau}{\beta}I_{n} & \frac{\tau}{\alpha\beta}I_{n} \\
                    -\frac{\alpha\tau}{\beta}I_{n} & I_{n}+\frac{\tau}{\beta}I_{n}-\beta \tau \big(\mathscr{D}+\mathscr{B})-|\mathscr{A}|\big) \\
                  \end{array}
\right),\\
&\Delta H=\left(\begin{array}{c@{\hspace{0.4em}}c}
                   \mathbf{0} & \mathbf{0} \\
                    \mathbf{0} & \beta \tau \big(|\Delta\mathscr{A}|-\Delta\mathscr{D}-\Delta\mathscr{B}\big) \\
                  \end{array}
\right).
\end{aligned}
\end{eqnarray*}

A necessary and sufficient condition for bipartite tracking under adjacency matrix disturbance is established below.

\begin{theorem}\label{theorem:11.3}
Consider systems (\ref{sys:5.1}) and (\ref{sys:5.2}), where the gain parameter $\beta$ satisfies the the following inequalities
\begin{eqnarray}\label{sys:11.6}
\begin{aligned}
\max\left\{\sqrt{\frac{1+\alpha}{b_m+\Delta b_m}},\frac{\alpha^{\frac{P}{P-1}}-1}{\tau (b_m+\Delta b_m)\varphi_1^{P-1}}\right\}<\beta\leq\frac{1}{\tau \tilde{d}_M},
\end{aligned}
\end{eqnarray}
where
\begin{eqnarray*}
\begin{aligned}
&b_m+\Delta b_m=\min\{|b_i+\Delta b_i| \mid i=1,2,\ldots,n\},\\
&\varphi_1=\min\big\{[\Phi']_{ij}\mid [\Phi']_{ij}>0, \ i,j=1,2,\ldots,n\big\},
\end{aligned}
\end{eqnarray*}
with $\Phi'\!=\!I_{n}\!+\!\frac{\tau}{\beta}I_{n}\!-\!\beta \tau(\mathscr{D}\!+\!\Delta\mathscr{D}\!+\!\mathscr{B}\!+\!\Delta\mathscr{B}\!-\!|\mathscr{A}|\!-\!|\mathscr{A}|)$, and $\alpha$ is defined in (\ref{sys:5.5}). The distributed protocol (\ref{sys:11.1}) solves the problem of bipartite tracking if and only if the communication topology meets the conditions \textbf{C1} and \textbf{C2}.
\end{theorem}

\begin{proof}
\emph{\underline{Sufficiency}}:\
Define
\begin{eqnarray*}
\begin{aligned}
\Phi=\left(\begin{array}{c@{\hspace{0.4em}}c}
                   I_{n}-\frac{\tau}{\beta}I_{n} & \frac{\tau}{\alpha\beta}I_{n} \\
                    \frac{\alpha\tau}{\beta}I_{n} & I_{n}+\frac{\tau}{\beta}I_{n}-\beta \tau \big(\mathscr{D}+\mathscr{B})-|\mathscr{A}|\big) \\
                  \end{array}
\right).
\end{aligned}
\end{eqnarray*}
Under the condition $\sqrt{\frac{1+\alpha}{b_m+\Delta b_m}}<\beta\leq\frac{1}{\tau \tilde{d}_M}$, it can be derived that $\Psi+\Delta H$ is a super-stochastic matrix and its row sums satisfy:
\begin{eqnarray*}
\begin{aligned}
&\Lambda_{i}[\Phi+\Delta H]=1-\frac{\tau}{\beta}+\frac{\tau}{\alpha\beta}<1;\\
&\Lambda_{n+i}[\Phi+\Delta H]=1+\frac{\alpha \tau}{\beta}+\frac{\tau}{\beta}-\beta\tau |a_{i0}|<1, \ \text{if} \ |a_{i0}|>0;\\
&\Lambda_{n+i}[\Phi+\Delta H]=1+\frac{\alpha \tau}{\beta}+\frac{\tau}{\beta}>1, \ \text{if} \ |a_{i0}|=0,
\end{aligned}
\end{eqnarray*}
where $i\in\{1,2,\ldots,n\}$. Furthermore, when $\beta>\frac{\alpha^{\frac{P}{P-1}}-1}{\tau (b_m+\Delta b_m)\varphi_1^{P-1}}$, we can derive by the proof method of Theorem \ref{theorem:5.5} that
\begin{eqnarray*}
\begin{aligned}
\lim_{k\rightarrow\infty}\big[\Phi+\Delta H\big]^k=\mathbf{0}.
\end{aligned}
\end{eqnarray*}
It follows that
\begin{align*}
\lim_{k\rightarrow\infty}\|\xi(k+1)\|_{\infty}
&\leq\lim_{k\rightarrow\infty}\big\|H+\Delta H\big\|^k_{\infty}\big\|\xi(0)\big\|_{\infty}\\
&\leq\lim_{k\rightarrow\infty}\big\|\Phi+\Delta H\big\|^k_{\infty}\big\|\xi(0)\big\|_{\infty}=0.
\end{align*}
This implies that the bipartite tracking under matrix disturbance is achieved.

\emph{\underline{Necessity}}:
The proof of the necessity is similar to the proof of Theorem~\ref{theorem:3.3}, so it is omitted here.
\end{proof}

\begin{remark}
In Theorem \ref{theorem:5.5}, the appropriate parameter selection range in (\ref{sys:5.7}) and (\ref{sys:5.8}) can guarantee the convergence of the products of infinite super-stochastic matrices $\Phi(k)$, $k\in\mathbb{N}$, namely, $\lim_{k\rightarrow\infty}\prod_{s=0}^k\Phi(s)=\mathbf{0}$. In Theorem \ref{theorem:11.3}, in order to guarantee that the matrix convergence $\lim_{k\rightarrow\infty}\big[\Phi+\Delta H\big]^k=\mathbf{0}$ holds under the matrix disturbance, our approach is to establish a more strict parameter selection range (\ref{sys:11.6}) to ensure that $\Phi+\Delta H$, $k\in\mathbb{N}$ is a super-stochastic matrix, and then the convergence of error system (\ref{sys:11.5}) can be proved by using the properties of super-stochastic matrix.
\end{remark}

\section{Bipartite bounded tracking of second-order MASs with external noise disturbance}\label{section:12}

In reality, external disturbance is common in MASs due to various uncertainties such as model mismatch, channel noise, measurement error, etc. In this section, we analyze the problem of bipartite bounded tracking for second-order MASs with an active leader under the influence of external noise disturbance.

Consider second-order MASs (\ref{sys:5.1}) and (\ref{sys:5.2}). We design a linear control protocol based on the position and velocity information of agent $v_i$ and its neighbors, as follows
\begin{eqnarray}\label{sys:12.1}
\begin{aligned}
u_{i}(k)=\omega_i(k)+u^c_{i}(k),
\end{aligned}
\end{eqnarray}
where $u^c_{i}(k)$ is given in (\ref{sys:9.2}), $\omega_i(k)\in\mathbb{R}^p$ denotes the noise disturbance that satisfies $\|\omega_i(k)\|\leq\bar{\omega}$, in which $\bar{\omega}$ is a constant.

When the system is disturbed, there will be a bipartite bounded tracking phenomenon. Next, we give the definition of bipartite bounded tracking consensus for the first time.

\begin{definition}\label{definition:12.1}
Using protocol (\ref{sys:12.1}), the bipartite bounded tracking for systems (\ref{sys:5.1}) and (\ref{sys:5.2}) is said to be realized if there exist real numbers $c_1, c_2$ such that
\begin{eqnarray*}
\begin{aligned}
\forall v_i\in\mathscr{V}_1,\ & \lim_{k\rightarrow \infty}\left\|x_{i}(k)-x_{0}(k)\right\|\leq c_1,\
\lim_{k\rightarrow \infty}\left\|\vartheta_{i}(k)-\vartheta_{0}\right\|\leq c_2;\\
\forall v_i\in\mathscr{V}_{2},\ &\lim_{k\rightarrow \infty}\left\|x_{i}(k)+x_{0}(k)\right\|\leq c_1,\
\lim_{k\rightarrow \infty}\left\|\vartheta_{i}(k)+\vartheta_{0}\right\|\leq c_2.
\end{aligned}
\end{eqnarray*}
\end{definition}

Denote
\[\omega(k)=[\omega^T_1(k),\ldots,\omega^T_m(k),-\omega^T_{m+1}(k),\ldots,-\omega^T_n(k)]^T.\]
Then systems (\ref{sys:5.1}) and (\ref{sys:5.2}) under control protocol (\ref{sys:12.1}) can be written as the following form
\begin{eqnarray}\label{sys:12.2}
\begin{aligned}
e(k+1)=[\Omega_1\otimes I_p] e(k)+[\Omega_2\otimes I_p]\omega(k),
\end{aligned}
\end{eqnarray}
where $e(k)=\big[e^T_x(k),e^T_\vartheta(k)\big]^T$ in which $e_x(k)$ and $e_\vartheta(k)$ are defined in (\ref{sys:3.5}) and (\ref{sys:5.5}), respectively, and
\begin{eqnarray*}
\begin{aligned}
\Omega_1=\left(
           \begin{array}{cc}
             I_n & \tau I_n \\
            -\tau\mathscr{H} & I_n-\beta\tau\mathscr{H}\\
           \end{array}
         \right),\Omega_2=\left(
           \begin{array}{c}
             \mathbf{0} \\
            \tau I_n\\
           \end{array}
         \right)
\end{aligned}
\end{eqnarray*}
with $\mathscr{H}=\mathscr{D}+\mathscr{B}-|\mathscr{A}|$.

The characteristic polynomial of $\Omega_1$ is written as
\begin{eqnarray*}
\begin{aligned}
\text{det}(\lambda I_{2n}-\Omega_1)=&\text{det}\left(
    \begin{array}{cc}
      (\lambda-1)I_{n} & -\tau I_{n}\\
      \tau\mathscr{H} & (\lambda-1)I_{n}+\beta\tau\mathscr{H}\\
    \end{array}
\right)\\
=&\lambda^2I_n+(\beta\tau\mathscr{H}-2I_n)\lambda+\tau^2\mathscr{H}-\beta\tau\mathscr{H}+I_n.
\end{aligned}
\end{eqnarray*}
There exists a orthogonal matrix $J$ that satisfies $J^{-1}\mathscr{H}J=\text{diag}(\mu_{1},\mu_{2},\cdots,\mu_{n})$, where $\mu_{i}, i=1,2,\cdots,n$ are the eigenvalues of $\mathscr{H}$, it follows that
\begin{eqnarray*}
\begin{aligned}
\text{det}(\lambda I_{2n}-\Omega_1)=\lambda^2+(\beta\tau\mu_i-2)\lambda+\tau^2\mu_i-\beta\tau\mu_i+1.
\end{aligned}
\end{eqnarray*}
Therefore, the eigenvalues of $\Omega_1$ satisfy
\begin{eqnarray*}
\lambda^2+(\beta\tau\mu_i-2)\lambda+\tau^2\mu_i-\beta\tau\mu_i+1=0.
\end{eqnarray*}
For convenience, denote $f(\lambda,\mu_{i})=\lambda^2+(\beta\tau\mu_i-2)\lambda+\tau^2\mu_i-\beta\tau\mu_i+1$.

Before proceeding further, we first introduce the following known lemmas.

\begin{lemma}(\cite{Tang2011Leader})\label{lemma:12.2}
$\Omega_1$ is an $M$-matrix if and only if the topology graph meets the conditions \textbf{C1} and \textbf{C2}.
\end{lemma}

\begin{lemma}(\cite{Parks1993Sta})\label{lemma:12.3}
Given a complex-coefficient polynomial
\begin{eqnarray*}
\begin{aligned}
f(w)=w^2 +cw+d
\end{aligned}
\end{eqnarray*}
where $c, d$ are complex numbers, $f(w)$ is Hurwitz stable if and only if $Re(c)>0$ and $Re(c)Im(c)Im(d)+Re^2(c)Re(d)-Im^2(d)>0$.
\end{lemma}

In the following, our aim is to derive a condition that all eigenvalues of matrix $\Omega_1$ are within the unit circle by choosing the appropriate parameter $\beta$, that is, $f(\lambda,\mu_{i})$ is Schur stable for any $\mu_{i}$. Due to a fact that $\mu_i$ may be complex for directed graphs, it is not easy to determine directly whether $f(\lambda,\mu_{i})$ is Schur stable. For this consideration, we apply the bilinear transformation $\lambda=\frac{z+1}{z-1}$ to $f(\lambda,\mu_{i})$  yields
\begin{eqnarray*}
\tau^2\mu_iz^2+2(\beta\tau\mu_i-\tau^2\mu_i)z+4+\tau^2\mu_i-2\beta\tau\mu_i,
\end{eqnarray*}
which ie equal to
\begin{eqnarray*}
z^2+\frac{2(\beta-\tau)}{\tau}z+\frac{4+\tau^2\mu_i-2\beta\tau\mu_i}{\tau^2\mu_i}.
\end{eqnarray*}
It is worth pointing out that the bilinear transformation maps the open left-half plane one-to-one onto the interior of the unit circle. Thus, $f(\lambda,\mu_{i})$ is Schur stable if and only if
\begin{eqnarray*}
\begin{aligned}
g(z,\mu_{i})=z^2+\frac{2(\beta-\tau)}{\tau}z+\frac{4+\tau^2\mu_i-2\beta\tau\mu_i}{\tau^2\mu_i}
\end{aligned}
\end{eqnarray*}
is Hurwitz stable. We hence only need to prove that $g(z,\mu_{i})$ is Hurwitz stable by choosing the appropriate parameter $\beta$.

\begin{lemma}\label{lemma:12.4}
Suppose that the topology graph satisfies the conditions \textbf{C1} and \textbf{C2}. All eigenvalues of $\Omega_1$ are within the unit circle if the parameter $\beta$ satisfies
\begin{eqnarray}\label{sys:12.3}
\left\{
\begin{aligned}
&\beta>\tau, \ \ 4Re(\mu_i)+\tau^2|\mu_i|-2\beta\tau|\mu_i|>0, \ \ \ \ \ \ \ \ \ \ Im(\mu_i)=0,\\
&\beta>\tau, \ \ (\beta-\tau)^2>\frac{4|Im(\mu_i)|}{4Re(\mu_i)-\tau|\mu_i|(2\beta-\tau)}, \ \ \ \ Im(\mu_i)\neq0.
\end{aligned}
\right.
\end{eqnarray}
\end{lemma}

\begin{proof}
Under the conditions \textbf{C1} and \textbf{C2}, we know from Lemma \ref{lemma:12.2} that $\Omega_1$ is an $M$-matrix. This implies that $Re(\mu_i)>0$ for any $i=1,2,\ldots,n$. Let
\begin{eqnarray*}
\begin{aligned}
&Re(c)=\frac{2(\beta-\tau)}{\tau}, Im(c)=0, \\
&Re(d)=\frac{4Re(\mu_i)+\tau^2|\mu_i|-2\beta\tau|\mu_i|}{\tau^2|\mu_i|}, Im(d)=\frac{-4Im(\mu_i)}{\tau^2|\mu_i|}.
\end{aligned}
\end{eqnarray*}
By Lemma \ref{lemma:12.3}, $g(z,\mu_{i})$ is Hurwitz stable if and only if
\begin{equation}\label{sys:12.4}
\frac{2(\beta-\tau)}{\tau}>0, \frac{4(\beta-\tau)^2}{\tau^2}\frac{4Re(\mu_i)+\tau^2|\mu_i|-2\beta\tau|\mu_i|}{\tau^2|\mu_i|}>\frac{16[Im(\mu_i)]^2}{\tau^4|\mu_i|^2}.
\end{equation}
If $Im(\mu_i)=0$, we can deduce from (\ref{sys:12.4}) that $g(z,\mu_{i})$ is Hurwitz stable if and only if
\begin{equation}\label{sys:12.5}
\beta>\tau, \ \ 4Re(\mu_i)+\tau^2|\mu_i|-2\beta\tau|\mu_i|>0.
\end{equation}
Now consider the case of $Im(\mu_i)\neq0$. To ensure the establishment of (\ref{sys:12.4}), it is need to satisfy
\begin{equation}\label{sys:12.6}
\beta>\tau, \ \ (\beta-\tau)^2>\frac{4|Im(\mu_i)|}{4Re(\mu_i)-\tau|\mu_i|(2\beta-\tau)}.
\end{equation}
This completes the proof.
\end{proof}

Next, we state the main result of bipartite bounded tracking under noise disturbance.

\begin{theorem}\label{theorem:12.5}
Suppose that the gain parameter $\beta$ satisfies the conditions in (\ref{sys:12.3}). The bipartite bounded tracking for systems (\ref{sys:5.1}) and (\ref{sys:5.2}) with distributed protocol (\ref{sys:12.1}) can be implemented if and only if the topology graph meets the conditions \textbf{C1} and \textbf{C2}.
\end{theorem}

\begin{proof}
\emph{\underline{Sufficiency}}: For error system (\ref{sys:12.2}), we have
\begin{eqnarray}\label{sys:12.7}
\begin{aligned}
\|e(k)\|_{\infty}&\leq\|\Omega^k_1\|_{\infty}\|e(0)\|_{\infty}
+\left\|\sum_{i=0}^{k-1}\Omega^{k-(i+1)}_1\Omega_2\right\|_{\infty}\|\omega(k)\|_{\infty}\\
&\leq\|\Omega^k_1\|_{\infty}\|e(0)\|_{\infty}
+\bar{\omega}\tau\left\|\sum_{i=0}^{k-1}\Omega^{k-(i+1)}_1\right\|_{\infty}.
\end{aligned}
\end{eqnarray}
According to a fact that there exists an orthogonal matrix $J_1$ such that $\Omega_1=J^{-1}_1KJ_1$, where $K={\rm diag}\{\lambda_1,\lambda_2,\ldots,\lambda_{2n}\}$ and $\lambda_i$, $i=1,2,\ldots,2n$ are the eigenvalues of matrix $\Omega_1$. Then we have
\begin{eqnarray}\label{sys:12.8}
\begin{aligned}
\bar{\omega}\tau\left\|\sum_{i=0}^{k-1}\Omega^{k-(i+1)}_1\right\|_{\infty}\leq\bar{\omega}\tau\|J^{-1}_1\|_{\infty}\|J_1\|_{\infty}
\left\|\sum_{i=0}^{k-1}|K|^{k-(i+1)}_1\right\|_{\infty}.
\end{aligned}
\end{eqnarray}
Under the conditions in (\ref{sys:12.3}), we know from Lemma \ref{lemma:12.4} that all eigenvalues of $\Omega_1$ are within the unit circle, which means that $\rho=\max\{|\lambda_1|,|\lambda_2|,\ldots,|\lambda_{2n}|\}<1$. Then we have $\lim_{k\rightarrow\infty}\|\Omega^k_1\|_{\infty}=0$. Furthermore,
\begin{eqnarray}\label{sys:12.9}
\begin{aligned}
\lim_{k\rightarrow\infty}\|e(k)\|_{\infty}&\leq\bar{\omega}\tau\|J^{-1}_1\|_{\infty}\|J_1\|_{\infty}
\left\|(I_n-|K|)^{-1}\right\|_{\infty}\\
&=\bar{\omega}\tau\|J^{-1}_1\|_{\infty}\|J_1\|_{\infty}(1-\rho)^{-1}.
\end{aligned}
\end{eqnarray}
According to Definition \ref{definition:12.1}, it is known that the bipartite bounded tracking is realized.

\emph{\underline{Necessity}}: The proof of the necessity is similar to the proof of Theorem~\ref{theorem:3.3}, so it is omitted here.
\end{proof}

\begin{figure}[t]
  \centering
    \subfigure[position trajectories of the agents]{
    \includegraphics[width=2.4in]{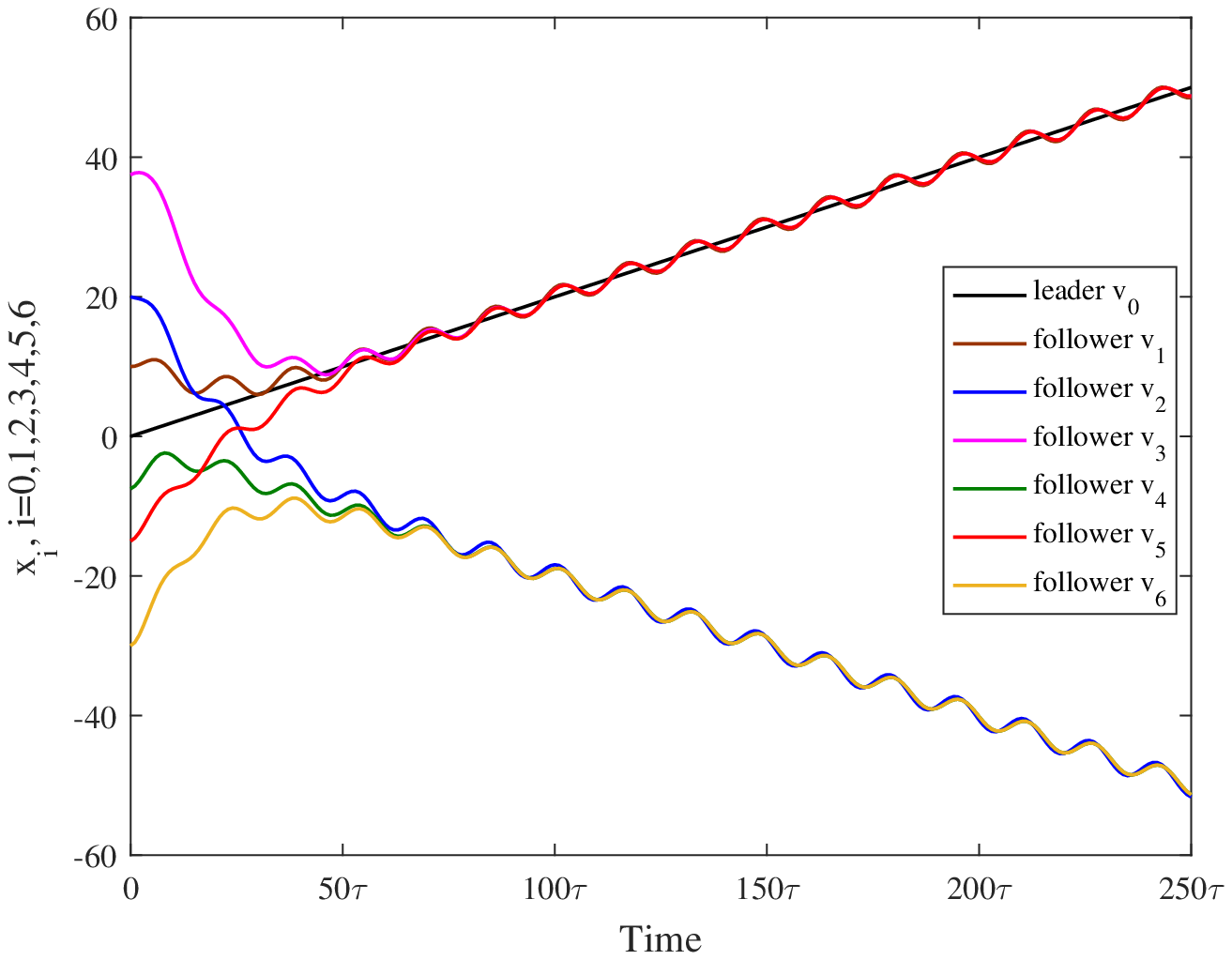}\label{fig10a}}
      \subfigure[velocity trajectories of the agents]{
    \includegraphics[width=2.4in]{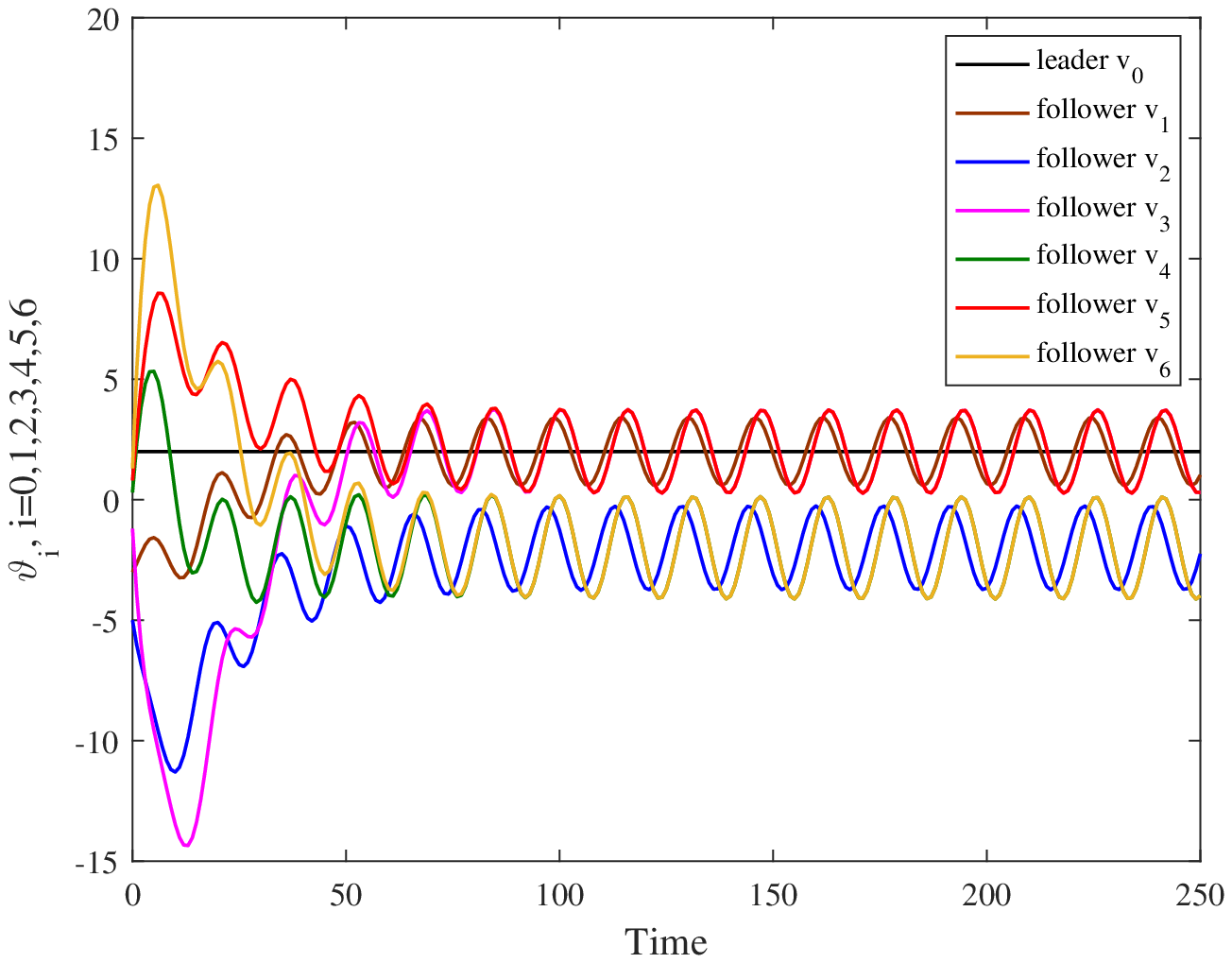}\label{fig10b}}
  \caption{State trajectories of all agents with noise disturbance in Example~\ref{example:12.7}. }\label{fig10}
\end{figure}

\begin{remark}\label{remark:12.6}
In Theorem (\ref{theorem:11.3}) and Theorem (\ref{theorem:12.5}), we consider the effects of matrix disturbance and external noise disturbance on system dynamics, respectively. In the case of matrix disturbance, we can guarantee the realization of bipartite tracking by giving the appropriate parameter selection range. For the case of external noise disturbance, we find another interesting dynamic phenomenon of bipartite bounded tracking for the first time.
\end{remark}

\begin{example}\label{example:12.7}
Consider the phenomenon of bipartite bounded tracking of second-order MASs with noise disturbance over the signed digraph $\tilde{\mathscr{G}}_2$ in Fig.~\ref{fig4}. To satisfy condition (\ref{sys:12.3}), we select $\tau=0.1$ and $\beta=2$. The noise disturbance satisfies $\omega_i(k)=0.5*\sin(0.4k)$. Fig.~\ref{fig10a} and Fig.~\ref{fig10b} show the position and velocity trajectories of the agents, from which it can be see that the bipartite bounded tracking is realized.
\end{example}

\section{Bipartite bounded tracking of MASs with a leader of unmeasurable velocity and acceleration}\label{section:13}

In this section, we consider the bipartite bounded tracking phenomenon for MASs when the followers are unable to accurately measure the leader's velocity and acceleration.

In the discrete-time setting, the leader in the system has the following dynamics:
\begin{eqnarray}\label{sys:13.1}
\begin{aligned}
&x_{0}(k+1)=x_{0}(k)+\tau\vartheta_0(k)+\frac{\tau^2}{2}a(k),\\
&\vartheta_{0}(k+1)=\vartheta_{0}(k)+\tau a(k),
\end{aligned}
\end{eqnarray}
where $a(k)=a_0(k)+\zeta(k)$ is the time-varying acceleration of the leader. In our problem formulation, the leader's velocity and acceleration are unknown to all agents, and that the leader's position is only available for a subset of agents. However, each agent has an approximate estimate $a_0(k)$ for the leader's acceleration. The error between the approximation $a_0(k)$ and the actual value $a(k)$ is bounded by a given boundary $\zeta$, that is, $\|a(k)-a_{0}(k)\|_{\infty}=\|\zeta(k)\|_{\infty}\leq\bar{\zeta}$.

Since the leader's velocity $\vartheta_0(k)$ and acceleration $a(k)$ cannot be measured during information interaction, its value cannot be used to control the design. Instead, in order to ensure the implementation of bipartite bounded tracking, it is necessary to estimate $\vartheta_0(k)$ and $a(k)$ through the location information of neighbors in the evolution process. The discrete-time dynamics of each follower is given by
\begin{eqnarray}\label{sys:13.2}
\begin{aligned}
&x_{i}(k+1)=x_{i}(k)+\tau u_{i}(k), \\
&\vartheta_{i}(k+1)=\vartheta_{i}(k)+\tau \tilde{a}_{i}(k), \ \ i=1,2,\ldots,n,
\end{aligned}
\end{eqnarray}
where $u_{i}(k)$ and $\tilde{a}_{i}(k)$ are the follower $v_i$'s estimates about the leader's velocity $\vartheta_0(k)$ and acceleration $a(k)$, respectively. We design the following local control scheme consists of two parts:
\begin{eqnarray}\label{sys:13.3}
\begin{aligned}
u_{i}(k)=\vartheta_i(k)+\kappa_i\frac{\tau}{2}\tilde{a}_{i}(k)&+\phi_1\sum\limits_{v_j\in\mathscr{N}_i}|a_{ij}|
\big[\sgn(a_{ij})x_{j}(k)-x_{i}(k)\big]\\
&+\phi_1|a_{i0}|\big[\sgn(a_{i0})x_{0}(k)-x_{i}(k)\big],
\end{aligned}
\end{eqnarray}
\begin{eqnarray}\label{sys:13.4}
\begin{aligned}
\tilde{a}_{i}(k)=\kappa_i a_0(k)&+\phi_2\phi_1\sum\limits_{v_j\in\mathscr{N}_i}|a_{ij}|\big[\sgn(a_{ij})x_{j}(k)-x_{i}(k)\big]\\
&+\phi_2\phi_1|a_{i0}|\big[\sgn(a_{i0})x_{0}(k)-x_{i}(k)\big],
\end{aligned}
\end{eqnarray}
where $\phi_1>0, \phi_2>0$ are gain parameters, and $\kappa_i$ satisfies: $\kappa_i=1$ if $v_i\in\mathscr{V}_1$, and $\kappa_i=-1$ if $v_i\in\mathscr{V}_2$.

By substituting (\ref{sys:13.3}) and (\ref{sys:13.4}) into systems (\ref{sys:13.1}) and (\ref{sys:13.2}), the following error system can be derived
\begin{eqnarray}\label{sys:13.5}
\begin{aligned}
e(k+1)=[\Upsilon_1\otimes I_p] e(k)+[\Upsilon_2\otimes I_p]\zeta(k),
\end{aligned}
\end{eqnarray}
where
\begin{eqnarray*}
\begin{aligned}
\Upsilon_1=\left(
           \begin{array}{cc}
             I_n-\big(\phi_1\tau+\frac{1}{2}\phi_1\phi_2\tau^2\big)\mathscr{H} & \tau I_n \\
            -\phi_1\phi_2\tau\mathscr{H} & I_n\\
           \end{array}
         \right),\Upsilon_2=\left(
           \begin{array}{c}
             \frac{\tau^2}{2}\mathbf{1}_n \\
            \tau\mathbf{1}_n\\
           \end{array}
         \right)
\end{aligned}
\end{eqnarray*}
with $\mathbf{1}_n\in\mathbb{R}^{n\times1}$ is a column vector in which all elements are 1.

The characteristic polynomial of $\Upsilon_1$ is written as
\begin{eqnarray*}
\begin{aligned}
\text{det}(\lambda I_{2n}-\Upsilon_1)=\lambda^2+\frac{2\phi_1\tau\mu_i+\phi_2\phi_1\tau^2\mu_i-4}{2}\lambda+\frac{2+\phi_2\phi_1\tau^2\mu_i-2\phi_1\tau\mu_1}{2}.
\end{aligned}
\end{eqnarray*}
Therefore, the eigenvalues of $\Upsilon_1$ satisfy
\begin{eqnarray*}
\lambda^2+\frac{2\phi_1\tau\mu_i+\phi_2\phi_1\tau^2\mu_i-4}{2}\lambda+\frac{2+\phi_2\phi_1\tau^2\mu_i-2\phi_1\tau\mu_1}{2}=0.
\end{eqnarray*}
For convenience, denote $f^*(\lambda,\mu_{i})=\lambda^2+\frac{2\phi_1\tau\mu_i+\phi_2\phi_1\tau^2\mu_i-4}{2}\lambda+\frac{2+\phi_2\phi_1\tau^2\mu_i-2\phi_1\tau\mu_1}{2}$. Applying the bilinear transformation $\lambda=\frac{z+1}{z-1}$ to $f^*(\lambda,\mu_{i})$, we can derive that $f^*(\lambda,\mu_{i})$ is Schur stable if and only if
\begin{eqnarray*}
\begin{aligned}
g^*(z,\mu_{i})=z^2+\frac{2-\phi_2\tau}{\phi_2\tau}z+\frac{4-2\phi_1\tau\mu_i}{\phi_2\phi_1\tau^2\mu_i}
\end{aligned}
\end{eqnarray*}
is Hurwitz stable. We hence only need to prove that $g^*(z,\mu_{i})$ is Hurwitz stable by choosing the appropriate parameter $\phi$.

\begin{lemma}\label{lemma:12.1}
Suppose that the topology graph satisfies the conditions \textbf{C1} and \textbf{C2}. All eigenvalues of $\Upsilon_1$ are within the unit circle if the parameters $\phi_1, \phi_2$ satisfy
\begin{eqnarray}\label{sys:13.6}
\left\{
\begin{aligned}
&\phi_1<\frac{2Re(\mu_i)}{\tau|\mu_i|^2}, \ \phi_2<\frac{2}{\tau}, \ \ \ \ \ \ \ \ \ \ \ \ \ \ \ \ \ \ \ \ \ \ \ \ \ \ \ \ \ \ \ \ \ \ \ \ \ \ \ \ Im(\mu_i)=0,\\
&\phi_1<\frac{2Re(\mu_i)}{|\mu_i|^2}, \ \phi_2<\frac{2}{\tau}-\sqrt{\frac{8\phi_2}{\tau^2\phi_1(2Re(\mu_i)-\phi_1\tau|\mu_i|^2)}}, \ \ Im(\mu_i)\neq0.
\end{aligned}
\right.
\end{eqnarray}
\end{lemma}

\begin{proof}
It is known from Lemma \ref{lemma:12.2} that $\Omega_1$ is an $M$-matrix under the conditions \textbf{C1} and \textbf{C2}. This implies that $Re(\mu_i)>0$ for any $i=1,2,\ldots,n$. Let
\begin{eqnarray*}
\begin{aligned}
&Re(c)=\frac{2-\phi_2\tau}{\phi_2\tau}, Im(c)=0, \\
&Re(d)=\frac{4Re(\mu_i)-2\phi_1\tau|\mu_i|^2}{\phi_2\phi_1\tau^2|\mu_i|^2}, Im(d)=\frac{-4Im(\mu_i)}{\phi_2\phi_1\tau^2|\mu_i|^2}.
\end{aligned}
\end{eqnarray*}
By Lemma \ref{lemma:12.3}, $g(z,\mu_{i})$ is Hurwitz stable if and only if
\begin{equation}\label{sys:13.7}
\frac{2-\phi_2\tau}{\phi_2\tau}>0, \frac{(2-\phi_2\tau)^2}{\phi^2_2\tau^2}\frac{4Re(\mu_i)-2\phi_1\tau|\mu_i|^2}{\phi_2\phi_1\tau^2|\mu_i|^2}>\frac{16[Im(\mu_i)]^2}{\phi_2^2\phi_1^2\tau^4|\mu_i|^4}.
\end{equation}
If $Im(\mu_i)=0$, we can deduce from (\ref{sys:13.7}) that $g(z,\mu_{i})$ is Hurwitz stable if and only if
\begin{equation}\label{sys:13.8}
\phi_1<\frac{2Re(\mu_i)}{\tau|\mu_i|^2}, \ \phi_2<\frac{2}{\tau}.
\end{equation}
Consider the case of $Im(\mu_i)\neq0$. Since $[Im(\mu_i)]^2<|\mu_i|^2$, the inequalities in (\ref{sys:13.7}) hold if the conditions are satisfied
\begin{equation}\label{sys:13.9}
\phi_2<\frac{2}{\tau}, \ \ (2-\phi_2\tau)^2\big(2Re(\mu_i)-\phi_1\tau|\mu_i|^2\big)>\frac{8\phi_2}{\phi_1}.
\end{equation}
Clearly, $(2-\phi_2\tau)^2\big(2Re(\mu_i)-\phi_1\tau|\mu_i|^2\big)>\frac{8\phi_2}{\phi_1}$ is satisfied only if $2Re(\mu_i)-\phi_1\tau|\mu_i|^2>0$, that is, $\phi_1<\frac{2Re(\mu_i)}{|\mu_i|^2}$. Therefore, the inequalities in (\ref{sys:13.9}) hold if the following conditions are satisfied
\begin{equation}\label{sys:13.10}
\phi_1<\frac{2Re(\mu_i)}{|\mu_i|^2}, \ \phi_2<\frac{2}{\tau}-\sqrt{\frac{8\phi_2}{\tau^2\phi_1(2Re(\mu_i)-\phi_1\tau|\mu_i|^2)}}.
\end{equation}
This completes the proof.
\end{proof}

Now we represent the necessary and sufficient condition for bipartite bounded tracking of MASs with a leader whose velocity and acceleration are unmeasurable.

\begin{theorem}\label{theorem:13.2}
Suppose that the gain parameters $\phi_1, \phi_2$ satisfy the conditions in (\ref{sys:13.6}). Consider systems (\ref{sys:13.1}) and (\ref{sys:13.2}). The distributed protocols (\ref{sys:13.3}) and (\ref{sys:13.4}) solve the bipartite bounded tracking problem if and only if the topology graph meets the conditions \textbf{C1} and \textbf{C2}.
\end{theorem}

\begin{proof}
This proof is similar to the proof of Theorem \ref{theorem:12.5}, so it is omitted here.
\end{proof}

\begin{figure}[t]
  \centering
    \subfigure[position trajectories of the agents]{
    \includegraphics[width=2.4in]{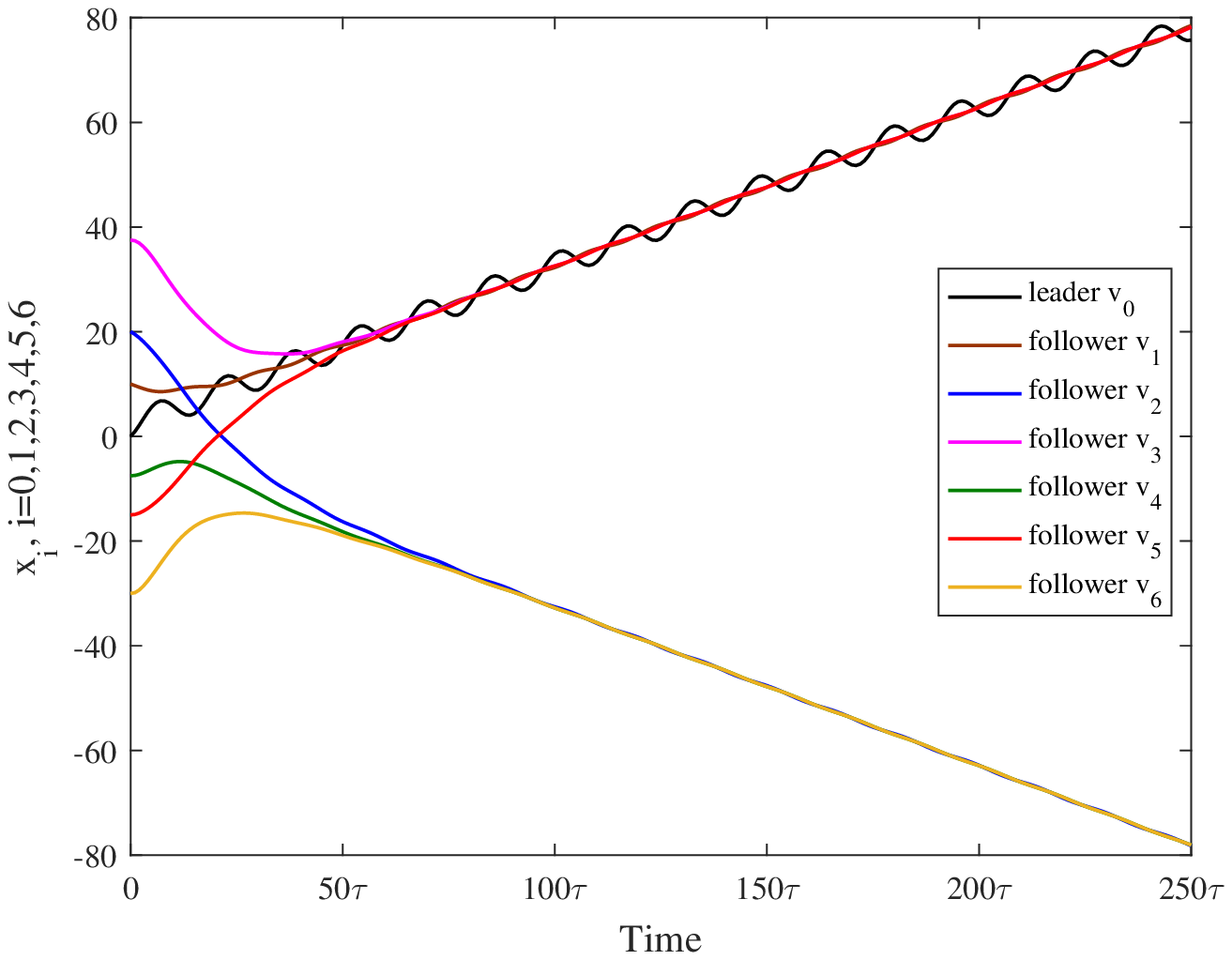}}
      \subfigure[velocity trajectories of the agents]{
    \includegraphics[width=2.4in]{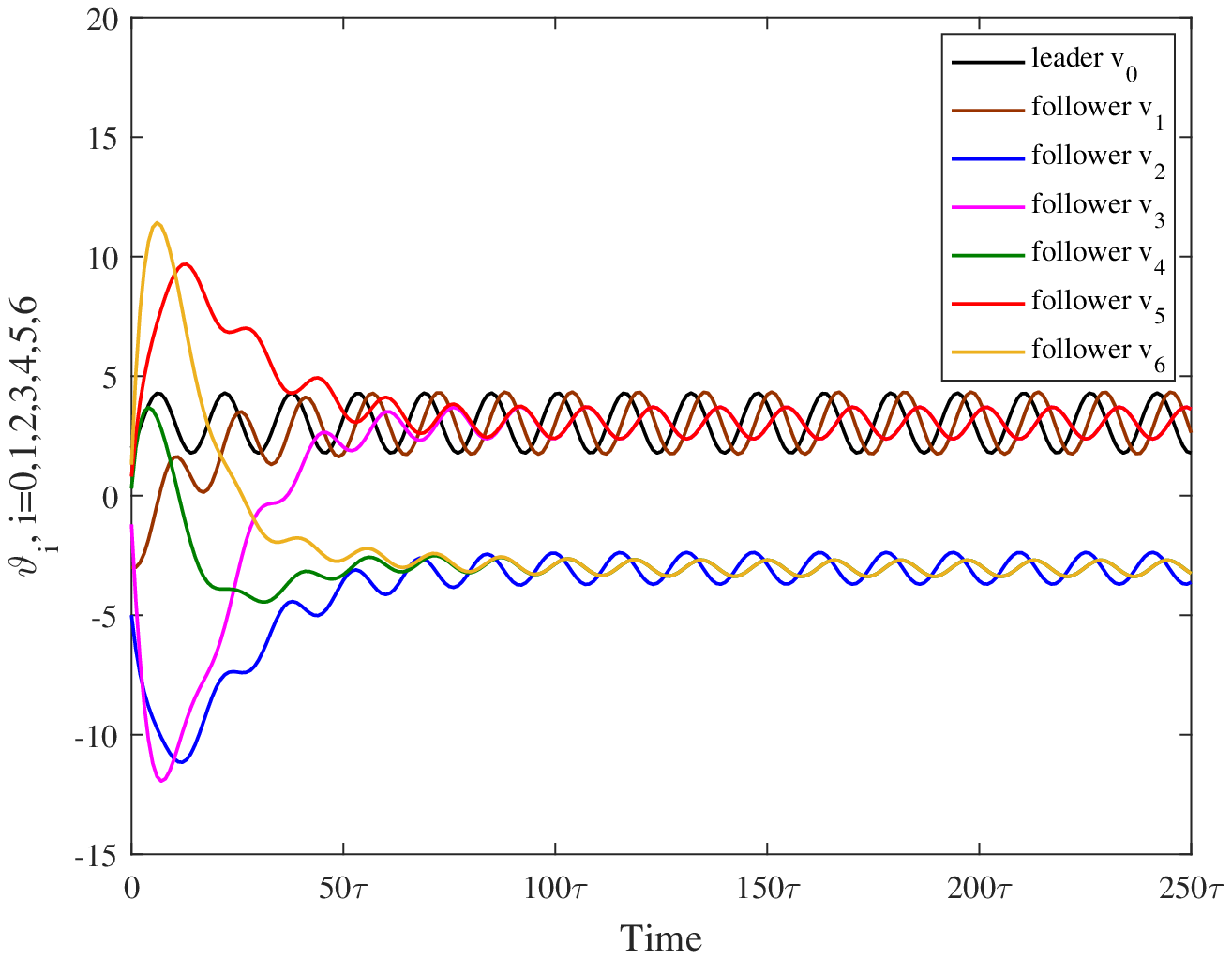}}
  \caption{State trajectories of all agents in Example~\ref{example:13.3}. }\label{fig11}
\end{figure}

\begin{example}\label{example:13.3}
Consider second-order MASs with a leader of unmeasurable velocity and acceleration. The agents interact with each other through the signed digraph $\tilde{\mathscr{G}}_2$ in Fig.~\ref{fig4}. To satisfy condition (\ref{sys:13.6}), we select $\tau=0.1$, $\phi_1=3$ and $\phi_2=1$. It is assumed that the leader moves with a time-varying acceleration $a(k)=1+0.5\sin(0.4k)$, and each follower has an approximate estimate $a_0(k)=1$ for the leader's acceleration. Finally, we can observe from the position and velocity trajectories in Fig.~\ref{fig11} that the bipartite bounded tracking is achieved.
\end{example}

\section{Conclusion}\label{section:14}

In this paper, the concept of super-stochastic matrix has been put forward for the first time, and the conclusions about the product convergence of ISubSM and ISupSM have been established by using systematic algebraic-graphical methods. Based on the product convergence of ISubSM and ISupSM, the bipartite tracking problems of first-order MASs, second-order MASs and general linear MASs under the asynchronous interactions have been examined, respectively. Moreover, the product convergence of ISubSM and ISupSM has also been applied to the analysis of bipartite tracking dynamics under different actual scenario settings, including time delays, switching topologies, random networks, lossy links, matrix disturbance, external noise disturbance, and a leader of unmeasurable velocity and acceleration. Finally, the correctness of the theoretical results has been verified by numerical examples. It is worth noting that the product properties of ISubSM and ISupSM can not only solve the bipartite tracking issues considered in this paper, but also have certain reference value for analyzing other coordination control issues of discrete-time MASs, such as containment control, formation control, flocking control, etc., which is also what we are trying to do in the future.

\section*{Acknowledgments}

The authors gratefully acknowledge the suggestions and comments by the associate editor and anonymous reviewers. This research was supported in part by the National Science Foundation of China (U1830207, 61772003, 61903066), the National Key R\&D Program of China (2017YFC1501005, 2018YFC150 5203), the China Postdoctoral Science Foundation (2017M612944, 2018T110962), the Special Postdoctoral Foundation of Sichuan Province, the Fundamental Research Funds for the Central Universities (ZYGX2018J087), the Australian Research Council (DP120104986) , and the NSW Cyber Security Network in Australia (P00025091).

\bibliographystyle{siamplain}

\end{document}